\pdfoutput=1
\def\BuildMode{arxiv}

\documentclass[conference]{IEEEtran}

\usepackage{main}
\usepackage{amsfonts}
\usepackage{makecell}

\newcommand{\PaperTitle}{Improving Network Clock Synchronization by Marking Congestion}
\title{\PaperTitle}
\author{
\IEEEauthorblockN{Yash Deshpande\IEEEauthorrefmark{1},
Quirin Vogel\IEEEauthorrefmark{2},
Laura Becker\IEEEauthorrefmark{1},
Kaan Aykurt\IEEEauthorrefmark{1},
Wolfgang Kellerer\IEEEauthorrefmark{1}}

\IEEEauthorblockA{\IEEEauthorrefmark{1}
Chair of Communication Networks, Technical University of Munich, Germany.}

\IEEEauthorblockA{\IEEEauthorrefmark{2}
Department of Statistics, University of Klagenfurt, Austria.}
}

\newcommand{\YD}[1]{\textcolor{blue}{#1}}
\newcommand{\QV}[1]{\textcolor{red}{#1}}

\providecommand{\BuildMode}{full}
\usepackage{atbegshi}
\makeatletter
\edef\BM@mode{\BuildMode}
\def\BM@full{full}\def\BM@journal{journal}\def\BM@appendix{appendix}\def\BM@arxiv{arxiv}
\ifx\BM@mode\BM@full\else\ifx\BM@mode\BM@journal\else
\ifx\BM@mode\BM@appendix\else\ifx\BM@mode\BM@arxiv\else
  \errmessage{Unknown \string\BuildMode\space `\BM@mode' (use full, journal, appendix or arxiv)}
\fi\fi\fi\fi

\newif\ifBM@discard
\AtBeginShipout{\ifBM@discard\AtBeginShipoutDiscard\fi}
\let\BM@protected@write\protected@write
\long\def\BM@discard@write#1#2#3{%
  \ifx#1\@auxout
    \begingroup #2\let\protect\@unexpandable@protect
      \edef\reserved@a{\immediate\write#1{#3}}\reserved@a
    \endgroup
  \fi
  \if@nobreak\ifvmode\nobreak\fi\fi}
\def\BM@startdiscard{\clearpage\global\BM@discardtrue
  \global\let\protected@write\BM@discard@write}
\def\BM@stopdiscard{\clearpage\global\BM@discardfalse
  \global\let\protected@write\BM@protected@write}

\def\BuildModeBegin{\ifx\BM@mode\BM@appendix\BM@startdiscard\fi}
\def\BuildModeAppendix{%
  \ifx\BM@mode\BM@journal\BM@startdiscard\fi
  \ifx\BM@mode\BM@appendix
    \BM@stopdiscard
    \setcounter{page}{1}%
    \twocolumn[{\centering
      {\LARGE Supplementary Material}\\[0.6em]
      {\Large\PaperTitle}\\[1.5em]}]%
  \fi}

\ifx\BM@mode\BM@arxiv
  \renewcommand{\yd}[1]{#1}
  \renewcommand{\YD}[1]{#1}
  \renewcommand{\QV}[1]{#1}
  \renewcommand{\ROne}[1]{#1}
  \renewcommand{\RTwo}[1]{#1}
  \renewcommand{\RThree}[1]{#1}
\fi
\makeatother

\begin{document}
\BuildModeBegin

\begin{acronym}
    \acro{AP}{access point}
    \acro{PTP}{Precision Time Protocol}
    \acro{OWD}{One-way delay}
    \acro{RTT}{Round-trip time}
    \acro{SMS}{Smart manufacturing systems}
    \acro{CPS}{Cyber-physical systems}
    \acro{IIoT}{Industrial Internet-of-Things}
    \acro{TUB}{Time Uncertainty Bound}
    \acro{UTC}{Coordinated Universal Time}
    \acro{GPS}{Global Positioning System}
    \acro{OWD}{One-way delay}
    \acro{ppm}{parts per million}
    \acro{RTT}{Round trip time}
    \acro{NIC}{Network interface controllers}
    \acro{DES}{Discrete event simulator}
    \acro{MAC}{media access control}
    \acro{TSN}{Time Sensitive Networking}
    \acro{NTP}{Network Time Protocol}
    \acro{COTS}{commercial off-the-shelf}
    \acro{PDV}{packet delay variation}
    \acro{LAN}{local area networks}
    \acro{SNMP}{Simple Network Management Protocol}
    \acro{SDN}{Software defined networking}
    \acro{NIC}{Network Interface Card}
    \acro{NC}{Network Calculus}
    \acro{MDR}{Maximum Drift Rate}
    \acro{TDD}{Time Division Duplex}
    \acro{TC}{Transparent Clock}
    \acro{BC}{Boundary Clock}
    \acro{ECN}{Explicit Congestion Notification}
    \acro{PI}{Proportional Integral}
    \acro{SVM}{Support vector machine}
    \acro{ML}{Machine Learning}
    \acro{CMC}{Congestion Marking Clock}
    \acro{SCMC}{Serial CMC}
    \acro{CCMC}{Concurrent CMC}
    \acro{RED}{Random Early Detection}
    \acro{FR}{Forward-Reverse}
    \acro{MTU}{Maximum Transmission Unit}
    \acro{PDF}{Probability Density Function}
    \acro{CCDF}{Complementary Cumulative Density Function}
    \acro{RMS}{Root Mean square}
    \acro{MSE}{Mean Square Error}
    \acro{LM}{Long-Medium}
    \acro{SF}{Short-Fast}
    \acro{SS}{Short-Slow}
    \acro{SM}{Short-Medium}
    \acro{LS}{Legacy Switch}
    \acro{GNN}{Graph Neural Networks}
    \acro{IR}{Improvement Region}
    \acro{AQM}{Active Queue Management}
\end{acronym} 
\newcommand{\Clock}[1]{\ensuremath{C_{#1}}}
\newcommand{\ClockEst}[1]{\ensuremath{\hat{C}_{#1}}}

\newcommand{\Offset}{\theta}
\newcommand{\OffsetEst}{\hat{\theta}}

\newcommand{\Delay}[2]{\ensuremath{\delta_{\text{#1}}^{\text{#2}}}}

\newcommand{\DelayHop}[3]{\ensuremath{\delta_{\text{#1},#3}^{\text{#2}}}}

\newcommand{\DelayEst}[2]{\ensuremath{\hat{\delta}_{\text{#1}}^{\text{#2}}}}

\newcommand{\DelayEstHop}[3]{\ensuremath{\hat{\delta}_{\text{#1},#3}^{\text{#2}}}}

\newcommand{\abs}[1]{\left|#1\right|}
\newcommand{\Skew}[2]{\ensuremath{S_{#1#2}}}

\newcommand{\Drift}[2]{\ensuremath{D_{#1#2}}}

\newcommand{\Error}[1]{\ensuremath{E(#1)}}

\newcommand{\tub}[2]{\ensuremath{\epsilon_{#1#2}}}

\newcommand{\gk}[1]{\left\{#1\right\}}
\newcommand\mycom[2]{\genfrac{}{}{0pt}{}{#1}{#2}}
\newcommand{\ek}[1]{\left[#1\right]}
\newcommand{\rk}[1]{\left(#1\right)}
\newcommand{\hk}[1]{^{(#1)}}
 \newcommand{\ex}{{\rm e}} 
 \renewcommand{\d}{{\rm d}} 
\renewcommand{\P}{\mathbb{P}}
\newcommand{\E}{\mathbb{E}}
\newcommand{\1}{\boldsymbol{1}}
\newcommand{\e}{\varepsilon}
\newcommand{\var}{\mathrm{Var}}
\newcommand{\cov}{\mathrm{Cov}} 
\maketitle
\pagestyle{plain}
\begin{abstract}
Achieving consistent time across devices in distributed systems often involves exchanging timestamped messages over a network. Precise time synchronization is crucial for applications such as cellular networks, industrial automation, and transactional databases. However, delay variance in synchronization packets, often caused by congestion from competing traffic, degrades synchronization accuracy. 
The clock error grows directly with the variance and the asymmetry of forward and reverse delays. 
Detecting whether a packet experienced congestion can help improve synchronization through filtering and statistical methods.

We propose an in-network congestion indication and filtering mechanism for synchronization messages used in protocols such as the Network Time Protocol (NTP) and Precision Time Protocol (PTP). Network devices mark packets that experienced queuing, allowing clocks to correct errors caused by varying delays. Our approach requires only simple changes at switches and endpoints and is
backward compatible: it uses the ECN field in the IP header or currently
unused fields in PTP and NTP headers, which devices without support for our
method simply ignore.

We implement our method on a Tofino P4 target, where over a single hop it reduces the variance of the measured delay by up to 95\%, and the variance at the output of min-RTT and median-delay filters by up to 99\%. 
In packet-level simulations, it reduces the root-mean-square (RMS) error of the offset estimate by up to 70\% over three hops. 
In closed-loop simulations with drifting clocks over up to seven hops, it reduces the RMS time error of the synchronized clock by about 85\%, with the benefit increasing in absolute terms as the number of hops increases.

\end{abstract}

\acresetall
\section{Introduction}
\label{sec:introduction}

Precise time synchronization is a cornerstone of modern distributed systems, enabling coordinated operations across various applications. 
In cellular networks~\cite{ITUstd}, accurate timing ensures seamless handovers and efficient spectrum utilization. 
Industrial automation systems rely on synchronized clocks for coordinated control and monitoring~\cite{Gore2020}, while transactional databases~\cite{spanner} require precise timestamps to maintain data consistency and integrity.

Two primary protocols facilitate network-based time synchronization: the \ac{NTP}~\cite{mills1992rfc, rfc5905} and the Precision Time Protocol (PTP)~\cite{ptp, ieee1588}. 
NTP is widely used across the internet, providing millisecond-level accuracy suitable for general-purpose applications. 
PTP, standardized as IEEE 1588, offers sub-microsecond precision, making it ideal for time-sensitive environments such as power grids and manufacturing systems.
PTP leverages hardware support from switches and end-hosts and is intended to be deployed in \ac{LAN}s where the operator controls the network and knows the topology. 
PTP can operate at both the Ethernet and IP layers, while NTP operates only at the IP layer. 
Both PTP and NTP work on the principle of sending time-stamped packets over the network~\cite{Banerjee2023}.

However, the accuracy of these protocols can be compromised by variable network delays, particularly those caused by queuing congestion in routers and switches~\cite{Banerjee2023}.
Both PTP and NTP assume that packet delays in both directions, from the client to the server and vice versa, are similar and that time-stamped packets take half of the \ac{RTT} to reach their peer.  
When synchronization packets experience \ac{PDV}, the resulting time estimates at the peer end can be significantly skewed, leading to errors in the application's operation.

Various filtering and advanced estimating techniques have been explored to mitigate these issues~\cite{kalmanFiltering, genghuygens, bayesian_optimization}. 
However, relying on filtering to estimate an unknown link delay has fundamental limitations~\cite{LimitationsClockSync}. 
A large proportion of the \ac{PDV} is caused by queuing delays in switches and routers due to congestion from other network traffic~\cite{zilbermanTime}, which can be highly dynamic.
Therefore, hardware implementations of a \ac{TC} or a \ac{BC} for PTP switches have been proposed and standardized. 
A \ac{TC} adds the resident time of the \ac{PTP} packet in the switches to the PTP header.  
A \ac{BC} synchronizes its clock to a source clock and then acts as a source to other downstream devices. 
However, implementing a \ac{TC} or a \ac{BC} in switches and routers is difficult and expensive~\cite{TransparentClockImpl, BoundaryClockImpl}. 
Moreover, \ac{NTP} does not support such intermediary correction mechanisms. 

\ac{ECN}~\cite{ramakrishnan2001addition} is a network feature that allows switches and routers to signal congestion to endpoints without dropping packets, enabling more efficient and reliable congestion control in TCP-like algorithms.
Its success lies in its easy implementation~\cite{ECNAdoption,ECNTuning}: the switch flips a single bit in the IP header to indicate whether the packet experienced congestion.  
Inspired by ECN's success, we explore the utility of similar congestion marking in PTP and NTP protocols to improve synchronization performance. 
Since ECN marking also indicates a congested packet, this information can be used to filter \ac{PDV} in clock synchronization packets.  
Specifically, our novel \textbf{contributions} in this paper are:
\begin{enumerate}
    \item We propose and evaluate a \textbf{novel clock offset estimation method} for networked clock synchronization mechanisms, such as \ac{PTP} and \ac{NTP}, \ROne{in operator-controlled networks}, that uses congestion markings in packet headers.
    The method is \textbf{backward-compatible} with current versions of \ac{NTP} and \ac{PTP} and is suitable for modern methods such as Firefly~\cite{firefly}.
    \item We propose a \textbf{new switch mechanism} called \ac{CMC}, which can mark congestion in the PTP or NTP header with \textbf{more granularity} than ECN and mark packets \textbf{over many hops}.
    \item We \textbf{implement} the CMC and ECN marking on a P4 target (Tofino), and demonstrate that the congestion-marking indicates delay within \ROne{6\% accuracy}. 
    \ROne{Moreover, we show, via an analysis of compiler outputs, that the CMC
    requires substantially fewer pipeline stages, action slots, SRAM, and PHV resources than a TC and no stateful memory. The bit-shift variant of CMC adds no pipeline stages over plain forwarding.} 
    \item We \textbf{analytically} determine the conditions under which our proposed method improves the average \ac{MSE} and empirically \textbf{show} that these \textbf{conditions for improvement are slack}, allowing users to select from a \textbf{large parameter} set without degrading the system performance.  
    \item We propose a \textbf{Markovian Model and an Algorithm} to obtain the \textbf{theoretical MSE} after applying our method. The theoretical MSE can be used to determine the optimal system parameters.  
    \item \ROne{We show that our method also improves the performance of RTT-based low-pass filters used in clock synchronization: over a single hop on hardware, marking reduces the variance at the output of min-RTT and median-delay filters by up to 99\% (Fig.~12a). Even without filtering, corrected measurements outperform an uncorrected filter over twelve samples (Fig.~12b), so shorter filters suffice, and the synchronization reacts faster to changes in network conditions.}
    \item \ROne{Hardware measurements and OMNeT++ simulations quantify the
improvement in three metrics: the delay variance, the RMS offset
estimation error, and the RMS time error of the synchronized clock. Over
a single hop on hardware, marking reduces the delay variance by
18--62\% with one threshold and by up to 95\% with eight (Fig.~8b). In
open-loop simulations, it reduces the RMS offset estimation error by
about 40\% over a single hop (Fig.~11) and by up to 70\% over three hops
(Fig.~14). In closed loop with drifting clocks, it reduces the RMS time
error by about 85\% over up to seven hops (Fig.~15).}
\end{enumerate}

The rest of the paper is organized as follows: Sec. \ref{sec:background} provides the necessary background and the current state of the art in networked clock synchronization and congestion marking. 
Then Sec. \ref{sec:offset_estimation} introduces the main concept of the paper and shows how offset estimation accuracy and \ac{PDV} filtering relate to estimating the queuing delay of a synchronization packet. 
Sec. \ref{sec:methods} provides insights into the implementation of our method and proposes header fields to be added to existing IP, PTP, and NTP packets to implement the marking counter. 
We also implement our switch in P4~\cite{p4} and demonstrate its lower complexity than a \ac{TC}. 
In Sec. \ref{sec:analysis}, we propose the Markovian queuing model and Algorithm \ref{alg1} to obtain expressions for the expected \ac{MSE} of offset estimation error and the improvement achieved by our proposed method. 
Finally, we demonstrate the improvement of our scheme through OMNET++-based packet simulations in Sec.~\ref{sec:evaluation} before concluding the paper. 
\section{Background and Related Work}
\label{sec:background}
\RThree{In this Section, we first describe how NTP and PTP estimate the clock offset and why delay variation degrades
this estimate (Sec.~\ref{subsec:ptp_and_ntp}), followed by a review of existing correction methods,
including filtering and transparent clocks (Sec.~\ref{subsec:correcting_est_errors}). We then introduce
congestion marking (Sec.~\ref{subsec:congestion_marking}) and illustrate its potential for clock
synchronization with a motivating example (Sec.~\ref{subsec:motivating_example}).}

\subsection{Fundamentals of Network Clock Synchronization}
\label{subsec:ptp_and_ntp}

\begin{figure}
  \centering
  \begin{subcaptionbox}{\label{fig:ntp_message}}[0.48\columnwidth]
    {\includegraphics[width=\linewidth]{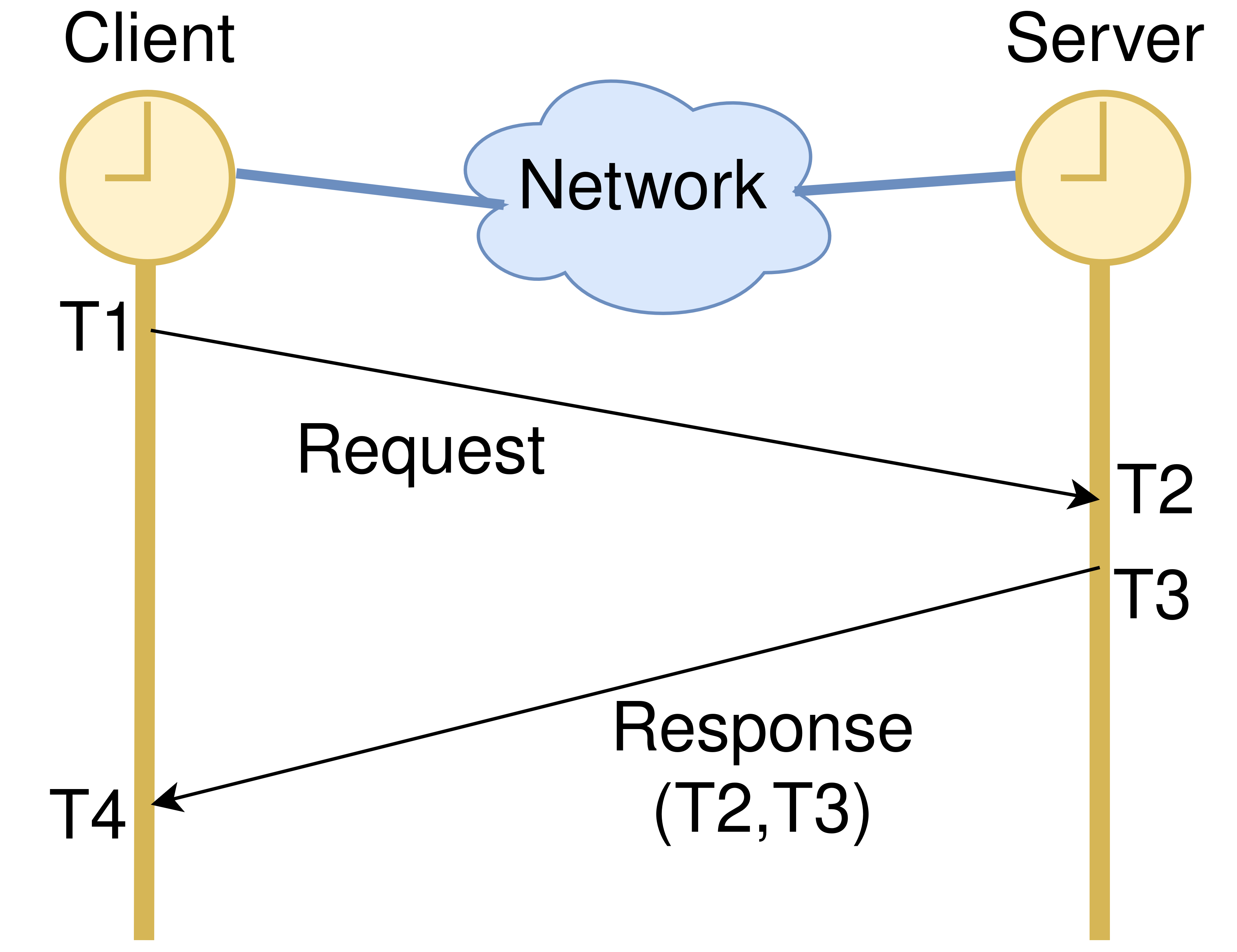}}
  \end{subcaptionbox}
  \begin{subcaptionbox}{\label{fig:ptp_message}}[0.48\columnwidth]
    {\includegraphics[width=\linewidth]{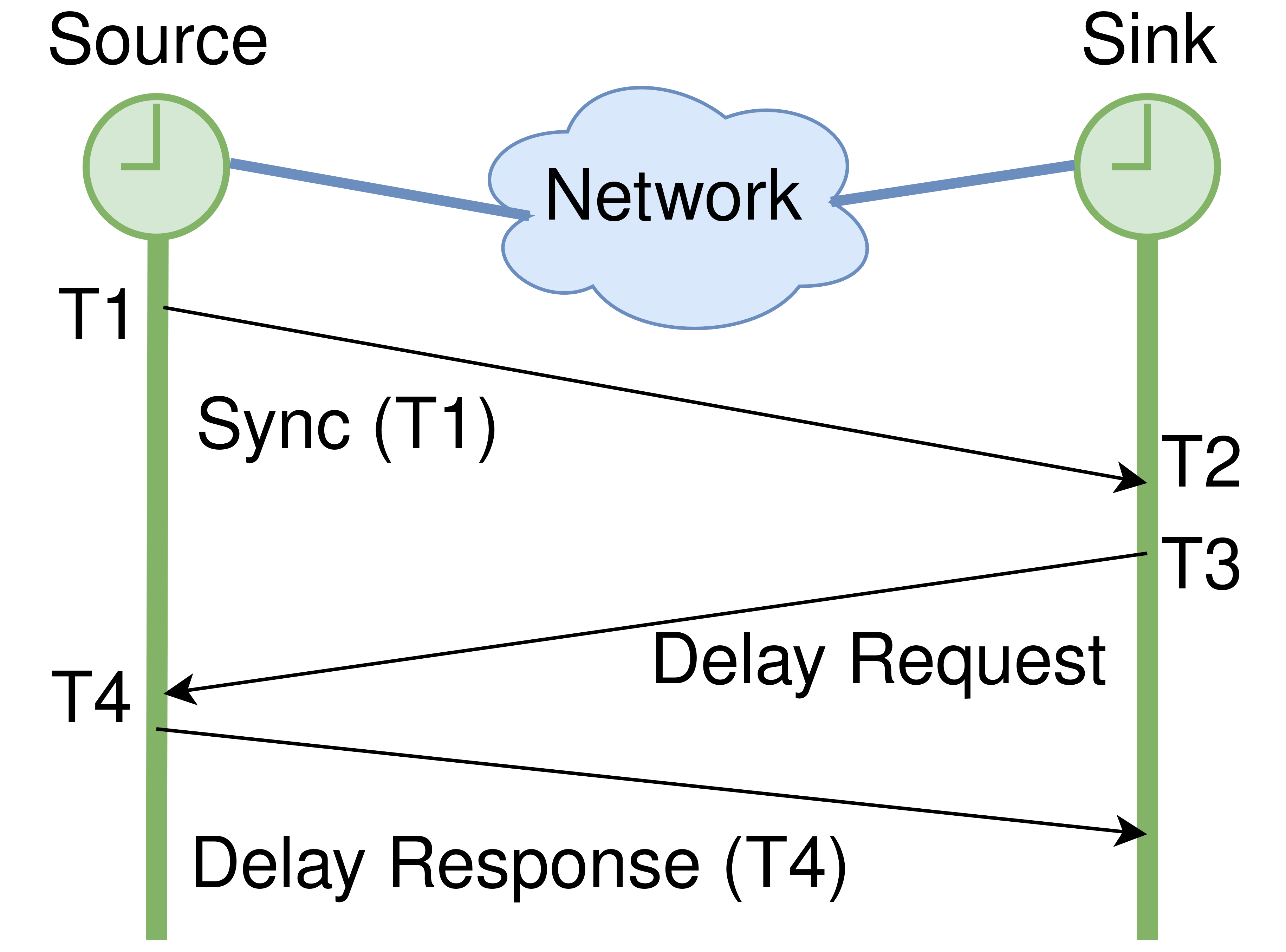}}
  \end{subcaptionbox}
  \caption{\small{Mechanism of (a) NTP: The client periodically initiates the synchronization procedure, exchanging two messages. (b) PTP has a very similar mechanism, but the source (server) periodically initiates the synchronization procedure, and three messages are exchanged. In both protocols, the client uses 4 time values to calculate its offset from the source. }}
  \label{fig:main}
\end{figure}

The fundamental operation of \ac{NTP} and \ac{PTP} is illustrated in Fig. \ref{fig:ntp_message} and \ref{fig:ptp_message}, respectively.
Both use a modified form of Cristian`s algorithm~\cite{cristian1989probabilistic}, where a client estimates its clock offset $\OffsetEst$ from a server. 
In NTP, the client sends a request message and records the time $T1$ at which it was sent. 
The server records the time at which this request message was received as $T2$, then it prepares the response message. 
The server puts the time  when this response is sent back to the client, $T3$, and $T2$, in the response message.
The client records the time this response message is received as $T4$.
Two equations represent the real offset $\Offset$, assuming the devices make no timestamping errors. 
\begin{equation}
    T2 = T1 + \Offset + \Delay{Req}{},
    \label{eq:forward_time}
\end{equation}
and 
\begin{equation}
    T4 = T3 - \Offset + \Delay{Resp}{}.
    \label{eq:reverse_time}
\end{equation}
Here, $\Delay{Req}{}$ and $\Delay{Resp}{}$ are the unknown network delays of the request and response messages, respectively. 
Since the devices' clocks are not accurately synchronized, there is no real way of knowing the exact value of this delay. 
Hence, Cristians' algorithm makes a probabilistic assumption to simplify the two equations with three unknowns: that the forward delay is the same as the reverse delay $\Delay{Req}{} \approx \Delay{Resp}{}$.
Solving for $\OffsetEst$ we get, 
\begin{equation}
    \OffsetEst = \frac{(T2-T1)+(T3-T4)}{2}
    \label{eq:offset_est}
\end{equation}

The above equations also hold for \ac{PTP}, albeit for different message names. 
The source starts the synchronization procedure, so three messages must be exchanged for the sink to acquire all four timestamps. 
Devices that support \ac{PTP} achieve higher timestamp accuracy through hardware support, enabling significantly improved synchronization performance. 
The precision of the devices mainly causes the timestamping error, and for \ac{COTS} \ac{NIC} devices, it was found to be maximum $\SI{8}{\nano\second}$~\cite{moongen}, which is negligible compared to the \ac{PDV} in the network~\cite{yashinfocom}.
Chrony~\cite{chrony} is an implementation of NTP that can also utilize hardware timestamping to improve performance. 

\RTwo{Global Navigation Satellite Systems (GNSS), such as GPS, provide an
alternative source of time: a GNSS-based primary reference time clock is
accurate within 100\,ns to 40\,ns of UTC, depending on its
class~\cite{g8272}. However, GNSS requires an antenna with an unobstructed
view of the sky, which is rarely available for devices inside data
centers, buildings, or industrial plants, and equipping every device with
a receiver is costly. Moreover, GNSS signals are weak and can be jammed
or spoofed~\cite{psiaki2016}. In practice, GNSS is therefore used to
discipline a few source/server clocks, from which time is distributed to
all other devices over the network using PTP or NTP~\cite{ITUstd, spanner}.
Our method is complementary to GNSS: it does not replace the time source,
but reduces the error accumulated while distributing its time through the
network.}

\subsection{Correcting Estimation Errors}
\label{subsec:correcting_est_errors}

Assuming perfect timestamping, the primary source of error from \ac{PDV} comes when the forward and reverse paths of the synchronization messages experience different delays. 
From Equations \ref{eq:forward_time}, \ref{eq:reverse_time} and \ref{eq:offset_est} we get the offset estimation error $\epsilon$, 
\begin{equation}
    \Offset - \OffsetEst = \epsilon = \frac{(\Delay{Resp}{}-\Delay{Req}{})}{2}.
    \label{eq:offset_error}
\end{equation}

The clock adjustment process can be viewed as a closed-loop control system in which a control signal is applied to the clock based on noisy offset measurements. 
Thus, a two-step control process is usually applied: a low-pass filter to reduce the effect of packets that experienced significant delays, and a \ac{PI} filter to incrementally apply the control signal in small steps.

The \textit{ptp4l}~\cite{linuxptp} is an IEEE 1588 Precision Time Protocol implementation from the Linux PTP project. 
Here, during the low-pass filtering, the offset estimation process is decoupled from the delay estimation process.  
The delay measurement is a \ac{RTT} value. 
The delay of the \textit{sync} packet is then assumed to be the median of the last $M$ measurements. 
The filter length $M$ is configurable. 
Huygens~\cite{genghuygens} suggests using coded probes and \ac{SVM}-like learning over coded probes to find packets that were not congested. 
However, this requires a large number of packets, which take up to 5 Mbps of the link bandwidth.
This can be negligible for data centers with high-capacity links, but is difficult to deploy for industrial networks with low-capacity links.
Similar coded probing methods have been proposed~\cite{levesqueDelayAssymetry, Murakami, lee}. However, they not only require twice the bandwidth of standard PTP, but also a change in the standard. 
\ac{NTP} uses a Huff-and-Puff filter~\cite{mills1992rfc} to pick packets that faced less delays while Chrony~\cite{chrony} and Firefly~\cite{firefly} use min-RTT filtering for the same purpose.
Clock drifts show a strong dependence on the temperature of the crystal oscillator driving the clock, which changes slowly.
Hence, Kalman Filtering~\cite{kalmanfilteringreview,kalmanFiltering, KalmanFiltering2} or Linear programming~\cite{LinearProgramming} is used to track these dependencies. 
Graham~\cite{graham} uses temperature sensors on the motherboard and \ac{ML} to compensate for the oscillator drift. 
These techniques treat \ac{PDV} as random Gaussian noise. 
However, competing network traffic is transient, and the distribution of link delay is not Gaussian~\cite{Timely}. 
These techniques do not meet the tight synchronization guarantees needed for many applications under heavy network traffic~\cite{LimitationsClockSync}.
\ROne{Despite the theoretical appeal of Kalman filtering and linear programming-based methods, RTT-based filters have remained dominant in deployed clock synchronization systems, owing to their simplicity, parameter-light operation, and robustness under the heavy-tailed, non-Gaussian delay distributions observed in practice. 
Thus, we use RTT-based filters as both baselines and show that our method further improves their performance. } 

\ac{RTT}- based filtering techniques are inherently limited in removing the delay asymmetry, as they must use statistical methods on the joint distribution of the forward and reverse delay. 
To achieve better clock synchronization performance, \ac{PTP} provides hardware support in intermediate switches via \ac{TC}, allowing delay asymmetry to be compensated individually in the forward and reverse paths. 
The \ac{TC} adds the residence time of the packet in the switch to a correction field. 
Then the above calculations for filtered delay and $\OffsetEst$ (Eq.~\ref{eq:offset_error}) remain the same with $T1$ being replaced by $T1+\Delay{ReqCorr}{}$ and $T4$ by $T4-\Delay{RespCorr}{}$. 
Here $\Delay{ReqCorr}{}$ and $\Delay{RespCorr}{}$ are the total residence times of the respective packets in all the intermediate switches.
However, this operation at the switch is complex~\cite{TransparentClockImpl}, as it requires two timestamping operations: extracting the correction field value from the PTP header, subtracting the ingress timestamp from the egress timestamp, and adding the result to the correction field. 
This difficult implementation hinders the adoption of \ac{TC} on generic mass-manufactured switch chipsets.
While some architectures for a full-hardware implementation of \ac{TC} have been proposed \cite{hwptptc}, the implementation remains complex and is supported on fewer ASICs \cite{NetTimeLogicPTPTransparentClock}. 
Fig.~\ref {fig:tc_comparison_processing} shows the difference between the processing time of a \ac{PTP} packet in the 1G Kontron D10 switch. 
The residence time of the packet in the switch increases by 3 orders of magnitude when using the \ac{TC} mode instead of the \ac{LS} mode. 
Moreover, we monitor the switch's CPU utilization while handling 1 source and 3 sinks simultaneously.  
Fig.~\ref{fig:tc_comparison_cpu} shows that as soon as the test starts, the utilization in the \ac{TC} mode shoots up, reaching nearly 20\%. 
The switch was not loaded with any other tasks during this test. 
Thus, the complex pipeline increases the processing delay for \ac{PTP} packets~\cite{tc_verifying_performance, substation_performance}, and some vendors have demonstrated significant errors in \ac{TC} residence-time accuracy under network load~\cite{tc_verifying_performance}.
\ROne{Moving synchronization into the data plane of programmable switches does not fully remove this dependency on the switch CPU. 
DPTP~\cite{dptp} implements time synchronization in the Tofino data plane and achieves
nanosecond-level accuracy under load. However, since the exact transmit timestamp is accessible only to the switch control plane, DPTP generates follow-up messages on the switch CPU, whose load grows with the number of ports. 
Moreover, DPTP is a separate protocol that relies on stateful memory on every switch and is not backward-compatible with existing PTP or NTP deployments. In contrast, our approach requires neither timestamps nor per-message state in the switch and works
alongside existing protocols.}
In Sec. \ref{subsec:impl_complexity}, we illustrate the required resources of a \ac{TC} implemented on P4~\cite{p4} by analyzing the compiler output and comparing it with our approach.  
\begin{figure}
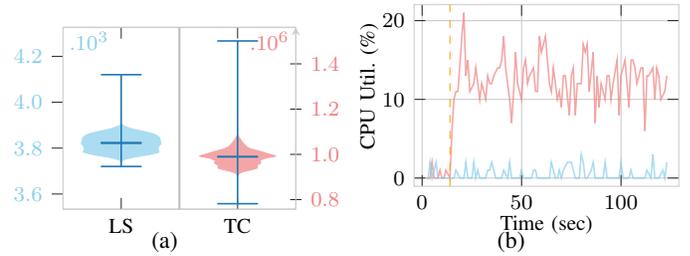

    \centering
    \begin{subfigure}{0.48\columnwidth}
        \centering

         \vspace{-0.8cm}
        \caption{}
        \label{fig:tc_comparison_processing}
    \end{subfigure}
    \hfill
    \begin{subfigure}{0.48\columnwidth}
        \centering
%
         \vspace{-0.8cm}
        \caption{}
        \label{fig:tc_comparison_cpu}
    \end{subfigure}
    \caption{\small{Complexity of Transparent Clocks: (a) Measurements of the packet residence time show an increase by 3 orders of magnitude when the switch is running in TC mode. (b) The inherent complexity of TC requires involving the CPU of the switch, limiting its scalability in larger networks.} }
    \label{fig:tc_comparison}
\end{figure}

Finally, to mitigate the impact of an erroneous offset estimation and the finite delay in clock adjustment, the client applies a \ac{PI} filter that incrementally adjusts its clock over time. 
The \ac{PI} filter is useful for stabilizing the clock synchronization control loop. 
The P and I coefficients can be finely tuned using Bayesian optimization methods~\cite {bayesian_optimization}. 
However, the \ac{PI} filter takes the offset estimation as an input, and a large error in the estimation procedure can only be partially corrected by such a filter. 

In this paper, we propose a scheme where the offset estimation error is improved by congestion marking - a scheme much simpler than a \ac{TC}. 
The proposed method is backward-compatible with widely used versions of PTP, NTP, and Chrony. 
It requires only a few changes at the endpoints and in the network switches for its basic adoption. 

Such marking methods are already used in congestion control mechanisms; in the next section, we provide some background on this topic. 

\subsection{Congestion Marking}
\label{subsec:congestion_marking}
\begin{figure}
    \centering
    \begin{subfigure}{\columnwidth}
        \centering
         \includegraphics[width=\textwidth]{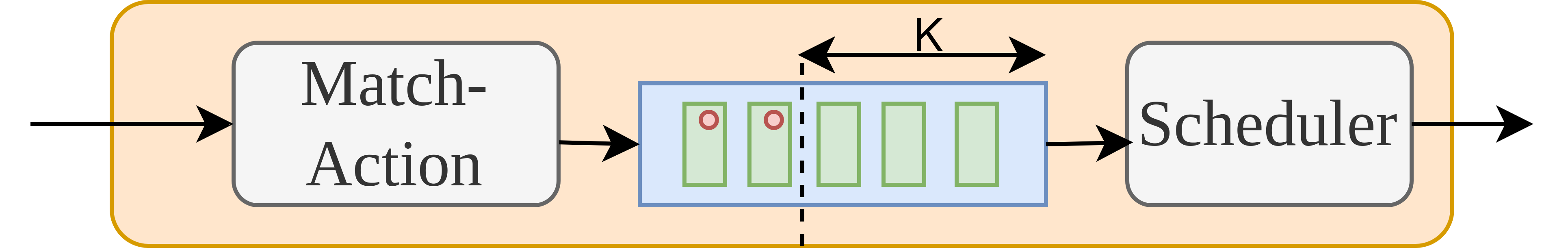}
        \caption{}
        \label{fig:ecn_egress}
    \end{subfigure}
    \hfill
    \begin{subfigure}{\columnwidth}
        \centering

         \caption{}
        \label{fig:scatter_plot_motivation}
    \end{subfigure}
    \caption{\small{ECN Marking at Egress Port: (a) The match-action pipeline, which puts the packets in different priority queues, also takes care of congestion marking. Before placing them in their respective queues, it sets the ECN flag to CE(11) if the queue length exceeds a threshold $K$. Thus, marked packets also indicate delayed packets due to queuing. (b) A delay measurement  of roughly 9000 packets. The red line shows the inferred time threshold from $K$. We see that most packets above the red line are orange (marked) and those below it are blue (unmarked), indicating ECN marking's ability to signal delay.}}
    \label{fig:ecn_background}
\end{figure}

The first proposal to mark congested packets was suggested in~\cite{ramakrishnan2001addition} as \ac{ECN} to improve congestion control.  
It suggested that two bits, called the \ac{ECN} field in the IP header, be used by switches and routers to mark packets that are queued beyond a threshold of $K$ bytes.  
This enables congestion control algorithms to adjust their flows before a packet is dropped. 
The two bits allow for four code points, which indicate whether a packet is ECN-capable and whether congestion is experienced. 
This leaves one code point (01) free for future experimental use~\cite{ecnNew}, most commonly adopted by L4S~\cite{rfc9331}.
The match-action pipeline that places packets in the respective egress queues also marks packets as congested if they are already marked as ECN-capable and the queue is filled above the $ K$-byte threshold. 
Thus, congestion-marked packets also indicate delayed packets, with the delay amount dependent on the threshold $K$. 
The congestion marking mechanism at the egress port of a switch is illustrated in Fig.~\ref{fig:ecn_egress}. 
To determine the extent to which a commodity switch's congestion marking scheme can indicate delay, we measure delay on the EdgeCore AS7726 switch at 10 Gbps. 
For ECN, we set $K=3600$ bytes and congest the switch's egress port with 1500-byte (MTU) packets from different sources at a combined rate of 4 Gbps.
The results are shown in Fig.~\ref{fig:scatter_plot_motivation}. 
The red line is obtained from $8K/10^{10}$, i.e. the total bits before the threshold divided by the line rate of the egress port. 
We then send probe \ac{PTP} packets. 
We observe that the mechanism correctly classifies packets below and above the threshold, with most misclassifications near the threshold line due to queue transience. 

\ac{ECN} is simple to implement: a switch only sets the two ECN bits in the IP
header without inspecting the packet any further than before the introduction of ECN. 
It is also backward compatible, since devices that do not support ECN simply ignore these
bits.
Choosing the right threshold $K$ and marking strategy depends on demand, capacity, topology, and TCP variants, and remains an active area of investigation ~\cite{ECNTuning, KaanTCP}.
In this work, we evaluate the threshold as a function of the mean packet queuing time and provide an Algorithm to find the optimal threshold given an empirical per-hop queuing distribution.

\subsection{Motivating Example}
\label{subsec:motivating_example}

\begin{figure}
    \centering
    \begin{subfigure}{0.98\columnwidth}
        \centering        \includegraphics[width=0.8\linewidth]{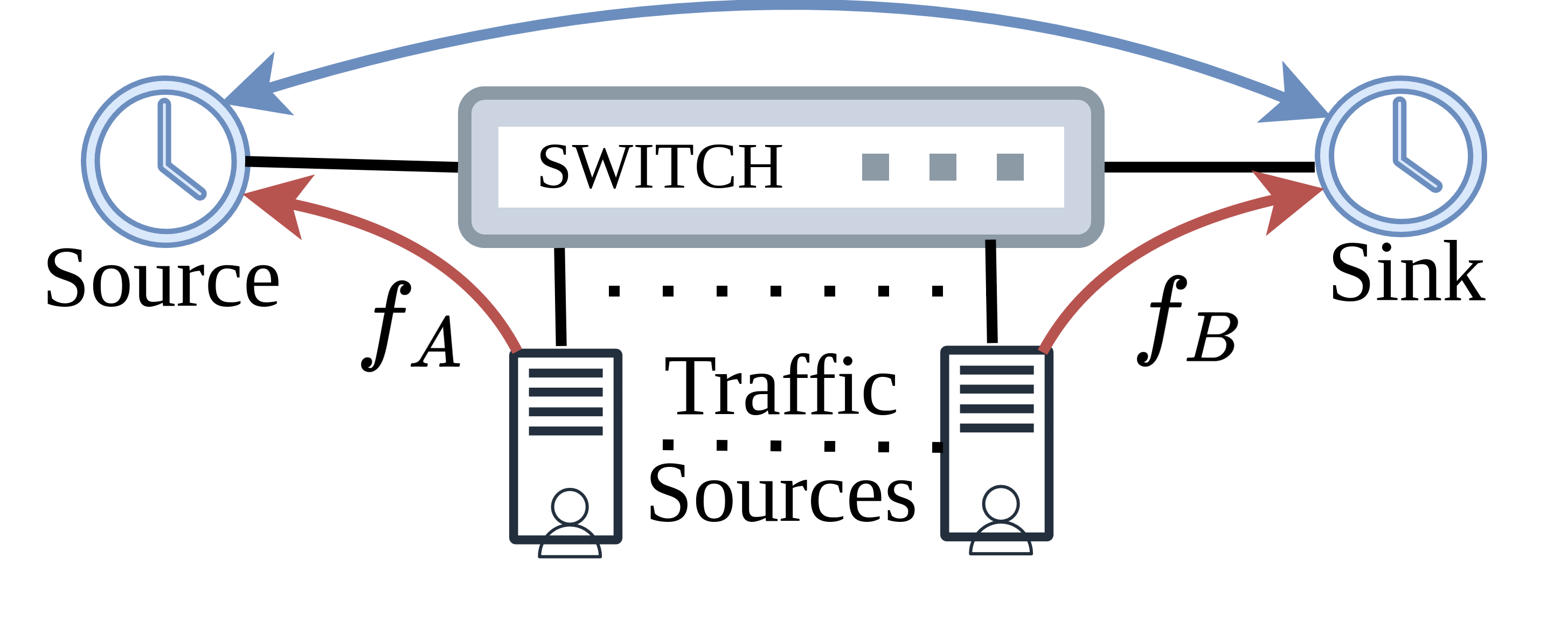}
        \vspace{-0.7cm}
        \caption{}
        \label{fig:omnet_setup}
    \end{subfigure}
    \hfill
    \begin{subfigure}{0.98\columnwidth}
        \centering
\begin{tikzpicture}

\definecolor{darkgray176}{RGB}{176,176,176}
\definecolor{darkorange25512714}{RGB}{255,127,14}
\definecolor{forestgreen4416044}{RGB}{44,160,44}
\definecolor{steelblue31119180}{RGB}{31,119,180}
\definecolor{green}{RGB}{0,128,0}

\begin{axis}[
width=9cm,
height=3.3cm,
tick align=outside,
tick pos=left,
label style={font=\footnotesize}, %
x grid style={darkgray176},
xlabel={$t$(s)},
xlabel style={at={(axis description cs:0.5,0.05)}, anchor=center, color=black},
xmajorgrids,
xmin=-0.099997035, xmax=10.350002965,
 tick label style={font=\footnotesize}, %
xtick style={color=black},
y grid style={darkgray176},
ylabel={$\epsilon$ ($\si{\micro\second}$)},
ymajorgrids,
ylabel style={at={(axis description cs:0.06,0.55)}, anchor=center, color=black},
ymin=-20.447149399995, ymax=56.422569399895,
ytick style={color=black},
legend style={
    at={(0.5,1.02)},
    anchor=south,
    legend columns=-1,
    font=\footnotesize,
    /tikz/every even column/.append style={column sep=0.2cm}
}
]
\path [draw=orange, fill=red, opacity=0.05]
(axis cs:2,0)
--(axis cs:2,100)
--(axis cs:8,100)
--(axis cs:8,-50)
--(axis cs:2,-50)
--cycle;
\addplot [semithick, steelblue31119180]
table {%
0.375002965 -0.04
0.500002965 -0.04
0.625002965 -0.04
0.750002965 -0.04
0.875002965 -0.04
1.000002965 -0.04
1.125002965 -0.04
1.250002965 -0.04
1.375002965 -0.04
1.500002965 -0.04
1.625002965 -0.04
1.750002965 -0.04
1.875002965 -0.04
2.000002965 -0.04
2.125002965 -0.04
2.25001329557 10.5465699999
2.375002965 -5.3332849999
2.500030616823 27.867823
2.625036986902 19.1709899999
2.750002965 -16.6224509999
2.875040945467 38.196467
3.000009129305 -12.7379289999
3.125023906479 17.947326
3.250023874377 10.5266369999
3.375028488296 15.1566069999
3.500024427123 8.788475
3.625032631179 19.0231169999
3.750003879162 -13.8309279999
3.875002965 -0.6250809999
4.000047989687 45.2406869999
4.125021662 -22.680344
4.250011657111 8.908111
4.375011037753 3.8146969999
4.500062402141 52.6557639999
4.62501446819 -16.7108809999
4.750002965 -5.9195949999
4.875016728771 13.9797709999
5.000010072791 0.313905
5.125024977647 18.5467509999
5.25003119217 -1.2761539999
5.375028342045 20.6439599999
5.500003121838 -12.4436849999
5.625026056989 23.10157
5.750006116308 -8.306687
5.875019183129 11.6414749999
6.000026523031 17.0814659999
6.125004558802 -10.097214
6.250007111003 3.437102
6.375002965 -2.2410019999
6.500003547161 0.798161
6.625034447984 31.2799029999
6.750031347766 12.7292739999
6.875005459851 -11.6085319999
7.000016141911 12.017485
7.125002965 -6.756456
7.250008262208 5.513208
7.375002965 -2.8166039999
7.500021097821 18.348821
7.62502764749 15.704079
7.750052066135 33.9348899999
7.875027725478 1.7544099999
8.000005382575 -9.874664
8.125002965 -1.3767879999
8.250002965 -0.04
8.375002965 -0.04
8.500002965 -0.04
8.625002965 -0.04
8.750002965 -0.04
8.875002965 -0.04
9.000002965 -0.04
9.125002965 -0.04
9.250002965 -0.04
9.375002965 -0.04
9.500002965 -0.04
9.625002965 -0.04
9.750002965 -0.04
9.875002965 -0.04
};\addlegendentry{Raw:$14.3\mu$s}
\addplot [semithick, red]
table {%
0.375002965 -0.04
0.500002965 -0.04
0.625002965 -0.04
0.750002965 -0.04
0.875002965 -0.04
1.000002965 -0.04
1.125002965 -0.04
1.250002965 -0.04
1.375002965 -0.04
1.500002965 -0.04
1.625002965 -0.04
1.750002965 -0.04
1.875002965 -0.04
2.000002965 -0.04
2.125002965 -0.04
2.25001329557 10.5465699999
2.375002965 -0.04
2.500030616823 27.867823
2.625036986902 33.1249019999
2.750002965 -0.04
2.875040945467 38.196467
3.000009129305 1.08702
3.125023906479 15.8641939999
3.250023874377 10.5266369999
3.375028488296 15.1405559999
3.500024427123 11.0793829999
3.625032631179 19.2834389999
3.750003879162 -9.7288999999
3.875002965 -10.6387399999
4.000047989687 34.6419469999
4.125021662 -10.8990619999
4.250011657111 -1.9509509999
4.375011037753 3.8146969999
4.500062402141 52.6557639999
4.62501446819 7.5548129999
4.750002965 -4.5140559999
4.875016728771 9.505715
5.000010072791 2.8497349999
5.125024977647 17.7545909999
5.25003119217 3.9785749999
5.375028342045 19.7134499999
5.500003121838 -5.506757
5.625026056989 18.3589039999
5.750006116308 -3.6425779999
5.875019183129 8.3960439999
6.000026523031 17.0814659999
6.125004558802 -4.8827629999
6.250007111003 -2.3305619999
6.375002965 -2.2410019999
6.500003547161 -1.402841
6.625034447984 29.9953299999
6.750031347766 26.397764
6.875005459851 0.5098489999
7.000016141911 12.017485
7.125002965 -2.2410019999
7.250008262208 4.137782
7.375002965 -2.8166039999
7.500021097821 15.5722169999
7.62502764749 22.1218859999
7.750052066135 43.6275309999
7.875027725478 18.2600219999
8.000005382575 -6.560836
8.125002965 -9.2344109999
8.250002965 -9.2344109999
8.375002965 -9.2344109999
8.500002965 -1.3767879999
8.625002965 -0.04
8.750002965 -0.04
8.875002965 -0.04
9.000002965 -0.04
9.125002965 -0.04
9.250002965 -0.04
9.375002965 -0.04
9.500002965 -0.04
9.625002965 -0.04
9.750002965 -0.04
9.875002965 -0.04
};\addlegendentry{Filtered:$14.1\mu$s}
\addplot [thick, green]
table {%
0.375002965 -0.04
0.500002965 -0.04
0.625002965 -0.04
0.750002965 -0.04
0.875002965 -0.04
1.000002965 -0.04
1.125002965 -0.04
1.250002965 -0.04
1.375002965 -0.04
1.500002965 -0.04
1.625002965 -0.04
1.750002965 -0.04
1.875002965 -0.04
2.000002965 -0.04
2.125002965 -0.04
2.25001329557 0.0465699999
2.375002965 -0.04
2.500030616823 -0.1321769999
2.625036986902 5.1249019999
2.750002965 -0.04
2.875040945467 10.1964669999
3.000009129305 -0.619695
3.125023906479 0.1574789999
3.250023874377 0.0266369999
3.375028488296 1.156607
3.500024427123 0.5793829999
3.625032631179 1.7834389999
3.750003879162 0.7710999999
3.875002965 -0.6250809999
4.000047989687 16.6556059999
4.125021662 -0.6250809999
4.250011657111 1.3230299999
4.375011037753 0.703672
4.500062402141 28.155764
4.62501446819 0.2451339999
4.750002965 -0.7043769999
4.875016728771 -0.6846059999
5.000010072791 -0.305804
5.125024977647 1.0467509999
5.25003119217 -0.8237259999
5.375028342045 1.0831589999
5.500003121838 0.1909419999
5.625026056989 2.1015699999
5.750006116308 -0.699215
5.875019183129 -0.8612899999
6.000026523031 -0.932388
6.125004558802 1.8561479999
6.250007111003 0.9083489999
6.375002965 0.0063459999
6.500003547161 0.798161
6.625034447984 3.2799029999
6.750031347766 0.1477639999
6.875005459851 -1.2082299999
7.000016141911 -0.926472
7.125002965 -0.3593829999
7.250008262208 2.0132079999
7.375002965 -0.04
7.500021097821 0.8488209999
7.62502764749 0.39849
7.750052066135 18.1848899999
7.875027725478 0.0320669999
8.000005382575 2.375336
8.125002965 -0.4844109999
8.250002965 -0.2982389999
8.375002965 -0.2982389999
8.500002965 -0.04
8.625002965 -0.04
8.750002965 -0.04
8.875002965 -0.04
9.000002965 -0.04
9.125002965 -0.04
9.250002965 -0.04
9.375002965 -0.04
9.500002965 -0.04
9.625002965 -0.04
9.750002965 -0.04
9.875002965 -0.04
};\addlegendentry{Marked:$4.5\mu$s}
\end{axis}

\end{tikzpicture}
         \vspace{-0.7cm}
        \caption{}
        \label{fig:omnet_results}
    \end{subfigure}
    \caption{\small{Motivating Example: (a) A simulation setup of PTP on OMNET++ where a short-lived moderate UDP flow causes asymmetric queing delay leading to an offset estimation error. (b) The raw offset estimation error (blue) can only be partially corrected by a median delay filter (red), but further improved (70\%) by marking packets (green).} }
    \label{fig:motivating_example}
\end{figure}
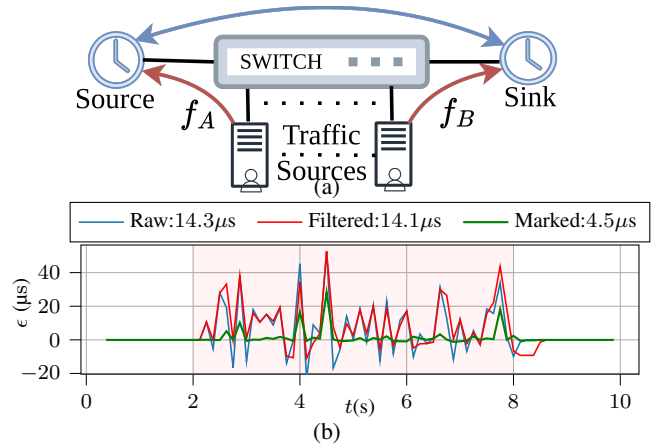

Fig. \ref{fig:omnet_setup} shows a simulation setup where a short-lived UDP flow $f_{B}$ causes delay asymmetry in the PTP synchronization procedure, leading to offset estimation error; in this case, $f_{A}$ is absent. 
The simulation is set up in OMNET++ using the INET framework, with modifications to the gPTP application.
In such a case, the offset estimation error during the duration of flow increases dramatically as seen in Fig.~\ref{fig:omnet_results}. 
Such a flow can only be partially corrected by a median-delay filter, as most packets experience some queuing delay, resulting in asymmetry between the forward and reverse paths of PTP packets. 
However, if the packets that experience queuing delay are marked by the switch, the asymmetry can be partially compensated by the sink, thereby greatly reducing the offset estimation error. 
The values in the legend are the \ac{RMS} errors between the two clocks in $\mu$s, and we can see that our approach reduces them by $\approx 70\%$. 

Our approach compensates for part of the queuing delay, reduces \ac{PDV}, and offers lower implementation complexity than \ac{TC}.
In the next section, we demonstrate how the offset estimation error $\epsilon$ is related to queuing estimation and introduce our congestion marking approach.

\section{Offset Estimation with Congestion Marking}
\label{sec:offset_estimation}
\RThree{This section presents the core idea of our method. We first show that the
offset estimation error depends only on the error in estimating the
queuing delays (Sec.~\ref{subsec:queuing_delay_corr}), then describe how congestion marks are used
to estimate these delays (Sec.~\ref{subsec:correction_of_marked_packets}), and finally measure how accurately
marks indicate delay on commodity and programmable switches
(Sec.~\ref{subsec:measurement_of_marking_accuracy}).}
\subsection{Generic Expressions for Queuing Delay Correction}
\label{subsec:queuing_delay_corr}

The delays in the forward and reverse directions of the network are usually very similar when packets do not experience congestion~\cite{yashinfocom}, hence Cristian's algorithm assumes that the statistical properties of the forward and reverse delays are the same. 
Usually, the forward and reverse paths of a packet are the same, as networks use a loop-free spanning tree to establish a unique path from the source to the destination and back.
In cases where they are not the same, a static compensation for the delay asymmetry can be made~\cite{levesqueDelayAssymetry}. 

Thus, we split the delay of message $m$ into a base delay and an additional queueing delay, 
\begin{equation}
    \Delay{m}{} = \Delay{m}{\text{base}} + \Delay{m}{q}, \textrm{ for } m \in \{\text{Resp}, \text{Req}\}.  
\end{equation}

Putting these in equations~\ref{eq:forward_time} and \ref{eq:reverse_time} and assuming that $\Delay{\text{Req}}{\text{base}}=\Delay{\text{Resp}}{\text{base}}$, we get that the offset estimation error is due to different queuing delays in the forward and reverse paths, 
\begin{equation}
    \epsilon = \frac{(\Delay{Resp}{q}-\Delay{Req}{q})}{2}.
    \label{eq:error_queueing}
\end{equation}
If one estimates the queuing delay, it can be compensated for in equations \ref{eq:forward_time} and \ref{eq:reverse_time}. 
These compensations percolate to the offset estimation error in the form of correction terms,
\begin{equation}
    \epsilon_{\text{comp}} =  \frac{\Delay{Resp}{q}-\Delay{Req}{q}-\DelayEst{Resp}{q}+\DelayEst{Req}{q}}{2}.
\end{equation}
If $\Delay{m}{E}=\Delay{m}{q}-\DelayEst{m}{q}$ is the error in the estimation of the queing delay, then the compensated offset estimation error is a function of only the error in the estimation of the queing delays, 

\begin{equation}
    \epsilon_{\text{comp}} = \frac{\Delay{Resp}{E}-\Delay{Req}{E}}{2}.
    \label{eq:error_delay_est}
\end{equation}

Taking the absolute value on both sides of the equation, as long as the difference between the queuing delay estimation error and the actual queuing delay is smaller, we can reduce the offset estimation error. 
Note that a \ac{TC} does exactly that but at a much higher implementation cost. 
It aims to minimize the offset estimation error by accounting for packet queuing and processing delays. 

\subsection{Correction of Marked Packets}
\label{subsec:correction_of_marked_packets}

As explained in Sec. \ref{subsec:congestion_marking}, a marked packet is an indication of a delayed packet.
Similar to ECN, we assume that the congestion-marking indication is carried in the packet header. 
Depending on the level of congestion and the number of available bits in the header, the indication can take an integer value up to $N$; hence, we call this header value a marking counter. 
We discuss different available bits in the PTP, NTP, and IP headers for the marking counter in Sec. ~\ref{subsec:available_packet_headers} and the implementation of the marking counter itself in Sec.~\ref{subsec:marking_coutner}.
Now, we discuss the correction method based on the counter value and its effect on the delay estimation error. 
Each time a switch places a packet in an egress queue with a queue length exceeding the threshold $K$, the marking counter is incremented.
We also present the general case, where at each hop, $R$ thresholds are available at the switch, each a multiple of $K$.  
The counter can have a maximum value of $N$, depending on the number of header bits and the marking strategy.
Let $\mathcal{L}_{m}=(L_{m}^{0}, L_{m}^{1},...,L_{m}^{s})$ be an ordered set of links traversed by message $m$ and $\DelayHop{m}{q}{l} \text{for } l \in \mathcal{L}_{m}$, the queuing delay experienced per hop such that the total queuing delay is the sum of individual queueing delays across all hops.   
For the ease of notation, when concerning only one message $m$, we enumerate a hop only with its number, i.e., we write $L_{m}^{a}$ as simply $a$.  
We assume $\Delay{K}{*}$ as the estimate of the minimum queuing delay that a marked packet experienced when the threshold was set to $K$. 
At each hop, the switch compares the current queue length against $R$ thresholds and increments the marking counter accordingly. 
The client then checks the number of congestion marks $n$ and corrects the delay by $n\Delay{K}{*}$, assuming that a mark corresponds to the packet facing a delay of at least $\Delay{K}{*}$.
Thus, an \textit{increment} in the counter corresponds to a \textit{decrement} in the delay estimation error by $\Delay{K}{*}$. 
Since the counter can take at most the value of $N$, we have a constraint that $n,R\leq N$.
The delay estimation error at each hop $l$ is, 
\begin{equation}
    \DelayHop{m}{E}{l} =
    \begin{cases}
    \DelayHop{m}{q}{l} - R\Delay{K}{*} & \text{if } \DelayHop{m}{q}{l} > R\Delay{K}{*} \\
    \DelayHop{m}{q}{l} - (R-1)\Delay{K}{*} & \text{if } R\Delay{K}{*} \geq \DelayHop{m}{q}{l} > (R-1)\Delay{K}{*} \\
    \vdots & \vdots \\
    \DelayHop{m}{q}{l} - \Delay{K}{*} & \text{if } 2\Delay{K}{*} \geq \DelayHop{m}{q}{l} > \Delay{K}{*} \\
    \DelayHop{m}{q}{l} & \text{otherwise}. 
    \end{cases}. 
    \label{eq:performed_correction}
\end{equation}
$\DelayHop{m}{q}{l} = \DelayHop{m}{E}{l} = 0$ if there is no queuing delay at hop $l$. 
The total delay estimation error is the sum of the delay estimation errors at each hop.
\begin{equation}
    \Delay{m}{E} = \sum_{l \in \mathcal{L}_{m}} \DelayHop{m}{E}{l}
    \label{eq:delay_est_sum}
\end{equation}

\subsection{Measurement of Marking Accuracy}
\label{subsec:measurement_of_marking_accuracy}

\begin{figure}
    \centering
    \begin{subfigure}{0.50\columnwidth}
        \centering
         \includegraphics[width=\textwidth]{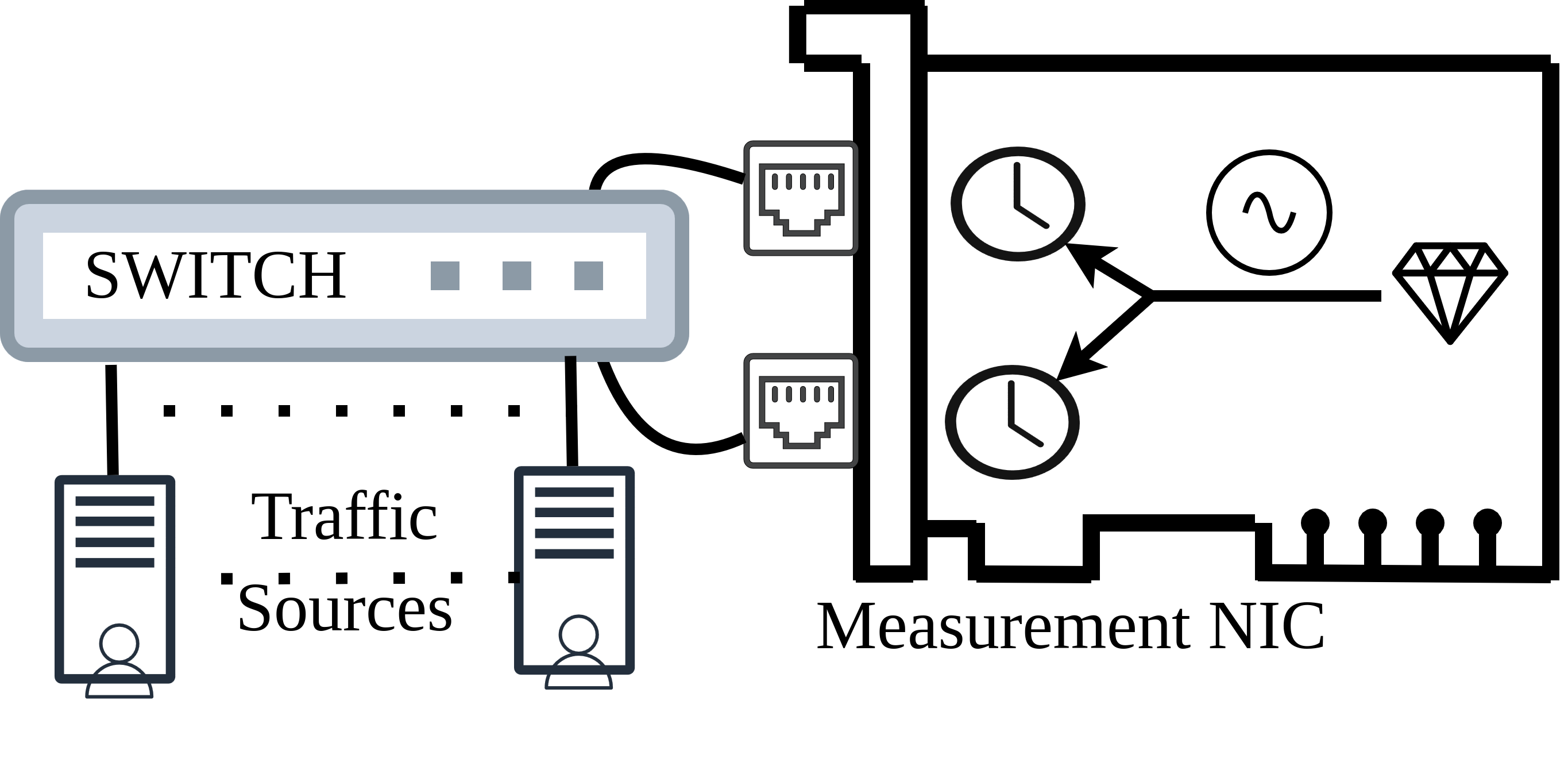}
        \caption{}
        \label{fig:measurement_setup}
    \end{subfigure}
    \hfill
    \begin{subfigure}{0.46\columnwidth}
        \centering
        \includegraphics[width=\textwidth]{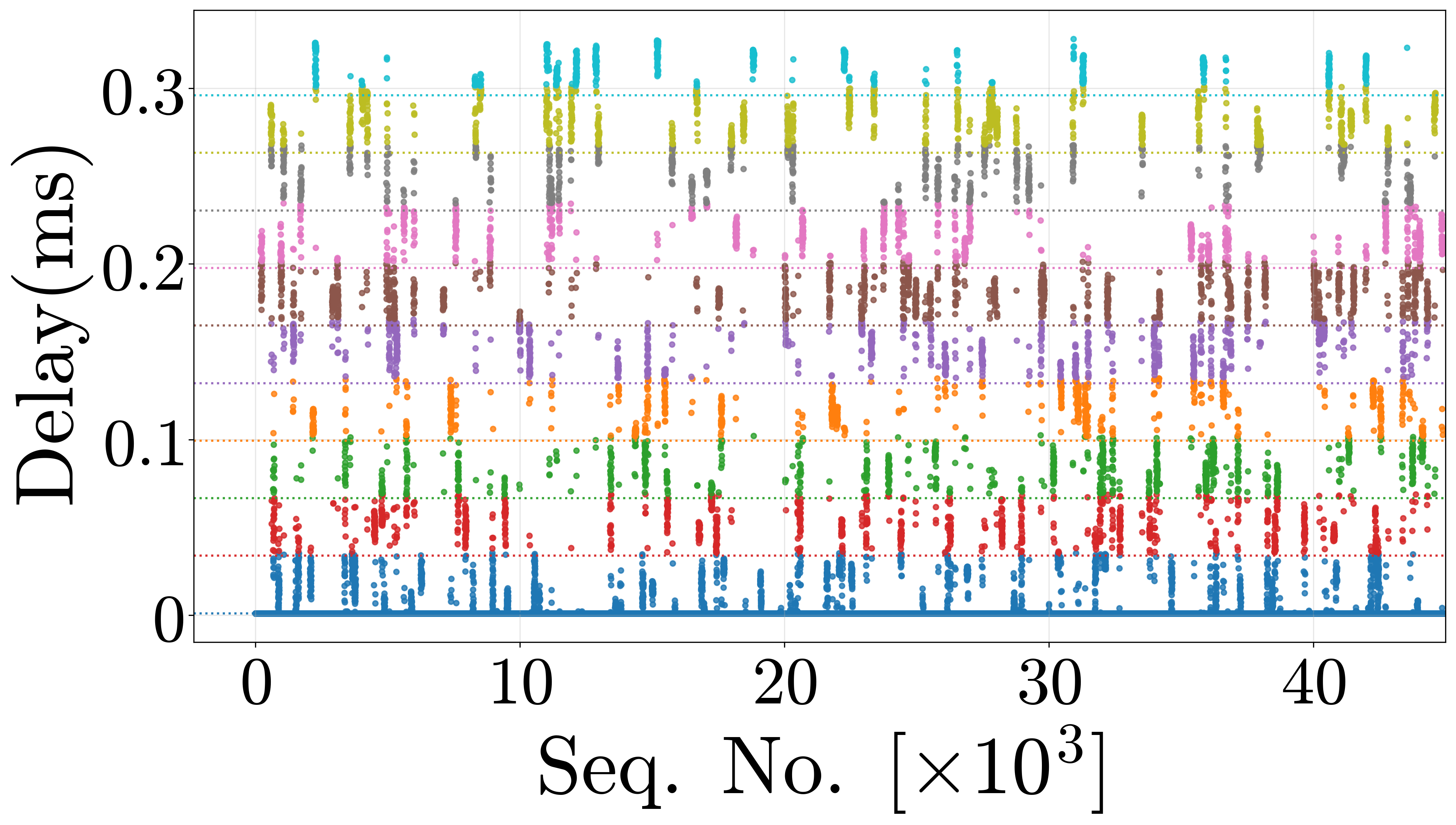}
        \caption{}
        \label{fig:delay_marking_tofino}
    \end{subfigure}
    \caption{\small{Measurement of Marking Accuracy: (a) A two-port NIC where the clocks are driven by the same oscillator is used to accurately measure the delay of packets. Separate sources of traffic congest the switch's outgoing port buffers, causing delays. (b) Example of marking results for one run of Tofino, where packets are colored according to the value of the marking counter. }}
    \label{fig:marking_accuracy}
\end{figure}

An \ac{ECN}-marked packet indicates congestion. 
We wish to use it to obtain information about the packet's delay. 
Thus, the question arises: to what extent are congestion signals an indication of delay, and how effectively can we map a set congestion threshold $K$-bytes to some delay $\Delay{K}{*}$. 
A simple way to map $K$ to $\Delay{K}{*}$ is to estimate the total delay for $K$ bytes according to the line-rate of the egress port.
We will consider more complex mappings between $K$ and $\Delay{K}{*}$ under multiple queues and egress schedulers in future work.
In this work, we use the simple equation that $\Delay{K}{*}$ = $(K\cdot8)/\text{LR}$. 
Where LR is the Line Rate in bits per second. 
This corresponds to the time required to serialize $K$ bytes up to the threshold.
\RThree{In the operator-controlled networks we consider (Sec.~\ref{subsubsec:ECN}), the operator
knows the line rates of the links and can configure the queue threshold
$K$ of each switch port accordingly.}
To validate this assumption, we use a loopback testbed similar to~\cite{yashinfocom}, as illustrated in Fig.~\ref{fig:measurement_setup}.
The two ports on the Intel X520 NIC are driven by the same crystal oscillator; therefore, their relative drift is zero. 
Then we modify the \textit{ptp4l} program to report the one-way delay (T2-T1) and the marking counter value in either the ECN field or the currently unused \textit{Reserved3} PTP field. 
The program uses hardware timestamping and therefore effectively eliminates jitter caused by CPU scheduling, interrupt handling, context switching, or operating system latency from the measurement process. 
Such a setup provides one-way delay values with an accuracy of $\SI{\pm 8}{\nano\second}$.
We then generate TCP and UDP traffic from four other sources at different rates using \textit{iperf} or \textit{nuttcp} to saturate the switch's outgoing port. 
For the \ac{CMC}, we use the Intel Tofino 1, and for the ECN measurements, we use the EdgeCore AS7726-32X (EC) in addition to the Tofino to demonstrate the ability of \ac{COTS} \ac{ECN} marking. 
The exemplary results for \ac {ECN} and \ac{CMC} are shown in Fig.~\ref{fig:scatter_plot_motivation} and \ref{fig:delay_marking_tofino} respectively, where different colors represent packets with a different value of ECN or marking counter, and the dashed horizontal line represents the theoretical threshold based on the line rate.
Each run covers around 46,500 measurements, and for each traffic type and rate, we conduct 10 such runs. 

The misclassification rate increases much more rapidly as the rate of interfering traffic rises, because queues can fill beyond the threshold very quickly.
The worst-case results were obtained with a total aggregated rate of 9.4 Gbps over a 10 Gbps link. 
Classical ECN tests for both switches were conducted only with UDP traffic, as TCP cross-traffic adjusts its congestion window after receiving a CE-marked packet.
\ROne{For the classical ECN test, the threshold was set to $K=85440$ Bytes for both, giving a $\Delay{K}{*}=68352$~\si{ns}.
These settings differ from the example in Fig.~3. 
For the accuracy measurement, we set the threshold such that a similar number of packets
experience delays below and above it, which tests the classification
most stringently, whereas Fig.~3 uses a smaller threshold to illustrate
the marking behavior more clearly.}
The \ac{CMC} results are obtained using TCP traffic, which performs worse due to its inherent burstiness caused by the congestion window. 
Table~\ref{tab:worst-case_marking_acc} shows the marking accuracy results. 
Between the two switches, the Tofino shows much better results than EC, with a maximum misclassification rate of 2.6\% compared to 8.5\% for EC.  
\ROne{However, the misclassification rate alone overstates the effect on
synchronization: a misclassified packet changes the delay estimate only
by its distance from the threshold, since a packet marked just below the
threshold, or left unmarked just above it, is corrected almost as well
as a correctly classified one. We therefore report the maximum distance
of a misclassified packet from its threshold, relative to the threshold,
which makes the error comparable across thresholds and switches.}
The maximum error for Tofino is just 2.2\% of the threshold value, compared to 5.8\% for the EC. 
\ROne{In the case of \ac{CMC}, the percentage of misclassified packets is higher (9.8\%), with 32 thresholds of 32,768~\si{ns} increasing the likelihood of misclassification. The overall settings were $R=32, K=40960~\text{bytes}, \Delay{K}{*}=32,768~\si{ns}$.  
The worst-case classification error was 11952~\si{ns} away from the threshold. 
While this absolute error may seem large, the misclassification occurred at the 28th level, with the threshold at 917504~\si{ns}, so the error is negligible (~1.3\%) relative to the threshold. 
The worst-case error, as a percentage of the first threshold, is 5.3\%, corresponding to 1730~\si{ns} in this case.} 
The overall results indicate that ECN-capable switches and the Tofino implementation of the CMC can accurately classify marked packets based on their experienced queuing delay.

\begin{table}[t]
\centering
\caption{\small{Worst-Case Marking Accuracy}}
\label{tab:worst-case_marking_acc}

\resizebox{\columnwidth}{!}{%
\begin{tabular}{|c|ccc|ccc|}
\hline
\multirow{2}{*}{} & \multicolumn{3}{c|}{\textbf{ECN}} & \multicolumn{3}{c|}{\textbf{CMC}} \\ \cline{2-7}
 & \multicolumn{1}{c|}{\% Misclassified}
 & \multicolumn{1}{c|}{Max Error (ns)}
 & Max Error \%
 & \multicolumn{1}{c|}{\% Misclassified}
 & \multicolumn{1}{c|}{Max Error (ns)}
 & Max Error \% \\ \hline

\textbf{Tofino}   & \multicolumn{1}{c|}{2.6} & \multicolumn{1}{c|}{1497} & 2.2
                  & \multicolumn{1}{c|}{9.8} & \multicolumn{1}{c|}{11952} & 5.3 \\ \hline

\textbf{EdgeCore} & \multicolumn{1}{c|}{8.5} & \multicolumn{1}{c|}{3951} & 5.8
                  & \multicolumn{1}{c|}{-}    & \multicolumn{1}{c|}{-}      & -     \\ \hline
\end{tabular}%
}
\end{table}

\section{Methods of Implementation}
\label{sec:methods}
\RThree{This section describes how our method is implemented in practice. We
present the marking counter (Sec.~\ref{subsec:marking_coutner}) and the packet header fields that
can carry it (Sec.~\ref{subsec:available_packet_headers}), compare the resources required by our P4
implementation with those of a transparent clock (Sec.~\ref{subsec:impl_complexity}), and discuss
deployment on non-programmable switches (Sec.~\ref{subsec:broader_adaptation}).}
\subsection{Marking Counter}
\label{subsec:marking_coutner}

\begin{figure}
    \centering
    \includegraphics[width=0.9\columnwidth]{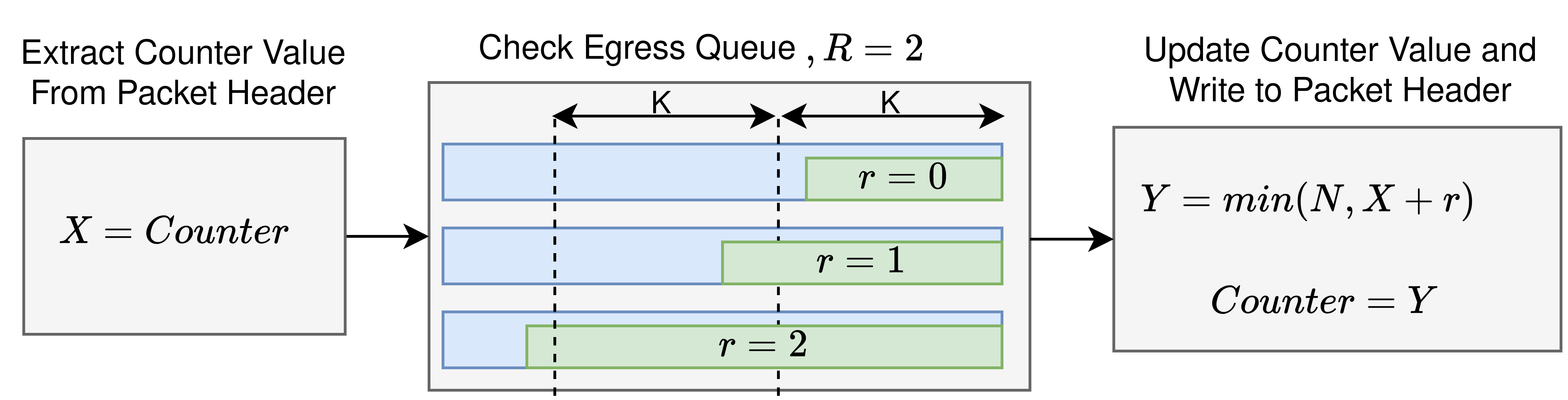}
    \caption{\small{Example working of the Marking Counter: The left side indicates the number $r$ thresholds crossed by the current buffer utilization at the switch egress. The right-hand side indicates the congestion marking operation and the resulting value of counter $Y$.}}
    \label{fig:marking_counter_example}
\end{figure}

The CMC manipulates a counter in the packet header that can take values up
to $N$; the header fields available for it are discussed in Sec.~\ref{subsec:available_packet_headers}. At
each switch, the counter is incremented by $r \in \gk{0, 1, \ldots, R}$ if
the egress queue holds between $K \cdot r$ and $K \cdot (r+1)$ bytes, and
by $R$ if it holds more than $K \cdot R$ bytes, as illustrated in
Fig.~\ref{fig:marking_counter_example}. If $X$ and $Y$ are the counter values
before and after a switch, the marking operation is
$Y = \min(N, X + r)$. This is repeated at every hop until the counter
reaches $N$ or the packet reaches its destination. The receiving endpoint
multiplies the counter by $\Delay{K}{*}$ and subtracts the result from the
receive timestamp, as explained in Sec.~\ref{subsec:correction_of_marked_packets}. With $b$ bits in the
header, the counter allows $N = 2^b - 1$ marks in total.

\RTwo{The counter is backward-compatible with both switches and endpoints. A switch that does not support our method forwards synchronization packets
without modifying the counter, so its queuing delay remains uncorrected
while all supporting switches on the path still contribute their
correction. Endpoints that do not support our method ignore the counter
and operate as before. Supporting switches and endpoints can therefore be
deployed incrementally.}

\subsection{Available Packet Headers}
\label{subsec:available_packet_headers}
Header bits are scarce, since protocol extensions compete for them while
maintaining backward compatibility. We discuss three options for carrying
the counter, which differ in the number of bits they provide and in the
assumptions they place on the network.

\subsubsection{ECN}
\label{subsubsec:ECN}

\ROne{Our method assumes deterministic marking with a threshold $\Delay{K}{*}$
known to the endpoints. 
This rules out NTP over the Internet: endpoints know neither the thresholds nor the number of marking hops along the path, and middleboxes may re-mark, clear, or otherwise treat ECN-capable packets differently. 
We therefore consider operator-controlled networks, such as PTP deployments and data center, enterprise, or industrial networks running NTP, Chrony, or Firefly.}

\ROne{Within such networks, the operator can configure ECN marking as a step function of the instantaneous queue length, so that every ECN-capable packet is marked deterministically once the queue exceeds $K$ bytes. This configuration is common in data center networks~\cite{dctcp} and is supported by commodity switches; both the EdgeCore and the Tofino measurements in Sec.~\ref{subsec:congestion_marking} and Sec.~\ref{subsec:measurement_of_marking_accuracy} use it, without
probabilistic marking. 
The endpoints then only need to send synchronization packets as ECN-capable, using ECT(0) (10), which corresponds to the special case of $\mathbf{b=1}$. 
However, the ECN field is shared with other applications, whose congestion control may rely on \ac{AQM} configured by the operator. 
RED-family AQMs~\cite{RED} mark probabilistically based on an averaged queue length, CoDel~\cite{codel} and PIE~\cite{Pie} act on the sojourn time of packets, and the dual-queue AQM of L4S~\cite{rfc9331} applies sojourn-time-based marking to both of its queues. 
A mark then does not imply that a packet was delayed by at least $\Delay{K}{*}$, nor does its absence imply the opposite. 
If the marking probability as a function of the queuing delay is known, it could be incorporated into the model of Sec.~\ref{sec:analysis}; we leave this to future work. Otherwise, the operator must ensure deterministic step marking on all queues that carry
synchronization packets.}

The second ECN-capable code point, ECT(1) (01), was released for
experimental use~\cite{ecnNew} and has since been adopted by
L4S~\cite{rfc9331}. 
Where the operator can ensure that no traffic in the network uses ECT(1), the ECN field can serve as a modified ECN counter with $\textbf{b=2}$. 
\ROne{In both cases, the marking counter resides in the IP header, so switches can mark synchronization packets without inspecting the PTP or NTP header. 
Carrying the counter in the PTP or NTP header instead decouples it from the ECN semantics used by other traffic and, as discussed next, provides considerably more bits.}

\subsubsection{PTP Header}
\label{subsubsec:using_ptp_header}
Since our method improves with increasing $R$, which requires a larger $N$  (Corollary 1 in Sec.~\ref{subsec:improving_mse_over_one_hop}),
the PTP header offers a better option: its Reserved3 field of 4~bytes
($b = 32$), which immediately follows the correctionField in the PTPv2
header, can carry the counter. Using the PTP header requires switches to
parse it in order to identify synchronization packets and locate the
counter, but it also covers PTP over Layer~2, which has no IP header. The
counter works alongside a TC, since the correctionField remains unchanged
and both corrections can be applied simultaneously.

\subsubsection{NTP Header}
\label{subsubsec:using_ntp_header}
NTP allows custom extension fields aligned to 4-byte boundaries, which can
carry the counter within the network operator's domain. Their length is
limited only by the MTU, leaving up to 500~bytes for the counter. However,
a larger $N$ only helps with many more hops or a finer threshold
$\Delay{K}{*}$, and the precision of the mapping between marks and delay
is limited (Sec.~\ref{subsec:measurement_of_marking_accuracy}). We therefore consider $b = 32$, as in the PTP
header, sufficient.

\RThree{\textbf{Takeaway:} ECN allows marking without inspecting the PTP or NTP
header, but provides at most two bits and requires deterministic marking
on all queues carrying synchronization packets. The PTP and NTP headers
provide 32 bits, improving the granularity of our marking and remain independent of the ECN semantics of other traffic.}

\subsection{Implementation Complexity}
\label{subsec:impl_complexity}
\ROne{
\begin{table}[t]
\centering
\caption{\small{Resource cost compared to the baseline of simple forwarding for TC and CMC.}}
\label{tab:complexity}
\resizebox{\columnwidth}{!}{%
\begin{tabular}{lccccccc}
\toprule
\textbf{Approach} & \textbf{SRAM} & \textbf{TCAM} & \textbf{Stages} & \textbf{Gate-} & \textbf{Action} & \textbf{PHV} & \textbf{Bridged} \\
 & (blocks) & (blocks) & (ingr./egr.) & \textbf{ways} & \textbf{slots} & (bits) & (bytes) \\
\midrule
TC                   & +8  & 0  & +1 / +7 & +20 & +25 & +542 & +4 \\
CMC (bit shift)                & 0   & 0  & 0 / 0   & +1  & +1  & +19  & 0  \\
CMC (table, $R \in \{7, 32\}$) & +1  & +2 & 0 / +2  & +1  & +4  & +46  & 0  \\
\bottomrule
\end{tabular}%
}
\end{table}} 
To quantify deployment cost, we implement a TC and two variants of the CMC in P4 on a Tofino switch and compare the resources reported by the compiler with those of a baseline forwarding program without any marking or TC logic. Table~\ref{tab:complexity} lists the incremental cost of
each approach in terms of the different resources.
\ROne{
\subsubsection{Metrics} 
The Intel Tofino 1 architecture employs a folded pipeline consisting of 12 physical Match-Action Unit (MAU) stages per hardware pipe, shared between ingress and egress processing.
The number of stages a program occupies is set by its longest chain of
dependent operations. 
Each stage provides SRAM, which stores exact-match tables and stateful
registers. 
Each stage also provides TCAM, which compares a key against all entries
at once and supports wildcards, as required for range matching or
longest-prefix matching in IP routing; it is considerably more costly
than SRAM and therefore scarcer.
Gateways evaluate the conditions of \textit{if} statements in the P4 program, and action slots
correspond to the very-long-instruction-word (VLIW) ALU operations that
modify header or metadata fields. The packet header vector (PHV) is the fixed
set of containers that carries parsed headers and metadata through the
pipeline; every field a program reads or writes must occupy PHV space.
Finally, bridged metadata is data that the ingress pipeline prepends to the
packet so that the egress pipeline can access it; it traverses the traffic
manager with the packet and therefore consumes buffer space and internal
bandwidth for every packet.}

\subsubsection{Transparent clock} 
The TC computes the residence time of each PTP event message (Sync or Delay Request) as the difference between its egress timestamp and ingress timestamp. 
The ingress timestamp is truncated to 32 bits and
carried to the egress pipeline as bridged metadata (4 bytes), which also
accounts for the additional ingress stage. 
In egress, the subtraction is implemented as a sequence of dependent operations (truncation, two's
complement, increment, and addition) before the result can be added to the
64-bit \textit{correctionField}. 
Since Tofino ALUs operate on containers of at
most 32 bits, each of these steps depends on the previous one, which leads to
the long stage chain. 
Together with the hash computations and register operations described below, this accounts for the 25 additional action slots.
The TC stores the residence time of each PTP event message in a
register and adds it to the \textit{correctionField} of the matching
Follow\_Up or Delay Response message when it later traverses the switch, requiring it to be stateful. 
To support concurrent PTP sessions, the register is indexed by a hash over the
domain, source port identity, and \textit{sequenceId} of the message, and each
entry stores a tag from a second hash, so that colliding entries are detected rather than applied to the wrong message. With 4,096 entries per register, this state occupies 8 SRAM blocks. 
Since the number of entries in flight grows with the number of sessions and their message rates, the register must be sized well above that number to keep collisions rare. 
Thus, the memory footprint of the TC grows with the deployment size. 
\RTwo{Statefulness has a further cost on switches with several independent
pipelines, each serving a subset of the ports~\cite{verdi2021heavy,
chiesa2023multipipe}. The TC must match each Delay Response to the residence
time of the corresponding Delay Request, which was recorded in the pipeline
of a different egress port; when the two ports belong to different
pipelines, this state must be shared across them.}
\ROne{
Distinguishing the four message types and checking the stored tags accounts for the 20 additional gateways, and the parsed timestamps, residence time, correction fields, port identities used as
hash input, and the computed indices and tags add 542 bits of PHV. }

In return, the TC is highly accurate: on our testbed, heavy TCP cross-traffic raises the
standard deviation of the TC-corrected delay from 5.7\si{ns} to only 39.2\si{ns},
whereas the marking error of the CMC is considerably larger (See Table~\ref{tab:worst-case_marking_acc}).

\subsubsection{Congestion marking} 
The CMC, in contrast, operates only in the
egress pipeline. It uses the metadata field \textit{enq\_qdepth}, which
reports the egress queue depth in 80-byte cells at enqueue time, and needs no
timestamps, no ingress processing, and no stored state, since the marking
travels in the event message itself. In the bit-shift variant, the threshold
is restricted to the cell size times a power of two,
$K \in \{80 \cdot 2^n \mid n \in \mathbb{Z}\}$, so that the congestion level
$r$ is obtained by right-shifting \textit{enq\_qdepth} by $n$ bits and added
to the marking counter. This single operation fits into the egress stage
already occupied by the baseline: it requires no additional stages, memory,
or bridged metadata, and only one additional gateway and action slot. 
Where arbitrary threshold values are required, the queue depth is mapped to $r$ by a range-match table
with a single parameterized action that carries $r$ as per-entry action data.
Range matching requires TCAM (2 blocks), and the action data occupies one
SRAM block. Since the counter update depends on the lookup result, this
variant needs two additional egress stages. Its cost is independent of the
number of thresholds: for $R=7$ and $R=32$, all metrics in
Table~\ref{tab:complexity} are identical, as additional thresholds only add
table entries, which fit within the same memory blocks. 
\RTwo{Since the CMC keeps no state and marks each message in the egress
pipeline of its own port, it is unaffected by the pipeline structure.}

\RThree{\textbf{Takeaway:} The bit-shift CMC adds virtually no cost over plain forwarding, and even the table-based CMC requires a fraction of the TC's resources and no state. The CMC thus trades the TC's nanosecond-level
accuracy for a much smaller resource footprint.}

\subsection{Deployment on Non-Programmable Switches}
\label{subsec:broader_adaptation}
\RTwo{We use the Tofino only as a proof of concept; the CMC itself does not require
a programmable data plane. Its simplest form is available on commodity
switches today: with ECN configured as a deterministic step function of
the queue length (Sec.~\ref{subsubsec:ECN}), synchronization packets sent as ECT(0)
are marked whenever the queue exceeds $K$ bytes, as our measurements on
the EdgeCore switch show (Sec.~\ref{subsec:measurement_of_marking_accuracy}, Table~\ref{tab:worst-case_marking_acc}). This requires only
configuration in the switches and software changes at the endpoints, and
corresponds to the case $R = 1$, which already reduces the time error by
about 50\% in our closed-loop evaluation (Sec.~\ref{subsec:closed_loop}).}

\RTwo{Multi-level marking requires a small extension of existing switch
hardware. To set the ECN bits, an ASIC already compares the queue depth
of each packet at enqueue time against a threshold. The CMC generalizes
this comparison to $R$ thresholds and adds the result to a header field instead of setting a bit. 
Writing per-hop queue information into packets is already
supported by commodity switch hardware: in-band telemetry schemes such as
INT~\cite{int} and IFA are implemented in merchant silicon such as
Broadcom's Trident~4~\cite{trident4} and in switches such as Juniper's
QFX5120~\cite{juniperifa} and Cisco's Nexus~9000~\cite{ciscoint}, which
insert per-hop metadata such as queue occupancy into packet headers.}  
\section{Analysis}
\label{sec:analysis}
\RThree{To understand when and how much congestion marking improves
synchronization, we analyze its effect on the MSE of the offset
estimation error for general queuing delay distributions. This section
answers three questions. 
First, under which conditions does congestion marking improve the MSE over a single hop, and how does the improvement scale with the number of thresholds (Sec.~\ref{subsec:improving_mse_over_one_hop})? 
Second, under which conditions does this improvement extend to a path of
multiple, possibly correlated hops (Sec.~\ref{subsec:improving_mse_over_many_hops})? 
Third, since the marking counter is shared along the path, how can the expected MSE be computed for a given choice of parameters (Sec.~\ref{subsec:markovian_model}), and how should these parameters be chosen (Sec.~\ref{subsec:choosign_optimal_parameters})?}

Since the network topology and traffic are not known a priori, we analyze
congestion marking for general queuing delay distributions rather than a
specific model such as $M/M/1$. This allows us to identify the conditions
and parameter regimes under which marking improves the MSE of the offset
estimation error.

\RTwo{We assume that the request and response packets take the same path i.e., the same links and switches, in opposite directions, hence $L_{\text{Req}}^{a}=L_{\text{Resp}}^{s-a}$, where $s$ is the number of hops. 
Within a switch, the two
directions may be processed by different internal pipelines, which only
affects the constant processing delay, and thus the base delay, but not
the queuing delay at the egress ports. Any resulting static asymmetry can
be compensated as discussed in Sec.~\ref{subsec:queuing_delay_corr}; otherwise, it adds the same
constant offset to both the uncorrected and the corrected estimate and
therefore does not change the improvement achieved by our method.}
\RThree{While Sec.~\ref{sec:offset_estimation} describes the delays and errors of a single message, we now
model them across all messages. Let $\mathcal{Q}_{m,l}, \mathcal{D}_{m,l}
\in [0, \infty)$ be random variables denoting the queuing delay of message
$m$ at hop $l$ and its estimation error after applying congestion marking,
respectively, so that $\delta^q_{m,l}$ and $\delta^E_{m,l}$ from
Sec.~\ref{subsec:correction_of_marked_packets} are their realizations. We do not restrict ourselves to
continuous distributions. In practice, $\mathcal{Q}_{m,l}$ is zero with a
non-zero probability $1 - \rho_l$, where $\rho_l$ is the utilization of
link $l$, and follows a continuous distribution otherwise. Since our
framework is completely general, we do not subdivide the distribution
into these two parts. The total queuing delay and its estimation error
are given by the random variables
$\mathcal{Q}_m = \sum_{l \in \mathcal{L}_m} \mathcal{Q}_{m,l}$ and
$\mathcal{D}_m = \sum_{l \in \mathcal{L}_m} \mathcal{D}_{m,l}$.}

We consider the \ac{MSE} of the offset estimation error $\epsilon$ or $\epsilon_{\mathrm{comp}}$ as our metric of interest in line with other works evaluating \ac{PTP} and \ac{NTP}~\cite{levesqueDelayAssymetry, Murakami,lee}. 
The expected \ac{MSE} of $\epsilon_{\mathrm{comp}}$ is, 
\begin{equation}
    \mathcal{E}_{\mathrm{comp}} = \E\left[\frac{(\mathcal{D}_{\text{Req}}-\mathcal{D}_{\text{Resp}})^{2}}{4}\right]
    \label{eq:modified_expected_errror}.
\end{equation}

The same can be applied to the distribution of total queuing delay to obtain the \ac{MSE} without correction for any given value,
\begin{equation}
    \mathcal{E} = \E\left[\frac{(\mathcal{Q}_{\text{Req}}-\mathcal{Q}_{\text{Resp}})^{2}}{4}\right]
    \label{eq:rmse_no_correction}.
\end{equation}

\ROne{We assume that the queuing delays of the request and response messages are
independent, which we refer to as Assumption~1 (\textbf{A1})}. 
The MSE can then be
decomposed into, 
\begin{equation}
\mathcal{E}
= \frac{1}{4}
\left\{
\underbrace{\mathrm{Var}(\mathcal{Q}_{\text{Req}})+\mathrm{Var}(\mathcal{Q}_{\text{Resp}})}_{\text{sum of variance}}
\;+\;
\underbrace{\left(\E[\mathcal{Q}_{\text{Req}}-\mathcal{Q}_{\text{Resp}}]\right)^{2}}_{\text{bias}}
\right\}
.
\label{eq:decomposed_mse}
\end{equation}
The first term indicates that even if the mean queuing delay in the forward and reverse paths is the same (under equal load), their individual variances still contribute to the \ac{MSE}. 
The second term (bias) indicates that the \ac{MSE} is higher under an asymmetric load in the network, i.e., when the total mean queuing delay in one direction exceeds that in the other. 
\ROne{\textbf{A1} holds even when flows in the two directions are coupled, for instance
by TCP ACKs. The request and response messages traverse different egress queues at each switch. 
In the worst case, a receiver returns one ACK of about 86 bytes for every two data frames of about 1518 bytes. 
The resulting reverse-path load is at most about 3\% of the forward data rate, which does not yield any meaningful correlation to affect Eq.~\ref{eq:decomposed_mse}}. 
As mentioned in Sec.~\ref{subsec:queuing_delay_corr}, 
Cristian's algorithm and all subsequent filtering algorithms presented in the literature assume that the forward and reverse delays have the same statistics, therefore reducing the bias term to zero.
Consequently, the filtering techniques presented so far demonstrate a reduction in the variance~\cite{kalmanFiltering, kalmanfilteringreview, KalmanFiltering2, genghuygens, variance1}.
Throughout the rest of the paper, we demonstrate conditions and proofs with and without the assumption of equal means, which we refer to as Assumption 2 (\textbf{A2}). 
In the next subsection, we show that congestion marking improves both the variance and the bias term for any general queuing distribution.

\subsection{MSE Improvement over a Single Hop}
\label{subsec:improving_mse_over_one_hop}

The queuing delay at a switch depends on the load, packet sizes, and
burstiness of the interfering traffic, and its distribution typically has
a long tail~\cite{Whitt1990HeavyTail, Garetto2003LongTail, Liebeherr2009DelayTail}.
Since the MSE squares the error, this tail contributes most of it, which
makes it relatively easy to choose a threshold that improves the MSE.
Proposition~\ref{prop:conditions} states sufficient conditions.

\begin{proposition}[Conditions for MSE Improvement over one hop]
\label{prop:conditions}
Consider the per-hop queuing distributions $\mathcal{Q}_{\text{Req},l}$ and $\mathcal{Q}_{\text{Resp},l}$ and a threshold $\Delay{K}{*}$ and number of thresholds $R$. The  MSE over a hop is improved if all of the following conditions hold:

\begin{enumerate}
    \item \textbf{Variance Improvement Region (C1):}  
    The threshold lies in the region
    \[
    \Delay{K}{*} > \frac{2\E[\mathcal{Q}_{m,l}]}{1 + (2R - 1)\P(\mathcal{Q}_{m,l} > R\Delay{K}{*})},
    \]
    referred to as the \ac{IR}.

    \item \textbf{Stochastic Dominance (C2):}  The distribution with the larger mean also has a larger CCDF within the \ac{IR}, i.e., if
    \[
    \E[\mathcal{Q}_{\text{Req},l}] > \E[\mathcal{Q}_{\text{Resp},l}],
    \]
    then
    \[
       \P(\mathcal{Q}_{\text{Req},l} > r\Delay{K}{*}) > \P(\mathcal{Q}_{\text{Resp},l} > r\Delay{K}{*}),
   \quad r = 1,\dots,R.
    \]
    \item \textbf{Upper Bound on the Threshold (C3):}  
    The threshold satisfies
\[
    \Delay{K}{*}
    <
    2 \cdot
    \frac{\E[\mathcal{Q}_{\text{Req},l}] - \E[\mathcal{Q}_{\text{Resp},l}]}
    {\sum_{r=1}^{R} \left(
    \P(\mathcal{Q}_{\text{Req},l} > r\Delay{K}{*})
    -
    \P(\mathcal{Q}_{\text{Resp},l} > r\Delay{K}{*})
    \right)}.
\]
\end{enumerate}
\end{proposition}

\begin{proof}
The proofs are presented in  Appendix~\ref{sec:conditions_for_improvement}.
\end{proof}
\textbf{C1} alone guarantees a reduction of the variance and therefore
suffices under \textbf{A2}; \textbf{C2} and \textbf{C3} additionally
ensure that the bias does not increase.
Although these conditions may appear restrictive, they leave considerable
freedom in practice. Since the denominator in \textbf{C1} is at least one,
the \ac{IR} always includes all thresholds above twice the mean queuing
delay, and it extends further for long-tailed distributions; for an
exponential distribution, it covers the entire range. Moreover, the
\ac{IR} grows with $R$.

We first verify these conditions on the queuing delays measured on our
testbed (Sec.~\ref{subsec:measurement_of_marking_accuracy}) with cross-traffic of 1, 4, 4.4, and 9.2\,Gbit/s
(Fig.~\ref{fig:ccdf_plot}). The dotted lines show the right-hand side of \textbf{C1}
for $R = 1$, the most restrictive case, so the \ac{IR} extends to the right of
their intersection with the CCDFs. It covers 95, 81, 78, and 65\% of the
delay range, respectively, and grows further with $R$
(Fig.~\ref{fig:tail_improvement}). All pairs of distributions show stochastic
dominance (\textbf{C2}), even the close 4G and 4.4G datasets.
Fig.~\ref{fig:d_upper_bound} shows the ratio of the upper bound in \textbf{C3} to
$\Delay{K}{*}$ for a pair of close (4G, 4.4G) and a pair of distant (1G,
9.2G) distributions. The ratio never falls below 2.7, so \textbf{C3} is never
violated.

\begin{figure}
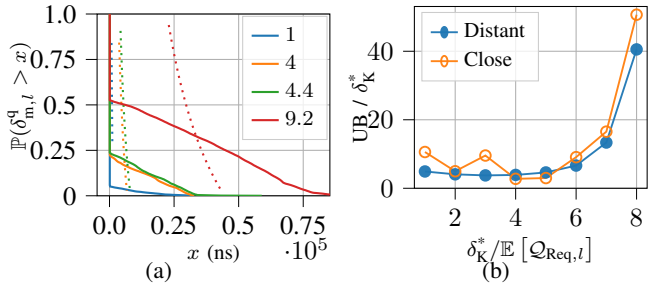

    \centering
    \begin{subfigure}{0.48\columnwidth}
        \centering

          \vspace{-0.8cm}
        \caption{}
        \label{fig:ccdf_plot}
    \end{subfigure}
    \hfill
    \begin{subfigure}{0.48\columnwidth}
        \centering
%
         \vspace{-0.8cm}
        \caption{}
        \label{fig:d_upper_bound}
    \end{subfigure}
    \vspace{-0.2cm}
    \caption{\small{(a) The CCDF of the queuing distribution for four different cross-traffic rates on a 10G Link. The queuing distributions exhibit a long tail, extending from where the dotted line crosses the solid one. (b) The headroom from the upper bound to  the selected threshold for two pairs of distributions from Fig.(a). The UB is at least 2.7 times the selected threshold and usually much larger.}}
    \label{fig:one_hop_mse_conditions}
\end{figure}

To verify the conditions on a wider variety of topologies and traffic, we
use the Sim2HW dataset~\cite{sim2hw}, which was originally produced to
predict end-to-end delay distributions using \ac{GNN}. It comprises 100
simulated networks with diverse traffic and link utilizations of up to
99.5\%.\footnote{The networks have 11.78 nodes on average, an average node
degree of 2.38 with a maximum of 9, roughly 8 switches, and an average
link utilization of 54\%. We rerun the OMNET++ simulations to record the
per-hop queuing delays over 60~seconds and consider all 503 distributions
with at least 100{,}000 packets, with 408{,}000 packets on average.}
Stochastic dominance (\textbf{C2}) holds for 989 of 990 pairs of
distributions in the same network, and for all pairs at the same switch.
The \ac{IR} covers between 66\% and 99.99\% of the delay range, with an
average of 95\%.

\begin{corollary}[Improvement with the Number of Thresholds]
If \textbf{C1} holds, increasing the number of thresholds $R$ strictly
reduces the variance over a single hop.
\end{corollary}
\begin{proof}
The proofs are presented in App.~\ref{sec:app_improvement_results}.
\end{proof}
Fig.~\ref{fig:variance_improvement} quantifies Corollary 1 on our testbed, with marking performed by our Tofino implementation (Sec.~\ref{subsec:impl_complexity}). The results therefore include
the inaccuracy of the mapping between marks and delay (Sec.~\ref{subsec:measurement_of_marking_accuracy}).
\footnote{\ROne{Unless stated otherwise, error bars in this paper denote
95\% confidence intervals based on the Student t-distribution over the
independent runs.}}
The threshold is simply set to the range of delays divided by $R$. Even with
$R = 1$, the delay variance decreases by 18--62\%, and with $R = 8$ by
80--95\%.

\RThree{\textbf{Takeaway:} The conditions for an improvement over a single hop
are loose in practice: the IR covers most of the delay range in both
measured and simulated data and grows with $R$, and stochastic dominance
holds almost without exception. On hardware, even a single, simply
chosen threshold reduces the delay variance by 18--62\%.}

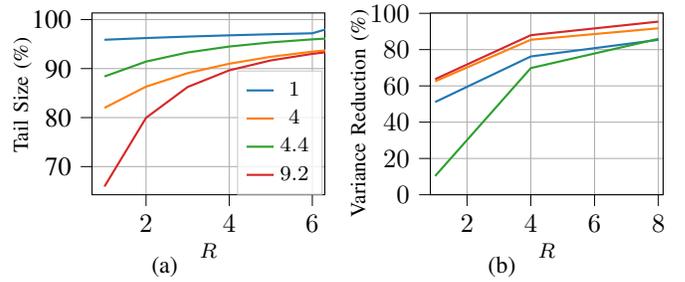
\begin{figure}
    \centering
    \begin{subfigure}{0.48\columnwidth}
        \centering
\begin{tikzpicture}

\definecolor{crimson}{RGB}{214,39,40}
\definecolor{darkgray}{RGB}{176,176,176}
\definecolor{darkorange}{RGB}{255,127,14}
\definecolor{forestgreen}{RGB}{44,160,44}
\definecolor{lightgray}{RGB}{204,204,204}
\definecolor{steelblue}{RGB}{31,119,180}
\begin{axis}[
width=1.1\columnwidth,
height=4cm,
legend style={
    font=\footnotesize,            %
    fill opacity=0.8,
    draw opacity=1,
    text opacity=1,
    at={(0.99,0.00)},
    anchor=south east,
    draw=lightgray
  },
  legend image post style={
      xscale=0.6,
  },
tick align=outside,
tick pos=left,
x grid style={darkgray},
xlabel={\footnotesize{$R$}},
xmajorgrids,
xmin=0.7, xmax=6.3,
xtick style={color=black},
y grid style={darkgray},
ylabel={\footnotesize{Tail Size (\%)}},
ylabel style={at={(axis description cs:0.1,0.55)}, anchor=center},
ymajorgrids,
ymin=0.642550121445524, ymax=1.01347394237526,
ytick style={color=black},
ytick={0.70,0.80,0.90,1.00},
yticklabels={70,80,90,100}
]
\addplot [thick, steelblue]
table {%
1 0.958817079365104
2 0.96231457325173
3 0.965287912366007
4 0.967850619946143
5 0.970065911286593
6 0.972000191005146
7 0.996613768696638
};
\addlegendentry{$1$}
\addplot [thick, darkorange]
table {%
1 0.819809286707207
2 0.863089648024121
3 0.890654154161841
4 0.909904084664167
5 0.924045746564294
6 0.934556492882356
7 0.942423365713895
};
\addlegendentry{$4$}
\addplot [thick, forestgreen]
table {%
1 0.883862872087595
2 0.914243741645844
3 0.932859720760884
4 0.944972711319517
5 0.953444465443224
6 0.959645448445621
7 0.964415415318294
};
\addlegendentry{$4.4$}
\addplot [thick, crimson]
table {%
1 0.659410295124149
2 0.799516386288881
3 0.862358881969153
4 0.896354962425361
5 0.916644599464274
6 0.930042309049294
7 0.939789344820155
};
\addlegendentry{$9.2$}
\end{axis}

\end{tikzpicture}
          \vspace{-0.8cm}
        \caption{}
        \label{fig:tail_improvement}
    \end{subfigure}
    \hfill
    \begin{subfigure}{0.48\columnwidth}
        \centering
\begin{tikzpicture}

\definecolor{crimson}{RGB}{214,39,40}
\definecolor{darkgray}{RGB}{176,176,176}
\definecolor{darkorange}{RGB}{255,127,14}
\definecolor{forestgreen}{RGB}{44,160,44}
\definecolor{lightgray}{RGB}{204,204,204}
\definecolor{steelblue}{RGB}{31,119,180}

\begin{axis}[
width=1.1\columnwidth,
height=4cm,
tick align=outside,
tick pos=left,
x grid style={darkgray},
xlabel={\footnotesize{$R$}},
xmajorgrids,
xmin=0.85, xmax=8.15,
xtick style={color=black},
y grid style={darkgray},
ylabel={\footnotesize{Variance Reduction (\%)}},
ylabel style={at={(axis description cs:0.095,0.45)}, anchor=center},
ymajorgrids,
ymin=12, ymax=100.193514571979,
ytick style={color=black}
]
\addplot [thick, steelblue, error bars/.cd, y dir=both, y explicit, error bar style={steelblue, solid}, error mark options={steelblue, solid, mark size=1.5pt, line width=0.6pt, rotate=90}]
table [x index=0, y index=1, y error plus index=2, y error minus index=3] {%
1 51.101463513445545 5.664930793600625 4.615708385302192
4 76.21606989560267 3.648710938196345 2.8541803817655733
8 85.37689748068283 2.29296635800992 1.8412903332221902
};
\addplot [thick, darkorange, error bars/.cd, y dir=both, y explicit, error bar style={darkorange, solid}, error mark options={darkorange, solid, mark size=1.5pt, line width=0.6pt, rotate=90}]
table [x index=0, y index=1, y error plus index=2, y error minus index=3] {%
1 62.502070801259116 4.644529196752053 3.8573495850987243
4 85.41673743628863 2.0833345732765736 1.7075329136517041
8 91.74006193567116 1.1807441018696636 0.9573712828521082
};
\addplot [thick, forestgreen, error bars/.cd, y dir=both, y explicit, error bar style={forestgreen, solid}, error mark options={forestgreen, solid, mark size=1.5pt, line width=0.6pt, rotate=90}]
table [x index=0, y index=1, y error plus index=2, y error minus index=3] {%
1 18.318323864343185 12.1182522744384027 3.048331179546565
4 69.80725362870609 2.7305567708102103 6.642173094727582
8 85.91939914300615 1.6243143281399313 3.705724110547635
};
\addplot [thick, crimson, error bars/.cd, y dir=both, y explicit, error bar style={crimson, solid}, error mark options={crimson, solid, mark size=1.5pt, line width=0.6pt, rotate=90}]
table [x index=0, y index=1, y error plus index=2, y error minus index=3] {%
1 63.716290696973594 2.8855578523203533 3.37981729013935
4 87.95087519395193 0.8082066478991123 0.8574288108119248
8 95.42239483045643 0.26565983966139584 0.29442238418303646
};
\end{axis}

\end{tikzpicture}
         \vspace{-0.8cm}
        \caption{}
        \label{fig:variance_improvement}
    \end{subfigure}
    \vspace{-0.2cm}
    \caption{\small{Improvement with $R$: (a) The proportion of the IR as a percentage of the total range of data increases with $R$, allowing more space for the user to select the threshold. (b) \ROne{The variance decreases as $R$ increases: it is reduced by 18 to 62\% with a single threshold and by 80 to 95\% with $R=8$ even without optimally selecting the thresholds.}}}
    \label{fig:improvement_with_R}
\end{figure}

\subsection{MSE Improvement over Multiple Hops}
\label{subsec:improving_mse_over_many_hops}

A synchronization packet may encounter congestion at several switches,
and its end-to-end queuing delay accumulates over the hops
(Fig.~\ref{fig:multi_hop_example}). By Eq.~\ref{eq:decomposed_mse}, the end-to-end MSE improves if the sum
of the variances and the bias decreases. 
\begin{figure}
    \centering
    \includegraphics[width=\linewidth]{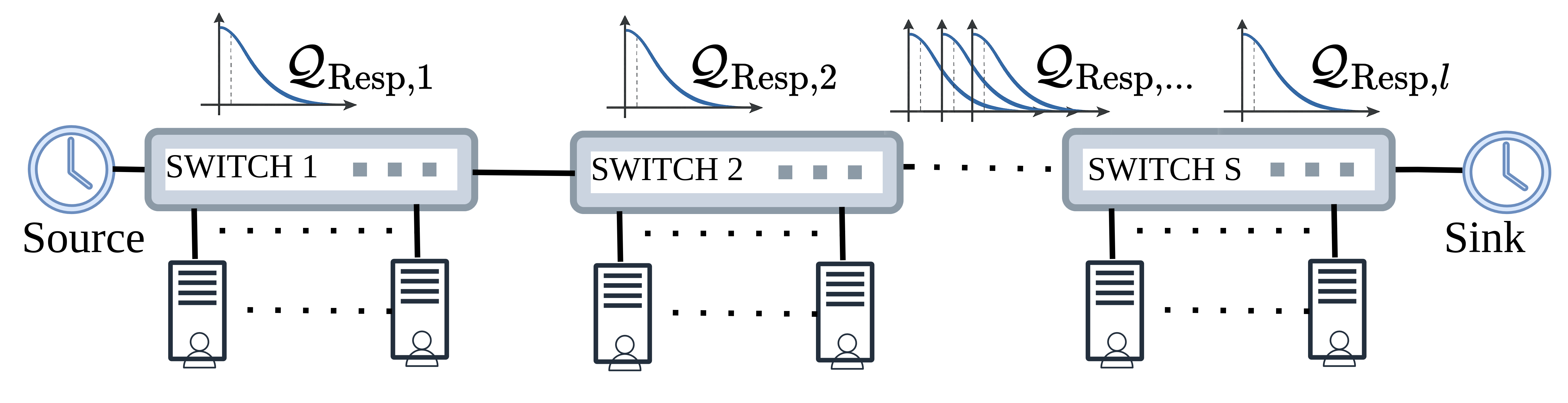}
    \vspace{-0.4cm}
    \caption{\small{The end-to-end queuing delay between source and sink is the sum of the per-hop queuing delays.}}
    \label{fig:multi_hop_example}
\end{figure}
\ROne{Next, we consider multiple hops, whose queuing delays may be
correlated. Proposition~\ref{cor:mse_improvement} states sufficient conditions under which
marking improves the end-to-end MSE.}
\begin{proposition}[MSE Improvement over Multiple Hops]
\label{cor:mse_improvement}
Under \textbf{A1}, consider a path of $s$ hops, where the queuing delays
$\mathcal{Q}_{\text{Req},l}$ and $\mathcal{Q}_{\text{Resp},l}$ at hop $l$
have finite variances and may be correlated across hops. Hop $l$ marks
with $R_l$ thresholds ($R_l = 0$ if it does not mark), and the counter
does not saturate, $\sum_l R_l \le N$. Without loss of generality, let
$\E[\mathcal{Q}_{\text{Req}} - \mathcal{Q}_{\text{Resp}}] \ge 0$. The
end-to-end MSE improves if all of the following hold:
\begin{enumerate}
  \item \textbf{C1} holds at every marking hop, for both directions, so
  that the variance of each hop decreases;
  \item \textbf{C4} (bias):
  \begin{equation*}
  0 \le \sum_{l=1}^{s} c_l \le 2\,\E[\mathcal{Q}_{\text{Req}} - \mathcal{Q}_{\text{Resp}}],
  \end{equation*}
  where
  $c_l = \Delay{K}{*} \sum_{r=1}^{R_l} \big[\P(\mathcal{Q}_{\text{Req},l} > r\Delay{K}{*})
  - \P(\mathcal{Q}_{\text{Resp},l} > r\Delay{K}{*})\big]$,
  so that the squared bias does not increase;
  \item \textbf{C5} (correlation): for both $m \in \gk{\text{Req}, \text{Resp}}$,
  \begin{multline*}
  2 \sum_{l \neq l'} \sum_{k=1}^{R_{l'}} p^{(l')}_{m,k}
  \big(\E[\mathcal{Q}_{m,l} \mid \mathcal{Q}_{m,l'} > k\Delay{K}{*}]
  - \E[\mathcal{Q}_{m,l}]\big) \\
  \ge \Delay{K}{*} \sum_{l \neq l'} \sum_{k=1}^{R_{l'}} \sum_{r=1}^{R_l}
  p^{(l')}_{m,k} \\ \cdot \big(\P(\mathcal{Q}_{m,l} > r\Delay{K}{*} \mid
  \mathcal{Q}_{m,l'} > k\Delay{K}{*}) - p^{(l)}_{m,r}\big),
  \end{multline*}
  where $p^{(l)}_{m,r} = \P(\mathcal{Q}_{m,l} > r\Delay{K}{*})$, so that
  the covariances between hops do not increase.
\end{enumerate}
\end{proposition}
\begin{proof}
The proofs are presented in App.~\ref{sec:app_multihop_general}.
\end{proof}

\ROne{To interpret \textbf{C4}, let $g_l = 2\,\E[\mathcal{Q}_{\text{Req},l} -
\mathcal{Q}_{\text{Resp},l}] - c_l$ denote the margin of hop $l$, so that
\textbf{C4} requires $\sum_l c_l \ge 0$ and $\sum_l g_l \ge 0$. We call a hop
\emph{good} if $c_l \ge 0$ and $g_l \ge 0$. A marking hop that is loaded
asymmetrically in the direction of the end-to-end bias is good if and
only if it satisfies \textbf{C2} and \textbf{C3}. If all hops are good \textbf{C4} holds. A hop
loaded in the opposite direction cannot be good, but \textbf{C4} still
holds as long as the good hops compensate for it. Moreover, since the
conditions of Proposition~2 are only sufficient, a hop that is not good
does not necessarily increase the MSE. If it satisfies \textbf{C2}, it
corrects its own asymmetry in the right direction, and since the MSE
comprises both the variance and the squared bias, a resulting change of
the end-to-end bias can be outweighed by the variance reduction that the
hop achieves under \textbf{C1}.}

\ROne{If the queuing delays at different hops are independent, which we refer to as Assumption 3 (\textbf{A3}), both sides of \textbf{C5} vanish,
so \textbf{C5} holds trivially. App.~\ref{sec:app_improvement_results} shows how \textbf{C4} and \textbf{C5} can be verified from measured delays.}
\ROne{To assess how restrictive these conditions are in practice, we
evaluate them on the Sim2HW networks. For every pair of nodes, we
consider the shortest path between them, which yields 6,259 paths of
1--9 hops. On each path, we send 20{,}000 PTP probes in each
direction at uniformly random times and record the exact queuing time at every hop, which preserves the correlations between hops. 
For each path, we sweep 20 log-spaced thresholds $\Delay{K}{*}$ and
$R \in \gk{1, 4, 8}$ with $N = 32$, and let a hop mark only if
\textbf{C1} holds there in both directions, which yields 336{,}534
settings in total. 
Given \textbf{C1} at every marking hop, \textbf{C4} is violated
in 8.8\% of the settings, mostly because the selected hop was not a \emph{good} hop. 
All \textbf{C4} violations happened when correction was applied to a single hop in the path. 
On the other hand, \textbf{C5} had 13.8\% violations. 
Both violations are small; moreover, the shortfall in \textbf{C5} amounts to a median of 0.1\% of the end-to-end delay variance showing little cross hop correlation. 
Consequently, the MSE of the offset estimation error increases in only 0.11\% of settings (370 of 336{,}534), with a median increase of 0.07\% and a maximum of 3.6\%. Since the conditions of Proposition~2 are only sufficient, violations rarely result in degradation in practice.}

\ROne{Proposition~\ref{cor:mse_improvement}
states conditions under which marking at any subset of hops improves the
end-to-end MSE. Our closed-loop evaluation (Sec.~\ref{subsec:closed_loop}) shows that, in
practice, the absolute reduction in RMSE grows with the number of marking hops.} Since the improvements add up over the hops and the
counter is backward compatible, switches can be upgraded incrementally.

\RThree{\textbf{Takeaway:} Over multiple hops, marking improves the MSE if the
corrections reduce the end-to-end bias in aggregate and the hops are not
negatively correlated. In the Sim2HW networks, these sufficient conditions are violated in 8.8\% (\textbf{C4}) and 13.8\% (\textbf{C5}) of the settings, yet the
MSE increases in only 0.11\% of them, and then by at most 3.6\%.}

\subsection{Markovian Model}
\label{subsec:markovian_model}

\ROne{Proposition~\ref{cor:mse_improvement} assumes that the counter never saturates, $\sum_l R_l \le N$.} In practice, a small
threshold may exhaust the counter within the first few hops.
Since the marking counter is shared along the path and limited to $N$,
the correction a message receives at a hop depends on the marks it has
already collected at previous hops. To capture this dependency, we model
the counter value $T_{m,l} \in \gk{0,\ldots,N}$ after hop $l$ as a
time-inhomogeneous Markov Chain~\cite{norris1998markov} with
$T_{m,0} = 0$. \ROne{To this end, we rely on \textbf{A3}. This corresponds to
Kleinrock's independence approximation~\cite{kleinrock1964}, which is
widely used to analyze networks of queues, since modeling the correlation
of queuing delays across hops is generally intractable.
Correlations
between hops could be incorporated by extending the state of the chain,
for instance by the queuing delay at the previous hop, but this would
require a model of these correlations, which is rarely available in
practice. Note that \textbf{A3} is only needed to compute the expected MSE;
Proposition~\ref{cor:mse_improvement} does not rely on it.}

\RThree{In principle, the end-to-end error distribution could be obtained by enumerating all possible counter trajectories along the path, but their
number grows exponentially with the number of hops. The Markov property
allows trajectories that reach the same counter value to be merged, which
makes the computation tractable (Appendix~\ref{sec:app_algorithm_implementation}).}
At hop $l+1$, at most $M = \min\gk{R, N - T_{m,l}}$ marks
can be added, and the number of marks added is
\begin{equation}
    r_{m,l+1} = \min\gk{M,\; \max\gk{0,\;
    \left\lceil \mathcal{Q}_{m,l+1}/\Delay{K}{*} \right\rceil - 1}},
    \label{eq:marks_added}
\end{equation}
i.e., $r$ marks for a queuing delay in $(r\Delay{K}{*}, (r+1)\Delay{K}{*}]$
and $M$ marks for any delay above $M\Delay{K}{*}$. The counter and the delay
estimation error then evolve as
\begin{equation} \label{eq:markov_model}
    T_{m,l+1} = T_{m,l} + r_{m,l+1}, \qquad
    \mathcal{D}_{m,l+1} = \mathcal{Q}_{m,l+1} - r_{m,l+1}\Delay{K}{*}.
\end{equation}
\ROne{Under \textbf{A3}, since $r_{m,l+1}$ depends only on $T_{m,l}$ and the queuing delay at hop
$l+1$, the counter forms a Markov chain whose transition probabilities
depend on the distribution of $\mathcal{Q}_{m,l+1}$ and thus on $l$.}
Eq.~\ref{eq:markov_model} encompasses three scenarios:
\begin{itemize}
    \item Counter not saturated $(T_{m,l} + R \le N)$: then $M = R$, and
    Eq.~\ref{eq:markov_model} reduces to Eq.~\ref{eq:performed_correction}.
    \item Counter saturates $(T_{m,l} < N < T_{m,l} + R)$: the delay is
    compensated only by the remaining $M = N - T_{m,l}$ marks, and a delay
    above $M\Delay{K}{*}$ sets the counter to $N$.
    \item Counter already saturated $(T_{m,l} = N)$: then $M = 0$, no further
    correction is possible, and $\mathcal{D}_{m,l+1} = \mathcal{Q}_{m,l+1}$.
\end{itemize}

Given the distributions of the queuing delays
$\gk{\mathcal{Q}_{m,l}}_{l\in\mathcal{L}_{m}}$ at every hop, we can
numerically compute the distribution of the counter and of the delay
estimation error at each hop. Summing the per-hop errors as in
Eq.~\ref{eq:delay_est_sum} then gives the end-to-end error
$\mathcal{D}_{m}=\sum_{l\in\mathcal{L}_{m}}\mathcal{D}_{m,l}$.

Algorithm~\ref{alg1} performs a forward probability propagation on the
finite-state, time-inhomogeneous Markov Chain and returns the distribution of the end-to-end delay estimation error $\mathcal{D}_m$ after applying
the correction.
\begin{algorithm}
\caption{Calculation of the distribution of $\mathcal{D}_{m}$}\label{alg1}
\begin{algorithmic}[1]
\State \textbf{Input}: Number of thresholds $R$, maximum counter value $N$, threshold $\Delay{K}{*}$, distributions of the queuing delays $\gk{\mathcal{Q}_{m,l}}_{l\in\mathcal{L}_{m}}$ along the path.
\State \textbf{Output}: Probability distribution of $\mathcal{D}_{m}$ after applying correction. 
\State Initialize tree $\mathfrak{T}$ with root labelled $0$, generation $0$
\Function{Ind}{$x, r, M$}
  \If{$r < M$}
    \State \Return $\chi\{x \in (r\Delay{K}{*}, (r+1)\Delay{K}{*}]\}$
      \Comment{for $r = 0$: $x \in [0, \Delay{K}{*}]$}
  \Else
    \State \Return $\chi\{x > M\Delay{K}{*}\}$
      \Comment{for $M = 0$: all $x$}
  \EndIf
\EndFunction
\For{$l$ in $\mathcal{L}_{m}$}
    \ForAll{leaves $p$ in generation $l-1$ with label $n$}
        \State \label{step:check_counter} $M \gets \min\{R, N - n\}$
        \For{$r = 0$ to $M$}
            \State \label{step:create_leaf} Create leaf $c$ with label $n+r$
            \State \label{step:leaf_assignment} $w_{p,c} \gets \P\big(\textsc{Ind}(\mathcal{Q}_{m,l}, r, M) = 1\big)$
            \State Assign weight $w_{p,c}$ to the edge $p, c$ 
        \EndFor
    \EndFor
\EndFor
\State $B \gets$ leaves in $\mathfrak{T}$
\State \label{step:final_sum} \Return $\sum_{c \in B} P_c \cdot \mathrm{Law}(\mathcal{D}_c)$
\end{algorithmic}
\end{algorithm}

Algorithm \ref{alg1} constructs a weighted tree $\mathfrak{T}$ to compute the distribution of the delay estimation error $\mathcal{D}_{m}$ along path $\mathcal{L}_{m}$.  
The tree depth corresponds to $\mid\mathcal{L}_{m}\mid$, where each generation corresponds to one hop. 
Some important steps of the Algorithm are:
\begin{itemize}
    \item \textbf{Step \ref{step:check_counter}}:At hop $l$, for every leaf $p$ from generation $l-1$ with counter value $n$, the number of admissible increments is $M=\min\{R, N - n\}$.
    \item \textbf{Step \ref{step:create_leaf}}: For each admissible increment $r\in\{0,...,M\}$, a new leaf is created.
    \item \textbf{Step \ref{step:leaf_assignment}}: Each leaf $c$ is then assigned a weight $w_{p,c}$ the probability that the queuing delay at hop $l$ falls into the interval corresponding to the increment $r$. The indicator function IND makes sure the correct $r$ is applied. 
    \item \textbf{Step \ref{step:final_sum}}: For each leaf $c$, $P_c$ is the product of the edge weights
along the path from the root to $c$, i.e., the probability of this
realization of counter increments, and $\mathcal{D}_c =
\sum_{i} (\mathcal{Q}_{m,i} - r_i \Delay{K}{*})$ is the resulting
end-to-end error, with each $\mathcal{Q}_{m,i}$ conditioned on the interval
that produced the increment $r_i$. Under \textbf{A3}, the output is the mixture of
these error distributions, weighted by the path probabilities.  
\end{itemize}
In short, the tree is the unfolded representation of the Markov chain over all the hops, and the final sum marginalizes over all feasible counter trajectories. 

\RThree{\textbf{Takeaway:} Since the counter is shared along the path, the best configuration
parameters depend on the entire path. Algorithm~\ref{alg1} computes the expected
MSE for any choice of $N$, $R$, and $\Delay{K}{*}$, which the operator can
use to select them.}

\subsection{Choosing the Parameters}
\label{subsec:choosign_optimal_parameters}

The number of bits, and thus $N$, is fixed by the header field
(Sec.~\ref{subsec:available_packet_headers}), which leaves $R$ and $\Delay{K}{*}$ as tunable parameters.
Given the per-hop queuing delay distributions, the operator computes the
MSE for all client-server pairs with Algorithm~\ref{alg1} and searches for the
parameters that best meet the objective, for instance minimizing the
average or maximum MSE, or maintaining a \ac{TUB}~\cite{sundial,yashinfocom,spanner}.
\RThree{This computation is performed by the operator outside the data path, not
on the switches or endpoints. It can be carried out once, based on
measured or simulated queuing delay distributions, or repeated
periodically as conditions change, for instance using the M/M/1 model
with the mean queuing delays reported by the switches or INT probes (Sec.~\ref{subsec:app_accuracy_mm1}).}

The delay distributions can be obtained in several ways: from queue
statistics exposed by switches, from simulations, from analytical models
such as exponential or gamma distributions, or from predictions by \ac{GNN}s,
which have been shown to predict delay distributions with good accuracy~\cite{sim2hw}.
How they map to
synchronization paths depends on the deployment: PTP and NTP typically
follow a spanning tree, while some implementations use multiple paths~\cite{multipathPTP,PTPsec}
or decentralized architectures~\cite{genghuygens,firefly}.
We therefore do not prescribe a single procedure; instead, the framework
in Algorithm~\ref{alg1} is intentionally flexible and allows different approaches to
parameter selection. Our analysis shows that the performance is robust to the
exact choice, as a broad range of parameters improves the MSE
(Sec.~\ref{subsec:improving_mse_over_one_hop}).

Two approaches known from ECN, whose threshold selection is still being
researched 25 years after its introduction~\cite{ramakrishnan2001addition, ZHUGE2026112151}, are natural here. In a planned deployment, the
delays are simulated or modeled, and the thresholds are optimized for the
given objective~\cite{ziegler2001quantitative, RED_Selection}. Alternatively, rules of
thumb can be derived for classes of networks~\cite{KaanTCP, virtulaisedTCP, dctcp}, which fits well with the profiles of
IEEE~1588 that tailor PTP parameters to domains such as telecom, power
utilities, industrial automation, or video broadcast. The organizations
proposing these profiles have the resources and expertise to model and
simulate the scenarios within a profile and to recommend such settings.
Sec.~\ref{subsec:app_accuracy_mm1} evaluates the model-based approach on a single hop and shows
that, although the resulting parameters are not optimal, they still yield
good results, especially for moderate link utilization.

Recent work on ECN also adapts
thresholds online from estimated link characteristics~\cite{ZHUGE2026112151}, sometimes using machine learning~\cite{cheng2024pet, jung2023rlecn, acc_paper, heirachichaltuning}.
Since synchronization packets are sent
far less frequently than TCP packets, far fewer observations are
available for such estimation in our case. We note the promise of online
estimation under this constraint and explore it in future work.

\RThree{\textbf{Takeaway:} The parameters can be selected from measured,
simulated, or modeled queuing delays, either once or periodically, and the
performance is robust to their exact choice.}

\section{Evaluation}
\label{sec:evaluation}
\RThree{We evaluate our method in four steps of increasing realism. Sec.~\ref{subsec:app_accuracy_mm1}
examines whether thresholds derived from a simple queuing model perform
well over a single hop. Sec.~\ref{subsec:improved_filtering} evaluates the interaction with the
RTT-based filters used in practice on our hardware testbed. Sec.~\ref{subsec:multihop}
extends the analysis to multiple hops and validates Algorithm~1 against
simulation. Finally, Sec.~\ref{subsec:closed_loop} closes the loop with drifting clocks and longer paths.}
\subsection{Accuracy of the Model-Based Method}
\label{subsec:app_accuracy_mm1}

We first examine how well a threshold can be selected from a simple
queuing model, using the single-hop topology of
Fig.~\ref{fig:omnet_setup} simulated in OMNeT++ with the INET
library. Simulation gives access to per-packet queuing delays, which are
difficult to obtain on hardware. We generate four cross-traffic patterns
of decreasing utilization (Table~\ref{tab:flow-patterns}): \ac{SF},
\ac{LM}, \ac{SM}, and \ac{SS}. Each consists of
packets with exponentially distributed interarrival times with mean
$\lambda$ and exponentially distributed sizes with mean $\mu$, so that the
egress queue approximates an $M/M/1$ queue. With 1\,Gbit/s links, a
packet of $B$ bytes has a service time of $8B$\,ns, which yields the
expected utilization $\rho$ and mean queuing delay $\lambda^{*}$ of the
model.

To isolate the offset estimation error, both clocks use ideal
oscillators, and the modified gPTP client stops short of adjusting its
clock and reports the computed offset, as in the hardware measurements
of Sec.~\ref{subsec:measurement_of_marking_accuracy}. Since the true offset is zero, the reported offset equals
the estimation error $\epsilon$. Each test collects roughly 800 samples
and is repeated with eight random seeds.

\begin{table}
    \centering
    \scriptsize
    \renewcommand{\arraystretch}{1.1}
    \setlength{\tabcolsep}{6pt}
    \begin{tabular}{|c||c|c|c|c|c|c|}
        \hline
        \textbf{Flow} & $\mu$ bytes & $\lambda$~(\unit{\micro\second})& $\rho$ & $\lambda^{*}$~(\unit{\micro\second}) & $\rho^{o}$ & $\lambda^{o}$~(\unit{\micro\second})  \\
        \hline\hline
        SF & 850  & 8    & 0.85 & 38.5 & 0.86 & 26.7 \\
        \hline
        LM & 1000 & 12   & 0.66 & 16 & 0.70 & 20.4  \\
        \hline
        SM & 750  & 12   & 0.50 & 6 & 0.54 & 12.5  \\
        \hline
        SS & 600  & 14   & 0.34 & 2.50 & 0.37 & 8.6  \\
        \hline
    \end{tabular}
    \caption{\small{Four types of flows are considered in the evaluation of the model-based approach, illustrating different levels of utilization of the link. We choose the $M/M/1$ model for comparison with the simulations. The last four columns compare the expected vs. observed utilization and mean queuing time.}}
    \label{tab:flow-patterns}
\end{table}
 
\subsubsection{Observed Queuing Behaviour}
\label{subsec:Observed_queuing_behaviour}

Table~\ref{tab:flow-patterns} compares the observed mean utilization $\rho^{o}$ and mean queuing delay $\lambda^{o}$ \ROne{pooled across the 8 runs} at the switch with the $M/M/1$ model. The utilization is slightly higher for all flows, and the mean queuing delay
is considerably higher except for SF, due to the padding of small packets
to the minimum frame size, the 54~bytes of Ethernet and PHY headers, and
the 12-byte inter-packet gap. For SF, the queue becomes transiently
overloaded and drops packets even with a very large buffer, which lowers
the observed delay.

\ROne{Fig.~\ref{fig:ks_test} shows how well an exponential distribution, as
assumed by the $M/M/1$ model, describes the collected queuing delays. We
use the Kolmogorov--Smirnov (KS) statistic $D$, i.e., the maximum distance
between the empirical CDF and an exponential CDF, only as a distance
measure and do not perform a goodness-of-fit test. 
Once again, the queuing distributions of all the 8 runs are pooled. 
A purely model-based approach, using the expected mean $\lambda^{*}$, is far from the observed
delays, as real switch queues deviate from the $M/M/1$ assumptions due to
the effects described above. An exponential distribution with the
observed mean $\lambda^{o}$, in contrast, is considerably closer. This
suggests a practical approach: although the full queuing distribution of
a switch port is often unavailable, many switches report the mean
queuing delay of their ports along with the utilization~\cite{lanz, rfc9197},
from which $\lambda^{o}$ and $\rho^{o}$ can be obtained.}

\subsubsection{MSE Improvement}
\label{subsubsec:eval_mse_improvement_model}

We now evaluate whether a threshold derived from such an exponential fit
also performs well. We measure the improvement by the reduction in the
root MSE, $\mathcal{I}= 1 - \sqrt{\frac{\mathcal{E}_{\mathrm{comp}}}{\mathcal{E}}}$,
which scales linearly rather than quadratically with the error.
Fig.~\ref{fig:simple_topology_Eval} compares two thresholds for each traffic pattern.
The empirical threshold (black circles) is the optimum obtained with Algorithm~\ref{alg1}
from the full measured queuing delay distribution. The model-based
threshold (red squares) is obtained with Algorithm~\ref{alg1} from an exponential
distribution with the observed mean $\lambda^{o}$. The red dashed line shows the
improvement that Algorithm~\ref{alg1} predicts for this threshold if the delays
were indeed exponentially distributed.

The model-based threshold performs close to the empirical optimum at
medium utilization (LM and SM) and somewhat worse at the two extremes,
for different reasons. For SF, packet drops cut off the tail of the delay
distribution, so the empirical optimum of about 33\,\si{\micro\second} lies far below
the model-based threshold of 53.7\,\si{\micro\second}. For SS, the observed delays
deviate most from the exponential distribution (Fig.~\ref{fig:ks_test}).
Nevertheless, the model-based threshold improves the RMS offset
estimation error by more than 20\% for all flows. An operator therefore
does not need the full queuing delay distribution: the mean queuing
delay reported by the switch suffices to select a threshold that
achieves most of the possible improvement. We will explore other model-based methods in future work.

\RThree{\textbf{Takeaway:} A purely model-based threshold selection does not
reflect real switch queues, but an exponential fit to the observed mean
queuing delay yields thresholds that improve the RMS offset estimation
error by more than 20\% for all traffic patterns, close to the empirical
optimum at moderate utilization.}

\begin{figure}
    \centering
\begin{tikzpicture}
\definecolor{green}{RGB}{0,128,0}
\begin{axis}[
    width=9cm,
    height=3.5cm,
    tick align=outside,
    label style={font=\footnotesize}, %
    xlabel={},
    tick label style={font=\footnotesize}, %
    xtick style={color=black},
    xtick={0,1,2,3},
    xmin=-0.15, xmax=3.15,
    xticklabel style={rotate=45.0},
    xticklabels={SF,LM,SM,SS},
    ylabel={KS statistic},
    ylabel style={at={(axis description cs:0.05,0.5)}, anchor=center, color=black},
    grid=both,
    legend style={
    at={(0.5,1.02)},
    anchor=south,
    legend columns=-1,
    font=\footnotesize,
    /tikz/every even column/.append style={column sep=0.5cm}
}
]
\addplot+[semithick, red, mark=square*, mark size=2, mark options={solid}, only marks] coordinates {(0,0.247251) (1,0.138956) (2,0.098950) (3,0.070921)};
\addlegendentry{$\lambda^{o}$}
\addplot+[semithick, green, mark=triangle*, mark size=3, mark options={solid}, only marks] coordinates {(0,0.196944) (1,0.234296) (2,0.366519) (3,0.500686)};
\addlegendentry{$\lambda^{*}$}
\end{axis}
\end{tikzpicture}
     \vspace{-0.6cm}
    \caption{\small{Measuring the distance to the empirical queuing delay distribution for a sample mean M/M/1 and theoretical mean, using the KS-statistic. It is apparent that using $\lambda^o$ produces a closer fit to the data. }}
    \label{fig:ks_test}
\end{figure}
\begin{figure}
    \centering
\begin{tikzpicture}

\definecolor{darkgray176}{RGB}{176,176,176}
\definecolor{green}{RGB}{0,128,0}
\begin{axis}[
width=9cm,
height=3.5cm,
tick align=outside,
label style={font=\footnotesize}, %
tick pos=left,
x grid style={darkgray176},
xlabel={},
xmajorgrids,
xmin=-0.15, xmax=3.15,
tick label style={font=\footnotesize}, %
xtick style={color=black},
xtick={0,1,2,3},
xticklabel style={rotate=45.0},
xticklabels={SF,LM,SM,SS},
y grid style={darkgray176},
ylabel={$\mathcal{I}$},
ymajorgrids,
ylabel style={at={(axis description cs:0.05,0.5)}, anchor=center, color=black},
ymin=0, ymax=0.55,
ytick style={color=black},
legend style={
    at={(0.5,1.02)},
    anchor=south,
    legend columns=-1,
    font=\footnotesize,
    /tikz/every even column/.append style={column sep=0.5cm}
}
]
\addplot [semithick, black, mark=*, mark size=2, mark options={solid}, only marks,
    error bars/.cd, y dir=both, y explicit, error bar style={black}, error mark options={rotate=90, mark size=2, black}]
table [x index=0, y index=1, y error index=2] {%
-0.05 0.405505066171 0.039543351594
0.95 0.423380411242 0.0518664892716
1.95 0.402320939802 0.0636451201467
2.95 0.388130708078 0.137325475874
};\addlegendentry{Empirical}
\addplot [semithick, red, mark=square*, mark size=2, mark options={solid}, only marks,
    error bars/.cd, y dir=both, y explicit, error bar style={red}, error mark options={rotate=90, mark size=2, red}]
table [x index=0, y index=1, y error index=2] {%
0.05 0.22771654445 0.0595303955095
1.05 0.394078916424 0.0517533789215
2.05 0.346347331417 0.0705620025452
3.05 0.277746685111 0.163011988781
}; \addlegendentry{Model ($\lambda^o$)}
\addplot [semithick, red, dashed, mark=none]
table {%
0 0.37677039197482565
1 0.3683806897320602
2 0.3639057041028628
3 0.36088964110797395
}; \addlegendentry{Predicted ($\lambda^o$)}
\end{axis}

\end{tikzpicture}
     \vspace{-0.3cm}
    \caption{\ROne{\small{Comparison of thresholds for a single hop. Empirical: optimum threshold
obtained from the full measured queuing delay distribution. Model
($\lambda^o$): threshold obtained from an exponential distribution with
the observed mean. Predicted ($\lambda^o$): improvement expected for the
model-based threshold if the delays were exponentially distributed. The
model-based threshold performs close to the optimum at medium
utilization and improves the RMSE by more than 20\% for all flows.}}}
    \label{fig:simple_topology_Eval}
\end{figure}

\subsection{Improved Filter Performance}
\label{subsec:improved_filtering}

So far, we have considered unfiltered delay measurements. In practice,
clock synchronization reduces their variance with \ac{RTT}-based filters, most
commonly the median delay~\cite{linuxptp,ptp}, the minimum RTT~\cite{mills1992rfc, chrony, firefly},
or a moving average~\cite{ptp}.
We define a filter operator $S_{M}(A,B)$ which takes $M$ samples of the forward and reverse delay distributions $A$ and $B$ and outputs a new delay value distribution $\mathcal{S}$ based on the corresponding \ac{RTT} samples $\frac{a_{i}+b_{i}}{2}$ where $a_{i}$ and $b_{i}$ are the $i$-th forward and reverse delay estimation measurements from $A$ and $B$.
Since the RTT combines both directions, such a filter cannot distinguish
whether excess delay occurred in the forward or the reverse
path~\cite{LimitationsClockSync}. We therefore apply the
filter to the corrected delays $\mathcal{D}_{m}$ rather than
$\mathcal{Q}_{m}$, since our method corrects each direction individually.
Increasing the filter length $M$ raises the chance of including samples
without queuing delay, but makes the synchronization react more slowly to
changes in the network or require more frequent messaging.

We use the Close and Distant datasets of Sec.~\ref{subsec:improving_mse_over_one_hop}. For each $R$, we
determine the optimal threshold with Algorithm~\ref{alg1}, round it to the nearest
multiple of the Tofino's 80-byte cell size, and repeat the hardware
measurements of Fig.~\ref{fig:measurement_setup} with the modified \textit{ptp4l} of Sec.~\ref{subsec:measurement_of_marking_accuracy},
recording the forward ($T2-T1$) and reverse ($T4-T3$) delays together with the marking counter. To mimic changing network
conditions, the cross-traffic alternates between 60\,s on and 20\,s off.
For each configuration, we conduct 12 runs.

We evaluate the median-delay and min-RTT filters, which together cover
the filters used in NTP, Chrony, PTP, and Firefly. With a filter length of
$M = 8$, marking reduces the variance of the filter output by up to 99\%
($R = 7$) compared to the same filter without correction, with a similar
dependence on $R$ as in Fig.~\ref{fig:variance_improvement} (Fig.~\ref{fig:filter_variance_reduction}). The
min-RTT filter consistently outperforms the median-delay filter.
Fig.~\ref{fig:filter_length} shows the variance as a function of $M$. Marking
reduces the variance for every filter length, and the corrected variance
without any filtering ($M = 1$) is lower than the uncorrected variance
with $M = 12$.

\RThree{\textbf{Takeaway:} Marking improves RTT-based filters by up to 99\% in
variance. Even unfiltered corrected samples outperform a twelve-sample
uncorrected filter, so shorter filters suffice and the synchronization
reacts faster to changes in the network.}
 
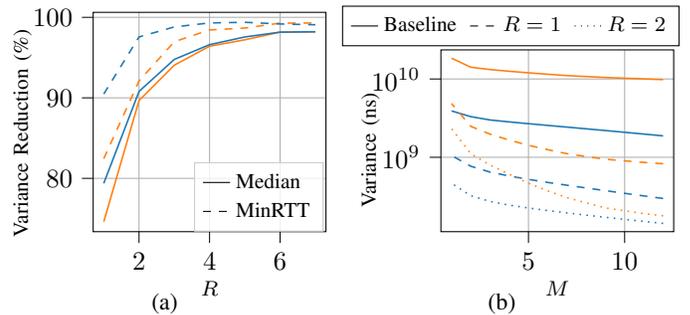
\begin{figure}
    \centering
    \begin{subfigure}{0.48\columnwidth}
        \centering
\begin{tikzpicture}

\definecolor{crimson2143940}{RGB}{214,39,40}
\definecolor{darkgray176}{RGB}{176,176,176}
\definecolor{darkorange25512714}{RGB}{255,127,14}
\definecolor{forestgreen4416044}{RGB}{44,160,44}
\definecolor{lightgray204}{RGB}{204,204,204}
\definecolor{steelblue31119180}{RGB}{31,119,180}

\begin{axis}[
width=1.1\columnwidth,
height=4.5cm,
legend cell align={left},
legend style={
    font=\footnotesize,            %
    fill opacity=0.8,
    draw opacity=1,
    text opacity=1,
    at={(0.99,0.00)},
    anchor=south east,
    draw=lightgray
  },
  legend image post style={
      xscale=0.6,
  },
tick align=outside,
tick pos=left,
x grid style={darkgray176},
xlabel={\footnotesize{$R$}},
xmajorgrids,
xmin=0.7, xmax=7.3,
xtick style={color=black},
y grid style={darkgray176},
ylabel={\footnotesize{Variance Reduction (\%)}},
ylabel style={at={(axis description cs:0.095,0.45)}, anchor=center},
ymajorgrids,
ymin=73.3744481232599, ymax=100.619105038162,
ytick style={color=black}
]
\addlegendimage{solid, black}
\addlegendentry{Median}
\addlegendimage{dashed, black}
\addlegendentry{MinRTT}
\addplot [semithick, darkorange25512714, error bars/.cd, y dir=both, y explicit, error bar style={darkorange25512714, solid}, error mark options={darkorange25512714, solid, mark size=1.5pt, line width=0.6pt, rotate=90}]
table [x index=0, y index=1, y error plus index=2, y error minus index=3] {%
1 74.6128416193919 2.756137653617614 3.263545171565397
2 89.6621236477004 1.063631554911538 1.014254084230629
3 94.0569924445478 0.8331877980581623 1.0137351373265204
4 96.4086668526731 0.5458215589094806 0.4137714237345591
5 97.2059483087807 0.43222062197250466 0.5969504525255189
6 98.1617756900935 0.3922185926870867 0.19299675453534348
7 98.1933227455797 0.15965882908312778 0.5875059870659953
};
\addplot [semithick,  steelblue31119180, error bars/.cd, y dir=both, y explicit, error bar style={steelblue31119180, solid}, error mark options={steelblue31119180, solid, mark size=1.5pt, line width=0.6pt, rotate=90}]
table [x index=0, y index=1, y error plus index=2, y error minus index=3] {%
1 79.3792830361005 1.0131308181857008 0.9948940577468761
2 90.7942811367089 0.497087356382508 0.5379525070643467
3 94.7715008621739 0.3041384516258745 0.2969952627209551
4 96.6039264434877 0.21896753076779873 0.19333687129419275
5 97.5601712683251 0.13059559218756078 0.16147864199804474
6 98.1585130551451 0.14679764800112594 0.12432553428180881
7 98.2105217097945 -0.1579229728359195 0.5029247353509874
};
\addplot [semithick,  darkorange25512714, dashed, error bars/.cd, y dir=both, y explicit, error bar style={darkorange25512714, solid}, error mark options={darkorange25512714, solid, mark size=1.5pt, line width=0.6pt, rotate=90}]
table [x index=0, y index=1, y error plus index=2, y error minus index=3] {%
1 82.4414618024699 2.579733540425181 1.3375552565400426
2 92.0978543885155 0.7020338089198503 0.9780186849184247
3 96.9716259325289 0.5722349073855781 0.43021737893259626
4 98.4311469172536 0.30337997780777926 0.2048058918028488
5 98.6681672798482 0.1569109222171221 0.34865741406130724
6 99.2775018667076 0.2824745099161845 0.09498957878439285
7 99.2959602168672 0.0836437446100291 0.2788843067861819
};
\addplot [semithick,  steelblue31119180, dashed, error bars/.cd, y dir=both, y explicit, error bar style={steelblue31119180, solid}, error mark options={steelblue31119180, solid, mark size=1.5pt, line width=0.6pt, rotate=90}]
table [x index=0, y index=1, y error plus index=2, y error minus index=3] {%
1 90.4764265163841 1.1192138397939573 0.6524601975687432
2 97.5619101320052 0.4129015177642543 0.2674799337723499
3 98.8074842388585 0.11571757855739406 0.12337414170869465
4 99.3032641282571 0.0741918164653157 0.0907161352114656
5 99.3807115420304 0.10646083251253913 0.07590254936964413
6 99.175409066232 0.10019916038243082 0.13246597344091526
7 99.0742451438064 0.0006766827675761533 0.3944061578393132
};
\end{axis}

\end{tikzpicture}
          \vspace{-0.8cm}
        \caption{}
        \label{fig:filter_variance_reduction}
    \end{subfigure}
    \hfill
    \begin{subfigure}{0.48\columnwidth}
        \centering
\begin{tikzpicture}

\definecolor{darkgray176}{RGB}{176,176,176}
\definecolor{darkorange25512714}{RGB}{255,127,14}
\definecolor{lightgray204}{RGB}{204,204,204}
\definecolor{steelblue31119180}{RGB}{31,119,180}

\begin{axis}[
width=1.1\columnwidth,
height=4cm,
 legend style={
    at={(0.3,1.02)},
    anchor=south,
    legend columns=-1,
    font=\footnotesize,
    /tikz/every even column/.append style={column sep=0.1cm}
},
  legend image post style={
      xscale=0.6,
  },
log basis y={10},
tick align=outside,
tick pos=left,
x grid style={darkgray176},
xlabel={\footnotesize{$M$}},
xmajorgrids,
xmin=0.45, xmax=12.55,
xtick style={color=black},
y grid style={darkgray176},
ylabel={\footnotesize{Variance ($\si{ns}^{2}$)}},
ymajorgrids,
ymin=112206564.77347, ymax=23511535619.9448,
ymode=log,
ytick style={color=black},
ylabel style={at={(axis description cs:0.095,0.35)}, anchor=center},
ytick={10000000,100000000,1000000000,10000000000,100000000000,1000000000000},
yticklabels={
  \(\displaystyle {10^{7}}\),
  \(\displaystyle {10^{8}}\),
  \(\displaystyle {10^{9}}\),
  \(\displaystyle {10^{10}}\),
  \(\displaystyle {10^{11}}\),
  \(\displaystyle {10^{12}}\)
}
]
\addlegendimage{solid, black}
\addlegendentry{Baseline}
\addlegendimage{dashed, black}
\addlegendentry{$R=1$}
\addlegendimage{dotted, black}
\addlegendentry{$R=2$}
\addplot [semithick, steelblue31119180,dashed, error bars/.cd, y dir=both, y explicit, error bar style={steelblue31119180, solid}, error mark options={steelblue31119180, solid, mark size=1.5pt, line width=0.6pt, rotate=90}]
table [x index=0, y index=1, y error plus index=2, y error minus index=3] {%
12 297909619.685403 34486824.72146803 56779072.250492156
11 321170879.646429 35815160.2964555 61137417.21328306
10 346615079.09891 38330826.804198205 64571734.440794945
9 374843756.094528 41644388.49471003 67608532.59846175
8 405642283.66881 46539975.11682761 73896797.82690042
7 439809038.832607 52624977.74768859 78495693.62996227
6 477435422.27481 57477053.94920391 85057733.00061566
5 522638165.10389 65254451.8135764 90360720.00577885
4 576753553.982931 77922758.49040997 96456354.0225836
3 647793685.221736 94100722.2774949 99068363.32185435
2 769079512.35854 113275009.25073922 108759249.21399069
1 1054566204.05391 143829267.7056178 124804021.37734544
};
\addplot [semithick, steelblue31119180, error bars/.cd, y dir=both, y explicit, error bar style={steelblue31119180, solid}, error mark options={steelblue31119180, solid, mark size=1.5pt, line width=0.6pt, rotate=90}]
table [x index=0, y index=1, y error plus index=2, y error minus index=3] {%
12 1877383536.03767 265804357.66364145 279503162.57369804
11 1979327869.69788 277979336.13575125 290565556.3380742
10 2084886084.41569 293303301.6539426 301876234.40778446
9 2195581712.15224 306525651.8593831 313819831.6327548
8 2310222888.95033 322212258.0426011 327665878.6107907
7 2429780355.20955 336468257.2400522 341632614.50513554
6 2553498309.93232 353754370.8124776 355179020.22857237
5 2686080887.22038 372282853.4542203 371597113.97030973
4 2829559001.34431 388413496.3607545 388817424.91170883
3 2997184785.94219 407244307.00163555 406351627.0286484
2 3305063846.79397 440675915.55477047 435271144.27921534
1 3912610989.40404 487919716.1789422 484978923.9729209
};
\addplot [semithick, darkorange25512714, dashed, error bars/.cd, y dir=both, y explicit, error bar style={darkorange25512714, solid}, error mark options={darkorange25512714, solid, mark size=1.5pt, line width=0.6pt, rotate=90}]
table [x index=0, y index=1, y error plus index=2, y error minus index=3] {%
12 827636139.054878 72676333.37913656 110695347.30444515
11 859868299.147045 73527629.40500557 113626090.87039232
10 901001830.572243 74565473.38516068 114259666.92931008
9 944621865.993527 77382149.4654851 116819341.40249121
8 1023521581.54992 79576531.01744807 119230099.70612013
7 1112157742.15479 78375493.52559876 120714394.3528943
6 1251074307.06207 83337464.43382978 122761044.3475442
5 1422422655.59755 89360663.91932702 122800809.17826772
4 1659453588.60725 96252730.2736218 126644033.76414227
3 1975926558.23101 107426614.8429544 131319635.3642323
2 2484668195.77032 127377468.45042515 128493060.27837467
1 4864011488.61141 207375938.8207445 191077475.37744713
};
\addplot [semithick, darkorange25512714, error bars/.cd, y dir=both, y explicit, error bar style={darkorange25512714, solid}, error mark options={darkorange25512714, solid, mark size=1.5pt, line width=0.6pt, rotate=90}]
table [x index=0, y index=1, y error plus index=2, y error minus index=3] {%
12 9817154363.45118 732275336.6466789 811975942.5018654
11 10016051296.5781 745587651.0513763 822271250.714592
10 10231406648.6691 758063774.3053265 832406300.999445
9 10447995410.2143 769390272.6311302 840785570.3294811
8 10747131070.6018 777631229.4939194 855549428.5870361
7 11050457869.9997 792794107.4363346 862154080.0711994
6 11498743738.818 805466417.866621 872164727.8294468
5 11981321433.252 820281013.8598518 890901567.1591606
4 12545466237.1799 831846584.7331676 902085094.7134037
3 13195909811.5611 850085280.5717201 919046756.792696
2 14131584542.119 880569374.9271908 931307210.4824142
1 18440345471.9115 986411230.5366974 1034711403.93507
};
\addplot [semithick, steelblue31119180, dotted, error bars/.cd, y dir=both, y explicit, error bar style={steelblue31119180, solid}, error mark options={steelblue31119180, solid, mark size=1.5pt, line width=0.6pt, rotate=90}]
table [x index=0, y index=1, y error plus index=2, y error minus index=3] {%
12 143063949.017742 22879662.093342215 18943367.192789048
11 151784725.202965 23896202.73299253 19789320.623362064
10 161009557.393452 24764731.209918916 20532443.36044565
9 171259051.407247 26022317.83742109 21262651.396761566
8 182557776.165859 27325421.57056144 21961677.515871346
7 195097884.703809 28567253.24864143 22882703.70162469
6 208866239.144154 29945320.486082762 24241605.83540386
5 225183696.28074 31551678.202521354 25572759.53125456
4 244980692.245977 33289912.375660717 26514482.426407278
3 273024054.903956 35323830.12548345 28271593.010692358
2 323457497.292758 39135912.3476637 31655001.709499955
1 465465234.839557 49580103.096853495 41498460.6755926
};
\addplot [semithick, darkorange25512714, dotted, error bars/.cd, y dir=both, y explicit, error bar style={darkorange25512714, solid}, error mark options={darkorange25512714, solid, mark size=1.5pt, line width=0.6pt, rotate=90}]
table [x index=0, y index=1, y error plus index=2, y error minus index=3] {%
12 178975664.063163 21967457.41918248 32444622.493676066
11 191049288.101756 21494032.88551393 33630870.126675636
10 207632186.482065 22112328.75776297 34420953.42159963
9 225588353.865466 22717915.78393489 35237275.1906555
8 262027244.188083 22845789.119680703 36613657.305368066
7 305068063.506877 23248808.236941814 38127826.75502765
6 372136166.555108 24719998.168563902 39453041.44406748
5 467136966.371997 27491671.86185485 39469529.929370165
4 606953086.717846 29787417.686772227 38946272.40858352
3 797832581.121707 36368407.33627403 38996941.738060474
2 1108538900.65844 47349229.67826104 43104585.74488425
1 2329191531.59431 93121546.55425501 82653308.91409874
};
\end{axis}

\end{tikzpicture}
         \vspace{-0.8cm}
        \caption{}
        \label{fig:filter_length}
    \end{subfigure}
    \vspace{-0.2cm}
    \caption{\small{Improved Filter Performance: (a) The reduction of variance by congestion marking for the same filter size as we increase $R$. (b) The variance of the filter output as we increase the filter size for no correction and with $R=1$ and $2$. The color scheme is the same as for Fig.~\ref{fig:d_upper_bound}, orange for Close and blue for Distant.}}
    \label{fig:filtering_results}
\end{figure} 

\subsection{Open-Loop Evaluation over Multiple Hops}
\label{subsec:multihop}

We extend the simulation of Sec.~\ref{subsec:app_accuracy_mm1} to three switches with independent
cross-traffic at each hop. Besides the four traffic patterns of
Table~\ref{tab:flow-patterns}, we consider a mixed scenario (MI) in which the
three switches carry SS, SM, and LM traffic, so that the request
experiences increasing and the response decreasing queuing delay along
the path.

\subsubsection{Optimal Threshold and MSE}

We record the per-packet queuing delays on all six links, three per
direction, in 10\,s simulations repeated eight times, with at least
570{,}000 samples per link. We then apply Algorithm~\ref{alg1} with $N = 16$ and
sweep all thresholds to find the lowest MSE for each $R$
(Fig.~\ref{fig:multihop_analysis}). As $R$ increases, the MSE decreases and the
optimal threshold becomes smaller, since more and finer thresholds
increase the granularity of the correction. The gain diminishes for
$R > 8$, because for $R > 5$ the marks collected over three hops can exceed
$N = 16$, so the counter saturates. Increasing $N$ to 32 removes this
limit (Fig.~\ref{fig:mse_changes_N}): the curves for both values of $N$ coincide
up to $R = 8$, but for $N = 32$ the MSE continues to decrease up to much
larger $R$. Since the implementation cost is independent of $R$
(Sec.~\ref{subsec:impl_complexity}) and $N$ is set by the header field, the achievable improvement
is limited only by the available header bits and the classification
error of the CMC (Sec.~\ref{subsec:measurement_of_marking_accuracy}). \ROne{Fig.~\ref{fig:omnet_impr} shows the resulting
improvement $\mathcal{I}$, which Algorithm~\ref{alg1} predicts to reach up to
85\%, with a similar dependence on $R$ for all traffic patterns.}

\begin{figure}
    \centering
    \pgfplotslegendfromname{sharedlegend}
    \begin{subfigure}{0.48\columnwidth}
        \centering
\begin{tikzpicture}
\begin{axis}[
width=1.1\columnwidth,
height=4cm,
tick align=outside,
tick pos=left,
x grid style={darkgray},
xlabel={\footnotesize{$R$}},
xmajorgrids,
xmin=0.25, xmax=16.75,
xtick style={color=black},
y grid style={darkgray},
ylabel={\footnotesize{Optimal $\Delay{K}{*} (\mu s)$}},
ylabel style={at={(axis description cs:0.11,0.40)}, anchor=center},
ymajorgrids,
ymin=1.61595012816253, ymax=26.4944785653256,
ytick style={color=black}
]
\addplot [thick, steelblue]
table {%
1 21.3030303030303
2 16.0939700030609
3 12.9139790930425
4 10.8381670354846
5 9.48999858100408
6 8.41717071384127
7 7.6586055552413
8 7.20970701436497
9 6.88776967696882
10 6.73655976506823
11 6.62658892004962
12 6.54661012367245
13 6.56467211139484
14 6.50157971737806
15 6.53146787241009
16 6.55350661298926
};
\addplot [thick, darkorange]
table {%
1 19.2727272727273
2 13.900826446281
3 11.0843976959679
4 9.42245900022008
5 8.15886072742389
6 7.29093464105884
7 6.52728976070589
8 6.01819317380393
9 5.60790460906318
10 5.27177671697014
11 5.0544819180413
12 4.85285212966265
13 4.78366198281222
14 4.62209134277129
15 4.63448329229659
16 4.58558027691601
};
\addplot [thick, forestgreen]
table {%
1 25.3636363636364
2 17.7768595041322
3 14.0869020786376
4 11.8198994705381
5 10.0645753832087
6 8.93520052735173
7 8.32595594441087
8 7.97091656450914
9 7.79614891250692
10 7.75259916515683
11 7.71916703587797
12 7.69350196693662
13 7.67379948976953
14 7.74628849204297
15 7.71432247874005
16 7.68978291297216
};
\addplot [thick, crimson]
table {%
1 15.2121212121212
2 10.6617692072238
3 8.5389215077002
4 7.26292826263145
5 6.34169099806477
6 5.67198518058667
7 5.11102117187664
8 4.76537668166136
9 4.49416768585537
10 4.3297985974881
11 4.23018096817461
12 4.16980664737855
13 4.13321614992639
14 4.11104009086448
15 4.04597338748489
16 4.05816568938478
};
\addplot [thick, mediumpurple]
table {%
1 11.1515151515152
2 7.75084175084175
3 6.30354045505561
4 5.27884188275985
5 4.61509078098294
6 4.11949487832386
7 3.74073968233499
8 3.49008548971131
9 3.26771994429864
10 3.06919754875086
11 2.97294426606102
12 2.88614525762493
13 2.84493764755931
14 2.78653668472453
15 2.7594008838114
16 2.74679232985176
};
\end{axis}

\end{tikzpicture}          \vspace{-0.8cm}
        \caption{}
        \label{fig:optimal_threshold_multihop}
    \end{subfigure}
    \hfill
    \begin{subfigure}{0.48\columnwidth}
        \centering
\begin{tikzpicture}
\begin{axis}[
width=1.1\columnwidth,
height=4cm,
log basis y={10},
tick align=outside,
tick pos=left,
legend to name=sharedlegend,
x grid style={darkgray},
xlabel={\footnotesize{$R$}},
xmajorgrids,
xmin=-0.8, xmax=16.8,
xtick style={color=black},
y grid style={darkgray},
ylabel={\footnotesize{Optimal MSE$(\mu s^{2})$}},
ylabel style={at={(axis description cs:0.095,0.45)}, anchor=center},
ymajorgrids,
ymin=1.36285860875727, ymax=615.922990439231,
ymode=log,
ytick style={color=black},
ytick={0.1,1,10,100,1000,10000},
yticklabels={
  $10^{-1}$,
  $10^{0}$,
  $10^{1}$,
  $10^{2}$,
  $10^{3}$,
  $10^{4}$
},
legend style={
    legend columns=-1,
    font=\footnotesize,
    /tikz/every even column/.append style={column sep=0.1cm}
}
]
\addplot [thick, steelblue]
table {%
0 438.324064978416
1 148.321218594919
2 77.2383372920791
3 48.5705599950895
4 33.8914401990389
5 25.2219674603517
6 19.7277018753167
7 16.1118410916898
8 13.8059033896376
9 12.4560973930666
10 11.7169307214868
11 11.348018792513
12 11.1675870491147
13 11.0965055634468
14 11.05142617275
15 11.0659334358712
16 11.0406521604866
};
\addlegendentry{\footnotesize{LM}}
\addplot [thick, darkorange]
table {%
0 320.583680794525
1 110.080637616094
2 58.52262957933
3 37.3984909334325
4 26.3563219158254
5 19.8066802935415
6 15.5392485500077
7 12.5789711488292
8 10.4699462577998
9 8.91652134722269
10 7.79667800124534
11 7.01466692956219
12 6.49448794264075
13 6.13392384822533
14 5.90799847588737
15 5.7914072490212
16 5.73409055813933
};
\addlegendentry{\footnotesize{MI}}
\addplot [thick, forestgreen]
table {%
0 466.488734369711
1 158.521388170101
2 79.3050953010648
3 48.3882356722735
4 33.0085416442797
5 24.135623791749
6 18.6472748121579
7 15.4100164245666
8 13.8404272702987
9 13.2480377383112
10 13.0506841130512
11 12.9905594979409
12 12.9843798667724
13 12.9590259727376
14 12.9844982261326
15 12.9605255489269
16 12.9797183492602
};
\addlegendentry{\footnotesize{SF}}
\addplot [thick, crimson]
table {%
0 177.539698790624
1 62.9466150918434
2 33.4854164368919
3 21.4519135298156
4 15.1275524141639
5 11.4117238789195
6 8.9552919456152
7 7.31652092756301
8 6.19789896366777
9 5.44259914278603
10 4.98220866189609
11 4.6742275280254
12 4.51840471550204
13 4.42927437605739
14 4.38363755712346
15 4.3592273299068
16 4.35307614827538
};
\addlegendentry{\footnotesize{SM}}
\addplot [thick, mediumpurple]
table {%
0 76.0536988432979
1 28.1217207095961
2 15.2701391573151
3 9.92783046295827
4 7.10619214108468
5 5.37542616367281
6 4.24332761767274
7 3.48271705667181
8 2.93590940448615
9 2.54530429410551
10 2.27175584101298
11 2.08663866678139
12 1.96640779067324
13 1.88264117809347
14 1.83366683552284
15 1.80461212567374
16 1.79943455866258
};
\addlegendentry{\footnotesize{SS}}
\end{axis}

\end{tikzpicture}         \vspace{-0.8cm}
        \caption{}
        \label{fig:improved_mse_multihop}
    \end{subfigure}
    \hfill
    \begin{subfigure}{0.48\columnwidth}
        \centering
\begin{tikzpicture}

\begin{axis}[
width=1.1\columnwidth,
height=4cm,
legend cell align={left},
legend style={
    font=\footnotesize,            %
    fill opacity=0.8,
    draw opacity=1,
    text opacity=1,
    at={(0.99,0.99)},
    anchor=north east,
    draw=lightgray
  },
  legend image post style={
      xscale=0.6,
},
log basis y={10},
tick align=outside,
tick pos=left,
x grid style={darkgrey176},
xlabel={\footnotesize{$R$}},
xmajorgrids,
xmin=-0.35, xmax=29.35,
xtick style={color=black},
y grid style={darkgrey176},
ylabel={\footnotesize{Optimal MSE$(\mu s^{2})$}},
ylabel style={at={(axis description cs:0.095,0.45)}}, 
ymajorgrids,
ymin=0.490997766907809, ymax=208.043229428552,
ymode=log,
ytick style={color=black},
ytick={0.01,0.1,1,10,100,1000,10000},
yticklabels={
  $10^{-2}$,
  $10^{-1}$,
  $10^{0}$,
  $10^{1}$,
  $10^{2}$,
  $10^{3}$,
  $10^{4}$
}
]
\addlegendimage{solid, black}
\addlegendentry{$N=16$}
\addlegendimage{dashed, black}
\addlegendentry{$N=32$}
\addplot [line width=0.72pt, forestgreen]
table {%
1 158.030578141947
2 99.2891506288563
3 80.6588762773884
4 65.9555635780476
5 52.691050362046
6 40.9879395961656
7 30.7103923877323
8 21.8707106650629
9 14.816482512197
10 13.0473733152915
11 12.9933836271229
12 12.9824324173101
13 12.95325143568
14 12.9841171815311
15 12.9610219458488
16 12.9802165712337
};

\addplot [line width=0.72pt, mediumpurple]
table {%
1 27.8823916556373
2 15.3500472189346
3 9.93135597870939
4 7.10539801285449
5 5.37113599803278
6 4.25087409795669
7 3.47895389357414
8 2.93287435576218
9 2.54271639901218
10 2.27697672201173
11 2.08642122005481
12 1.9583445980958
13 1.88373628659363
14 1.83252242161798
15 1.80957474748564
16 1.78700890644105
};
\addplot [line width=0.72pt, forestgreen, dashed]
table {%
1 158.030578141947
2 99.2891506288563
3 80.6588762773884
4 65.9555635780476
5 52.691050362046
6 40.9879395961656
7 30.7103763255374
8 21.8672797303686
9 14.5726196787624
10 9.14866402709728
11 7.30695774287161
12 6.36379969267718
13 5.62884704475473
14 5.07420255773359
15 4.67556141068981
16 4.43370535176906
17 4.27527506536064
18 4.20557802861935
19 4.16352774266131
20 4.14175935801118
21 4.13114902975844
22 4.12711450356137
23 4.12561774309216
24 4.12192863705001
25 4.12555338858271
26 4.1311762704676
27 4.11274508610225
28 4.12195148400365
};

\addplot [line width=0.72pt, mediumpurple, dashed]
table {%
1 27.8823916556373
2 15.3500472189346
3 9.93135339731414
4 7.1052677621752
5 5.36987861654376
6 4.24586468392478
7 3.45932981154874
8 2.88908816357733
9 2.45027793880172
10 2.12307902456616
11 1.85240731358947
12 1.63253933769732
13 1.45637092590734
14 1.30645321080276
15 1.18217164748103
16 1.08000863483393
17 0.992388169421555
18 0.919564539166621
19 0.859191437165196
20 0.811494100924253
21 0.768346763945471
22 0.737893144421193
23 0.709589239270696
24 0.689687404470797
25 0.673748515236052
26 0.666005137658229
27 0.651733130602129
28 0.646386049274308
};
\end{axis}

\end{tikzpicture}
         \vspace{-0.8cm}
        \caption{}
        \label{fig:mse_changes_N}
    \end{subfigure}
    \hfill
    \begin{subfigure}{0.48\columnwidth}
        \centering
\begin{tikzpicture}

\begin{axis}[
width=1.1\columnwidth,
height=4cm,
tick align=outside,
tick pos=left,
x grid style={darkgray176},
xlabel={\footnotesize{$R$}},
xmajorgrids,
xmin=-0.8, xmax=16.8,
xtick style={color=black},
y grid style={darkgray176},
ylabel={\footnotesize{$\mathcal{I}$}},
ylabel style={at={(axis description cs:0.095,0.45)}}, 
ymajorgrids,
ymin=-0.0433122062356776, ymax=0.90955633094923,
ytick style={color=black}
]
\addplot [thick, steelblue]
table {%
0 0
1 0.418293408511844
2 0.580097927167921
3 0.66707052939014
4 0.72184119776944
5 0.760072436294672
6 0.787781747045423
7 0.808228209871113
8 0.822454215759015
9 0.831431728042803
10 0.836447235277577
11 0.839040018341628
12 0.840335789520409
13 0.840852658706849
14 0.841163203623404
15 0.841068690514027
16 0.841210873623735
};
\addplot [thick, darkorange]
table {%
0 0
1 0.414017342861157
2 0.572737112818169
3 0.658522641937583
4 0.713299587654995
5 0.751377789043385
6 0.779738950800101
7 0.801904914838277
8 0.819202442954988
9 0.833253343777285
10 0.844021258927721
11 0.85208978173316
12 0.857668948120156
13 0.861668611526574
14 0.864244392074926
15 0.865577867907739
16 0.866244124713552
};
\addplot [thick, forestgreen]
table {%
0 0
1 0.417060675517823
2 0.587827352090666
3 0.677950699050139
4 0.734034485182614
5 0.772559877121405
6 0.800126611223707
7 0.818274279060209
8 0.827805237486066
9 0.831521701032001
10 0.83277177555732
11 0.833135735610775
12 0.833203904584491
13 0.833357372391048
14 0.833198686610352
15 0.833352832584494
16 0.833207385283333
};
\addplot [thick, crimson]
table {%
0 0
1 0.404559389536091
2 0.565754971995643
3 0.652517269066662
4 0.708038984532668
5 0.746544717832296
6 0.775367373845386
7 0.796962639343705
8 0.813148652666461
9 0.824903346109165
10 0.83248511789949
11 0.837716893172515
12 0.840474575278838
13 0.842059775449714
14 0.842855169592059
15 0.843282911266066
16 0.843422635058164
};
\addplot [thick, mediumpurple]
table {%
0 0
1 0.391919922266729
2 0.552105889227524
3 0.638749728235023
4 0.69443256820204
5 0.73408747876873
6 0.763763229100441
7 0.786054120408043
8 0.803544193357276
9 0.817018700640584
10 0.827175760415574
11 0.834417604437251
12 0.839243441277085
13 0.842674266687236
14 0.84475713167967
15 0.846009848101792
16 0.846206389505159
};
\end{axis}

\end{tikzpicture}
         \vspace{-0.8cm}
        \caption{}
        \label{fig:omnet_impr}
    \end{subfigure}
    \vspace{-0.2cm}
    \caption{\small{Algorithm \ref{alg1} over multiple hops: The (a) optimal threshold $\Delay{K}{*}$ and (b) the best MSE obtained at that $\Delay{K}{*}$ decrease with an increase in $R$ slowly saturating for $R>8$. (c) A further decrease in MSE can be achieved by increasing $N$. (d) The improvement from baseline is similar for all the flows for a given $R$.}}
    \label{fig:multihop_analysis}
\end{figure}
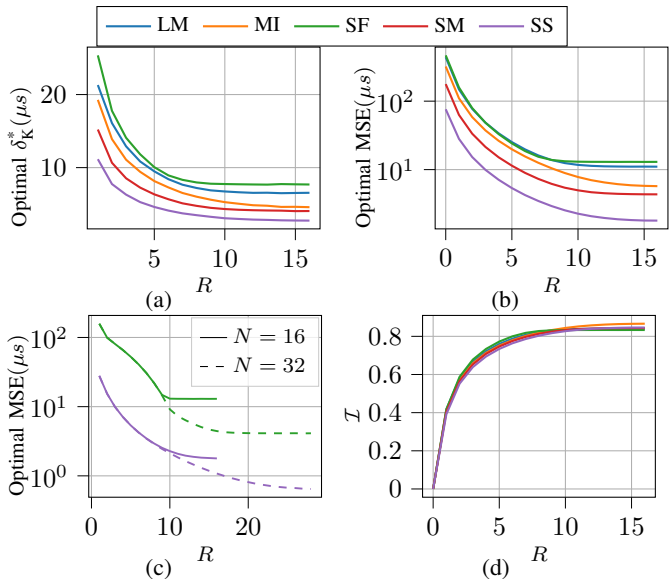

\subsubsection{MSE Improvement}

We validate these predictions in packet-level simulations with the
modified gPTP of Sec.~\ref{subsec:motivating_example} and~\ref{subsubsec:eval_mse_improvement_model}, running 60\,s simulations with four PTP
exchanges per second and eight seeds per setting. The observed
improvement matches Algorithm~\ref{alg1} within 5\%. Fig.~\ref{fig:multihop_results} compares
practical parameter settings. The median-delay filter alone ($M = 5$)
yields the smallest improvement, confirming Sec.~\ref{subsec:improved_filtering}. With classical ECN
($N = R = 1$), a single threshold must trade off the traffic patterns:
$\Delay{K}{*} = 12$~\unit{\micro\second}, the best setting for SS (Fig.~\ref{fig:optimal_threshold_multihop}), gives SF little
improvement, while $\Delay{K}{*} = 24$~\unit{\micro\second}, the best setting for SF,
gives SS little improvement. Since SS has a lower baseline error, the
larger threshold minimizes the worst-case error: with $\Delay{K}{*} = 24$~\unit{\micro\second}, a 40\%
improvement for SF gives an RMSE of 7.93~\unit{\micro\second}, while a 20\% improvement for SS
still gives a lower RMSE of 6.92~\unit{\micro\second}. The larger threshold is therefore a pragmatic choice
when the traffic is unknown. More thresholds remove this trade-off. With
$N = 32$ and the optimal thresholds of Fig.~\ref{fig:optimal_threshold_multihop} ($R_{o}1$,
$R_{o}2$), the results match the prediction of Fig.~\ref{fig:omnet_impr}, and
even a static threshold of 5~\unit{\micro\second} with $R = 5$ ($R_{5}5$) improves the RMSE by up
to 70\% while also serving SS traffic well.

\RThree{\textbf{Takeaway:} Algorithm~\ref{alg1} accurately predicts the multi-hop
improvement in terms of RMSE, which reaches up to 70\% in simulation. More thresholds and
a larger counter reduce the MSE further and make the result less
sensitive to the threshold, removing the trade-off between traffic
patterns that a single threshold faces.}

\begin{figure}
    \centering
\begin{tikzpicture}

\begin{axis}[
width=9cm,
height=4cm,
tick align=outside,
tick pos=left,
x grid style={darkgray},
label style={font=\footnotesize}, %
xlabel={},
xmajorgrids,
xmin=-0.15, xmax=3.15,
tick label style={font=\footnotesize}, %
xtick style={color=black},
xtick={0,1,2,3},
xticklabel style={rotate=45.0},
xticklabels={SF,LM,SM,SS},
y grid style={darkgray},
ylabel={$\mathcal{I}$},
ymajorgrids,
ymin=0, ymax=0.83,
ylabel style={at={(axis description cs:0.05,0.5)}, anchor=center, color=black},
ytick style={color=black},
legend style={
    at={(0.42,1.02)},
    anchor=south,
    legend columns=-1,
    font=\footnotesize,
    /tikz/every even column/.append style={column sep=0.1cm}
}
]
\addplot [thick, steelblue, mark=*, mark size=1.5, mark options={solid},
    error bars/.cd, y dir=both, y explicit, error bar style={steelblue, solid}, error mark options={rotate=90, mark size=2.5, steelblue}]
table [x index=0, y index=1, y error index=2] {%
-0.10 0.132496594895 0.0565770339994
0.90 0.124798053476 0.0716807156791
1.90 0.138967284311 0.0604981476148
2.90 0.116926162813 0.0692826833348
};\addlegendentry{$M5$}
\addplot [thick, darkorange, mark=square*, mark size=1.5, mark options={solid},
    error bars/.cd, y dir=both, y explicit, error bar style={darkorange, solid}, error mark options={rotate=90, mark size=2.5, darkorange}]
table [x index=0, y index=1, y error index=2] {%
-0.06 0.161080068597 0.061356895461
0.94 0.247697195714 0.0617225766836
1.94 0.377481676509 0.0464181927112
2.94 0.41667718366 0.0529007154905
};\addlegendentry{$\Delay{K}{*}12$}
\addplot [thick, forestgreen, mark=diamond*, mark size=1.5, mark options={solid},
    error bars/.cd, y dir=both, y explicit, error bar style={forestgreen, solid}, error mark options={rotate=90, mark size=2.5, forestgreen}]
table [x index=0, y index=1, y error index=2] {%
-0.02 0.395069815932 0.059052470533
0.98 0.412395220873 0.0548352887123
1.98 0.345486889565 0.0473600138316
2.98 0.196382616243 0.0644512826082
};\addlegendentry{$\Delay{K}{*}24$}
\addplot [thick, crimson, mark=triangle*, mark size=1.5, mark options={solid},
    error bars/.cd, y dir=both, y explicit, error bar style={crimson, solid}, error mark options={rotate=90, mark size=2.5, crimson}]
table [x index=0, y index=1, y error index=2] {%
0.02 0.400124641326 0.0534283820621
1.02 0.420620829532 0.0495992636602
2.02 0.403123985807 0.0509604202291
3.02 0.412024372325 0.0567185368947
};\addlegendentry{$R_{o}1$}
\addplot [thick, mediumpurple, mark=triangle*, mark size=1.5, mark options={solid,rotate=180},
    error bars/.cd, y dir=both, y explicit, error bar style={mediumpurple, solid}, error mark options={rotate=90, mark size=2.5, mediumpurple}]
table [x index=0, y index=1, y error index=2] {%
0.06 0.591933265661 0.0302212263292
1.06 0.578206061937 0.0465887706881
2.06 0.560911087349 0.03609087948
3.06 0.570698463047 0.0491452171067
};\addlegendentry{$R_{o}2$}
\addplot [thick, teal, mark=pentagon*, mark size=1.5, mark options={solid},
    error bars/.cd, y dir=both, y explicit, error bar style={teal, solid}, error mark options={rotate=90, mark size=2.5, teal}]
table [x index=0, y index=1, y error index=2] {%
0.10 0.732496594892 0.0356217131224
1.10 0.724798053474 0.0508381682697
2.10 0.658967284309 0.0415464515665
3.10 0.586926162782 0.0481910349342
};\addlegendentry{$R_{5}5$}
\end{axis}

\end{tikzpicture}
     \vspace{-0.4cm}
    \caption{\small{Observed Improvement for the multihop topology for different pragmatic parameter settings. }}
    \label{fig:multihop_results}
\end{figure}
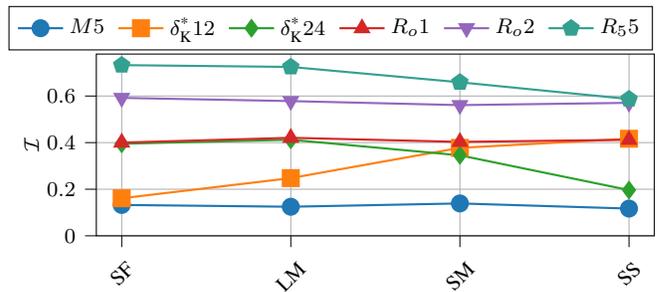

\subsection{Closed-Loop Evaluation over Multiple Hops}
\label{subsec:closed_loop}

The previous evaluations measure the offset estimation error with ideal
oscillators and without adjusting the clocks so that we obtain the PDV-dependent offset estimation error. 
To evaluate the effect of congestion marking on the synchronized clock itself, we now close the loop and extend the topology to longer paths. 
We simulate a linear chain of
$L \in \{1, 3, 5, 7\}$ switches between a source and sink with the same OMNeT++ setup mentioned in Sec.~\ref{subsec:multihop}. 
We consider two asymmetric scenarios (Distant and Close) built this time from the flow types of Table~\ref{tab:flow-patterns}: Distant uses SF traffic in the forward and SS traffic in the reverse direction, and Close uses LM and SM traffic, respectively.

The sink clock drifts with an initial frequency drift drawn uniformly from
$\pm 50$\,ppm and changes the drift value with a random walk of up to $\pm 0.1$\,ppm per second. 
The sink adjusts its clock with a PI servo using the default gains of linuxptp (0.7 for P and 0.3 for I). 
We compare congestion marking at $R = 1$ and $R = 5$, with median delay and minRTT, both with a filter size of 5. 
The used threshold is obtained from the data collected in the single-hop case in Sec.~\ref{subsec:app_accuracy_mm1} and then applied to Algorithm~1. 
For each configuration, we run eight 60\,s simulations and sample the time error of the sink clock every 10\,ms.

Fig.~\ref{fig:closed_loop} shows the resulting RMSE. 
The PI servo acts as a low-pass filter on the offset measurements: it attenuates their
variance, but its integral term passes any bias in the measurements directly into the clock adjustments~\cite{LimitationsClockSync, haestimtation2013}. 
Under asymmetric load, this bias accumulates at every hop, and the RMSE with the median delay filter grows almost linearly with the number of switches, from 11 to
73\,$\mu$s for Distant and from 6 to 28\,$\mu$s under Close. 
The min-RTT filter (F5) performs well over a single hop, with 3.6 and 1.1\,$\mu$s,
respectively, but its error grows quickly with the path length, reaching
53.1 and 16.4\,$\mu$s at $L = 7$, since the probability that a packet passes
all switches without queuing decreases with every hop. 
Congestion marking, in contrast, removes part of the bias at every switch before it reaches the servo. 
With $R=5$, the RMSE remains between 1.5 and 8.8\,$\mu$s for Distant and between 0.8 and 3.2\,$\mu$s for Close. 
Even with a single threshold ($R=1, N=1$), which is the classical ECN case, marking achieves results comparable to minRTT filtering for the one-hop case but outperforms it from 3 hops onwards. 
For a given $R$, the relative improvement over an uncorrected baseline remains nearly constant across all path lengths, at about 85\% for $R=5$ and about 55\% for $R=1$ in line with our results in Fig.~\ref{fig:omnet_impr}. 
Thus, the absolute reduction in RMSE grows with the number of hops. 
Since each marking switch corrects its own queuing delay, the benefit accumulates along the path as long as $N$ suffices, whereas the benefit of statistical filtering diminishes.

\RThree{\textbf{Takeaway:} Marking keeps its relative improvement of the RMS time error, about 85\% for
$R = 5$ and 55\% for $R = 1$, at all path lengths, so its absolute benefit
grows with the number of hops, while that of filtering declines.}

\begin{figure}
    \centering
    \begin{subfigure}{0.48\columnwidth}
        \centering
\xglobal\providecolor{steelblue}{RGB}{70,130,180}
\xglobal\providecolor{darkorange}{RGB}{255,140,0}
\xglobal\providecolor{forestgreen}{RGB}{34,139,34}
\xglobal\providecolor{crimson}{RGB}{220,20,60}
\xglobal\providecolor{mediumpurple}{RGB}{147,112,219}
\begin{tikzpicture}
\begin{axis}[
  width=1.07\columnwidth,
height=4cm,
  xlabel={\footnotesize{$L$}},  ylabel={\footnotesize{RMSE$(\mu s)$}},
  ylabel style={at={(axis description cs:0.11,0.40)}, anchor=center},
  xmin=0.5, xmax=7.5, xtick={1,3,5,7}, ymin=0,
  ymajorgrids, grid style={gray!25},
  tick align=outside, tick pos=left, tick style={gray!70!black},
  every axis plot/.append style={line width=1pt, mark size=2.2pt},
  error bars/error bar style={line width=0.7pt, solid},
  error bars/error mark options={mark size=2.5pt, line width=0.7pt},
]
\addplot[mediumpurple, dashed, mark=*, mark options={solid, draw=white, fill=mediumpurple},
  error bars/.cd, y dir=both, y explicit]
  table[x=N, y=rms, y error=ci] {
    N rms ci
    0.92 11.298 0.230
    2.92 32.096 0.863
    4.92 53.553 0.999
    6.92 73.248 1.288
  };
\addplot[forestgreen, dash dot dot, mark=diamond*, mark options={solid, draw=white, fill=forestgreen, scale=1.2},
  error bars/.cd, y dir=both, y explicit]
  table[x=N, y=rms, y error=ci] {
    N rms ci
    1.00 3.643 0.299
    3.00 18.001 0.788
    5.00 35.431 0.730
    7.00 53.091 1.772
  };
\addplot[crimson, densely dotted, mark=triangle*, mark options={solid, draw=white, fill=crimson},
  error bars/.cd, y dir=both, y explicit]
  table[x=N, y=rms, y error=ci] {
    N rms ci
    1.08 4.797 0.201
    3.08 13.731 0.366
    5.08 22.980 0.334
    7.08 30.766 0.501
  };
\addplot[teal, dashed, mark=pentagon*, mark options={solid, draw=white, fill=teal},
  error bars/.cd, y dir=both, y explicit]
  table[x=N, y=rms, y error=ci] {
    N rms ci
    1.16 1.499 0.077
    3.16 3.962 0.149
    5.16 6.568 0.219
    7.16 8.787 0.285
  };
\node[anchor=north west, font=\footnotesize,
      draw=steelblue, text=steelblue, fill=white, line width=0.6pt,
      rounded corners=5pt, inner xsep=4pt, inner ysep=2pt]
  at (rel axis cs:0.03,0.95) {Distant};
\end{axis}
\end{tikzpicture}          %
        \caption{}
        \label{fig:closed_loop1}
    \end{subfigure}
    \hfill
    \begin{subfigure}{0.48\columnwidth}
        \centering
\xglobal\providecolor{steelblue}{RGB}{70,130,180}
\xglobal\providecolor{darkorange}{RGB}{255,140,0}
\xglobal\providecolor{forestgreen}{RGB}{34,139,34}
\xglobal\providecolor{crimson}{RGB}{220,20,60}
\xglobal\providecolor{mediumpurple}{RGB}{147,112,219}
\begin{tikzpicture}
\begin{axis}[
  width=1.07\columnwidth,
height=4cm,
  xlabel={\footnotesize{$L$}},
  xmin=0.5, xmax=7.5, xtick={1,3,5,7}, ymin=0,
  ymajorgrids, grid style={gray!25},
  tick align=outside, tick pos=left, tick style={gray!70!black}, 
  legend style={at={(-0.3,1.02)}, anchor=south, overlay,
    legend columns=-1,
    font=\footnotesize,
    /tikz/every even column/.append style={column sep=0.1cm},
    },
    every axis plot/.append style={line width=1pt, mark size=2.2pt},
  error bars/error bar style={line width=0.7pt, solid},
  error bars/error mark options={mark size=2.5pt, line width=0.7pt},
]
\addplot[mediumpurple, dashed, mark=*, mark options={solid, draw=white, fill=mediumpurple},
  error bars/.cd, y dir=both, y explicit]
  table[x=N, y=rms, y error=ci] {
    N rms ci
    0.92 5.960 0.487
    2.92 12.630 0.869
    4.92 20.930 1.132
    6.92 27.669 1.339
  };
\addlegendentry{\footnotesize{M5}}
\addplot[forestgreen, dash dot dot, mark=diamond*, mark options={solid, draw=white, fill=forestgreen, scale=1.2},
  error bars/.cd, y dir=both, y explicit]
  table[x=N, y=rms, y error=ci] {
    N rms ci
    1.00 1.095 0.205
    3.00 6.006 0.662
    5.00 10.569 0.572
    7.00 16.368 0.943
  };
\addlegendentry{\footnotesize{F5}}
\addplot[crimson, densely dotted, mark=triangle*, mark options={solid, draw=white, fill=crimson},
  error bars/.cd, y dir=both, y explicit]
  table[x=N, y=rms, y error=ci] {
    N rms ci
    1.08 2.262 0.204
    3.08 5.317 0.456
    5.08 8.344 0.346
    7.08 11.503 0.593
  };
\addlegendentry{\footnotesize{R1}}
\addplot[teal, dashed, mark=pentagon*, mark options={solid, draw=white, fill=teal},
  error bars/.cd, y dir=both, y explicit]
  table[x=N, y=rms, y error=ci] {
    N rms ci
    1.16 0.791 0.066
    3.16 1.709 0.143
    5.16 2.616 0.168
    7.16 3.237 0.285
  };
\addlegendentry{\footnotesize{R5}}
\node[anchor=north west, font=\footnotesize,
      draw=darkorange, text=darkorange, fill=white, line width=0.6pt,
      rounded corners=5pt, inner xsep=4pt, inner ysep=2pt]
  at (rel axis cs:0.03,0.95) {Close};
\end{axis}
\end{tikzpicture}         %
        \caption{}
        \label{fig:closed_loop2}
    \end{subfigure}
    \caption{\small{Closed-loop RMS time error of the sink clock over increasing number of hops. The error of marking grows only slowly with $L$, and marking with $R = 5$ outperforms all other schemes at every path length.}}
    \label{fig:closed_loop}
\end{figure} 
\section{Conclusion}
\label{sec:conclusion}
In this paper, we propose a novel, backward-compatible congestion-marking method to improve clock synchronization performance, with a simpler implementation in switches than the current state of the art. 
\ROne{The method corrects for \ac{PDV} caused by congestion, effectively reducing offset estimation errors in popular clock synchronization methods such as \ac{PTP} and \ac{NTP}.
Congestion-marking improves the performance of commonly used RTT-based delay filters in legacy PTP/NTP and modern systems like Firefly.} 
We show that the method is relatively easy to implement in switches and endpoints and demonstrate several practical methods for selecting threshold-related parameters.   
Analytical and experimental results demonstrate the achieved improvement and provide insights into the optimal settings for different traffic conditions to achieve the best results with our method. 

\bibliographystyle{IEEEtran}

\begin{thebibliography}{10}
\providecommand{\url}[1]{#1}
\csname url@samestyle\endcsname
\providecommand{\newblock}{\relax}
\providecommand{\bibinfo}[2]{#2}
\providecommand{\BIBentrySTDinterwordspacing}{\spaceskip=0pt\relax}
\providecommand{\BIBentryALTinterwordstretchfactor}{4}
\providecommand{\BIBentryALTinterwordspacing}{\spaceskip=\fontdimen2\font plus
\BIBentryALTinterwordstretchfactor\fontdimen3\font minus
  \fontdimen4\font\relax}
\providecommand{\BIBforeignlanguage}[2]{{%
\expandafter\ifx\csname l@#1\endcsname\relax
\typeout{** WARNING: IEEEtran.bst: No hyphenation pattern has been}%
\typeout{** loaded for the language `#1'. Using the pattern for}%
\typeout{** the default language instead.}%
\else
\language=\csname l@#1\endcsname
\fi
#2}}
\providecommand{\BIBdecl}{\relax}
\BIBdecl

\bibitem{ITUstd}
``{Packet over Transport aspects – Synchronization, quality and availability
  targets},'' \emph{ITU-T G.8271, International Telecommunications Union},
  2020.

\bibitem{Gore2020}
\BIBentryALTinterwordspacing
R.~N. Gore, E.~Lisova, J.~Akerberg, and M.~Bjorkman, ``{Clock Synchronization
  in Future Industrial Networks: Applications, Challenges, and Directions},''
  in \emph{12th AEIT International Annual Conference, AEIT 2020}.\hskip 1em
  plus 0.5em minus 0.4em\relax Institute of Electrical and Electronics
  Engineers Inc., sep 2020. [Online]. Available:
  \url{https://ieeexplore.ieee.org/document/9241154}
\BIBentrySTDinterwordspacing

\bibitem{spanner}
J.~C. Corbett, J.~Dean, M.~Epstein, A.~Fikes, C.~Frost, J.~J. Furman,
  S.~Ghemawat, A.~Gubarev, C.~Heiser, P.~Hochschild, W.~Hsieh, S.~Kanthak,
  E.~Kogan, H.~Li, A.~Lloyd, S.~Melnik, D.~Mwaura, D.~Nagle, S.~Quinlan,
  R.~Rao, L.~Rolig, Y.~Saito, M.~Szymaniak, C.~Taylor, R.~Wang, and
  D.~Woodford, ``Spanner: Google’s globally distributed database,'' \emph{ACM
  Trans. Comput. Syst.}, 2013.

\bibitem{mills1992rfc}
\BIBentryALTinterwordspacing
D.~L. Mills, ``{Network Time Protocol (Version 3) Specification, Implementation
  and Analysis},'' RFC 1305, Mar. 1992. [Online]. Available:
  \url{https://www.rfc-editor.org/info/rfc1305}
\BIBentrySTDinterwordspacing

\bibitem{rfc5905}
\BIBentryALTinterwordspacing
D.~L. Mills, J.~Martin, J.~Burbank, and W.~Kasch, ``{Network Time Protocol
  Version 4: Protocol and Algorithms Specification},'' RFC 5905, Jun. 2010.
  [Online]. Available: \url{https://www.rfc-editor.org/info/rfc5905}
\BIBentrySTDinterwordspacing

\bibitem{ptp}
J.~C. Eidson, M.~Fischer, and J.~White, ``Ieee-1588™ standard for a precision
  clock synchronization protocol for networked measurement and control
  systems,'' in \emph{Proceedings of the 34th Annual Precise Time and Time
  Interval Systems and Applications Meeting}, 2002, pp. 243--254.

\bibitem{ieee1588}
``{IEEE Standard for a Precision Clock Synchronization Protocol for Networked
  Measurement and Control Systems},'' IEEE Std 1588-2019 (Revision of IEEE Std
  1588-2008), 2020.

\bibitem{Banerjee2023}
\BIBentryALTinterwordspacing
P.~Banerjee and D.~Matsakis, \emph{Network Time Protocol (NTP) and Precise Time
  Protocol (PTP)}.\hskip 1em plus 0.5em minus 0.4em\relax Cham: Springer Nature
  Switzerland, 2023, pp. 141--152. [Online]. Available:
  \url{https://doi.org/10.1007/978-3-031-30780-5_7}
\BIBentrySTDinterwordspacing

\bibitem{kalmanFiltering}
C.~Novaes, I.~Freire, A.~Klautau, I.~Almeida, and E.~Medeiros, ``Analysis of
  kalman filtering for clock synchronization in ptp-unaware networks,'' in
  \emph{2021 IEEE Latin-American Conference on Communications (LATINCOM)},
  2021, pp. 1--6.

\bibitem{genghuygens}
\BIBentryALTinterwordspacing
Y.~Geng, S.~Liu, Z.~Yin, A.~Naik, B.~Prabhakar, M.~Rosenblum, and A.~Vahdat,
  ``Exploiting a natural network effect for scalable, fine-grained clock
  synchronization,'' in \emph{15th {USENIX} Symposium on Networked Systems
  Design and Implementation ({NSDI} 18)}.\hskip 1em plus 0.5em minus
  0.4em\relax Renton, WA: {USENIX} Association, Apr. 2018, pp. 81--94.
  [Online]. Available:
  \url{https://www.usenix.org/conference/nsdi18/presentation/geng}
\BIBentrySTDinterwordspacing

\bibitem{bayesian_optimization}
T.~Kernen, M.~Machnikowski, G.~Shabat, and H.~Shteingart, ``Optimizing the ptp
  control loop,'' in \emph{2024 IEEE International Symposium on Precision Clock
  Synchronization for Measurement, Control, and Communication (ISPCS)}, 2024,
  pp. 1--7.

\bibitem{LimitationsClockSync}
N.~M. Freris, S.~R. Graham, and P.~R. Kumar, ``Fundamental limits on
  synchronizing clocks over networks,'' \emph{IEEE Transactions on Automatic
  Control}, vol.~56, no.~6, pp. 1352--1364, 2011.

\bibitem{zilbermanTime}
N.~Zilberman, M.~Grosvenor, D.~A. Popescu, N.~Manihatty-Bojan, G.~Antichi,
  M.~W{\'o}jcik, and A.~W. Moore, ``Where has my time gone?'' in \emph{Passive
  and Active Measurement}, M.~A. Kaafar, S.~Uhlig, and J.~Amann, Eds.\hskip 1em
  plus 0.5em minus 0.4em\relax Cham: Springer International Publishing, 2017,
  pp. 201--214.

\bibitem{TransparentClockImpl}
N.~Moreira, A.~Astarloa, J.~Lazaro, A.~Garcia, and E.~Ormaetxea, ``Ieee 1588
  transparent clock architecture for fpga-based network devices,'' in
  \emph{2013 IEEE International Symposium on Industrial Electronics}, 2013, pp.
  1--6.

\bibitem{BoundaryClockImpl}
\BIBentryALTinterwordspacing
R.~Ratzel and R.~Greenstreet, ``Toward higher precision: An introduction to ptp
  and its significance to ntp practitioners,'' \emph{Queue}, vol.~10, no.~8, p.
  40–53, Aug. 2012. [Online]. Available:
  \url{https://doi.org/10.1145/2346916.2354406}
\BIBentrySTDinterwordspacing

\bibitem{ramakrishnan2001addition}
\BIBentryALTinterwordspacing
K.~K. Ramakrishnan, S.~Floyd, and D.~L. Black, ``{The Addition of Explicit
  Congestion Notification (ECN) to IP},'' RFC 3168, Sep. 2001. [Online].
  Available: \url{https://www.rfc-editor.org/info/rfc3168}
\BIBentrySTDinterwordspacing

\bibitem{ECNAdoption}
\BIBentryALTinterwordspacing
S.~Bauer, R.~Beverly, and A.~Berger, ``Measuring the state of ecn readiness in
  servers, clients,and routers,'' in \emph{Proceedings of the 2011 ACM SIGCOMM
  Conference on Internet Measurement Conference}, ser. IMC '11.\hskip 1em plus
  0.5em minus 0.4em\relax New York, NY, USA: Association for Computing
  Machinery, 2011, p. 171–180. [Online]. Available:
  \url{https://doi.org/10.1145/2068816.2068833}
\BIBentrySTDinterwordspacing

\bibitem{ECNTuning}
\BIBentryALTinterwordspacing
H.~Wu, J.~Ju, G.~Lu, C.~Guo, Y.~Xiong, and Y.~Zhang, ``Tuning ecn for data
  center networks,'' in \emph{Proceedings of the 8th International Conference
  on Emerging Networking Experiments and Technologies}, ser. CoNEXT '12.\hskip
  1em plus 0.5em minus 0.4em\relax New York, NY, USA: Association for Computing
  Machinery, 2012, p. 25–36. [Online]. Available:
  \url{https://doi.org/10.1145/2413176.2413181}
\BIBentrySTDinterwordspacing

\bibitem{firefly}
\BIBentryALTinterwordspacing
P.~Namyar, Y.~Li, W.~Wang, N.~Dukkipati, K.~Yap, J.~Gong, C.~Chen, P.~Gao,
  D.~Ray, G.~Kumar, Y.~Ma, R.~Govindan, and A.~Vahdat, ``Firefly: Scalable,
  ultra-accurate clock synchronization for datacenters,'' in \emph{Proceedings
  of the ACM SIGCOMM 2025 Conference}, ser. SIGCOMM '25.\hskip 1em plus 0.5em
  minus 0.4em\relax New York, NY, USA: Association for Computing Machinery,
  2025, p. 434–452. [Online]. Available:
  \url{https://doi.org/10.1145/3718958.3750502}
\BIBentrySTDinterwordspacing

\bibitem{p4}
\BIBentryALTinterwordspacing
P.~Bosshart, D.~Daly, G.~Gibb, M.~Izzard, N.~McKeown, J.~Rexford,
  C.~Schlesinger, D.~Talayco, A.~Vahdat, G.~Varghese, and D.~Walker, ``P4:
  programming protocol-independent packet processors,'' \emph{SIGCOMM Comput.
  Commun. Rev.}, vol.~44, no.~3, p. 87–95, Jul. 2014. [Online]. Available:
  \url{https://doi.org/10.1145/2656877.2656890}
\BIBentrySTDinterwordspacing

\bibitem{cristian1989probabilistic}
F.~Cristian, ``Probabilistic clock synchronization,'' \emph{Distributed
  computing}, vol.~3, no.~3, pp. 146--158, 1989.

\bibitem{moongen}
P.~Emmerich, S.~Gallenm\"{u}ller, D.~Raumer, F.~Wohlfart, and G.~Carle,
  ``Moongen: A scriptable high-speed packet generator,'' in \emph{Proceedings
  of the 2015 Internet Measurement Conference}, 2015.

\bibitem{yashinfocom}
Y.~Deshpande, P.~Diederich, and W.~Kellerer, ``Towards a network aware model of
  the time uncertainty bound in precision time protocol,'' in \emph{IEEE
  INFOCOM 2023 - IEEE Conference on Computer Communications Workshops (INFOCOM
  WKSHPS)}, 2023, pp. 1--2.

\bibitem{chrony}
\BIBentryALTinterwordspacing
R.~P. Curnow and M.~Lichvar, ``Chrony: Versatile implementation of the network
  time protocol,'' 2003. [Online]. Available:
  \url{https://chrony.tuxfamily.org/}
\BIBentrySTDinterwordspacing

\bibitem{g8272}
{ITU-T}, ``{Timing Characteristics of Primary Reference Time Clocks},''
  Recommendation ITU-T G.8272/Y.1367, 2018, international Telecommunication
  Union.

\bibitem{psiaki2016}
M.~L. Psiaki and T.~E. Humphreys, ``{GNSS} spoofing and detection,''
  \emph{Proceedings of the IEEE}, vol. 104, no.~6, pp. 1258--1270, 2016.

\bibitem{linuxptp}
R.~Cochran, ``Linux ptp,'' \url{https://github.com/richardcochran/linuxptp},
  2025, version 4.4.

\bibitem{levesqueDelayAssymetry}
M.~Lévesque and D.~Tipper, ``Improving the ptp synchronization accuracy under
  asymmetric delay conditions,'' in \emph{2015 IEEE International Symposium on
  Precision Clock Synchronization for Measurement, Control, and Communication
  (ISPCS)}, 2015, pp. 88--93.

\bibitem{Murakami}
T.~Murakami and Y.~Horiuchi, ``Improvement of synchronization accuracy in ieee
  1588 using a queuing estimation method,'' in \emph{2009 International
  Symposium on Precision Clock Synchronization for Measurement, Control and
  Communication}, 2009, pp. 1--5.

\bibitem{lee}
S.~Lee, ``An enhanced {IEEE} 1588 time synchronization algorithm for asymmetric
  communication link using block burst transmission,'' \emph{IEEE
  Communications Letters}, vol.~12, no.~9, pp. 687--689, 2008.

\bibitem{kalmanfilteringreview}
A.~Bletsas, ``Evaluation of kalman filtering for network time keeping,''
  \emph{IEEE Transactions on Ultrasonics, Ferroelectrics, and Frequency
  Control}, vol.~52, no.~9, pp. 1452--1460, 2005.

\bibitem{KalmanFiltering2}
Z.~Chaloupka, N.~Alsindi, and J.~Aweya, ``Clock skew estimation using kalman
  filter and ieee 1588v2 ptp for telecom networks,'' \emph{IEEE Communications
  Letters}, vol.~19, no.~7, pp. 1181--1184, 2015.

\bibitem{LinearProgramming}
H.~Puttnies, E.~Schweissguth, D.~Timmermann, and J.~Schacht, ``Clock
  synchronization using linear programming, multicasts, and temperature
  compensation,'' in \emph{2019 IEEE Global Communications Conference
  (GLOBECOM)}, 2019, pp. 1--6.

\bibitem{graham}
\BIBentryALTinterwordspacing
A.~Najafi and M.~Wei, ``Graham: Synchronizing clocks by leveraging local clock
  properties,'' in \emph{19th USENIX Symposium on Networked Systems Design and
  Implementation (NSDI 22)}.\hskip 1em plus 0.5em minus 0.4em\relax Renton, WA:
  USENIX Association, Apr. 2022, pp. 453--466. [Online]. Available:
  \url{https://www.usenix.org/conference/nsdi22/presentation/najafi}
\BIBentrySTDinterwordspacing

\bibitem{Timely}
\BIBentryALTinterwordspacing
R.~Mittal, V.~T. Lam, N.~Dukkipati, E.~Blem, H.~Wassel, M.~Ghobadi, A.~Vahdat,
  Y.~Wang, D.~Wetherall, and D.~Zats, ``Timely: Rtt-based congestion control
  for the datacenter,'' in \emph{Proceedings of the 2015 ACM Conference on
  Special Interest Group on Data Communication}, ser. SIGCOMM '15.\hskip 1em
  plus 0.5em minus 0.4em\relax New York, NY, USA: Association for Computing
  Machinery, 2015, p. 537–550. [Online]. Available:
  \url{https://doi.org/10.1145/2785956.2787510}
\BIBentrySTDinterwordspacing

\bibitem{hwptptc}
H.~Kirrmann, C.~Honegger, D.~Ilie, and I.~Sotiropoulos, ``Performance of a
  full-hardware ptp implementation for an iec 62439-3 redundant iec 61850
  substation automation network,'' in \emph{2012 IEEE International Symposium
  on Precision Clock Synchronization for Measurement, Control and Communication
  Proceedings}, 2012, pp. 1--6.

\bibitem{NetTimeLogicPTPTransparentClock}
{NetTimeLogic GmbH}, ``{PTP Transparent Clock},''
  \url{https://www.nettimelogic.com/ptp-transparent-clock.php}, n.d., accessed:
  2025-12-18.

\bibitem{tc_verifying_performance}
J.~Burch, K.~Green, J.~Nakulski, and D.~Vook, ``Verifying the performance of
  transparent clocks in ptp systems,'' in \emph{2009 International Symposium on
  Precision Clock Synchronization for Measurement, Control and Communication},
  2009, pp. 1--6.

\bibitem{substation_performance}
H.~Liu, J.~Liu, T.~Bi, J.~Li, W.~Yang, and D.~Zhang, ``Performance analysis of
  time synchronization precision of ptp in smart substations,'' in \emph{2015
  IEEE International Symposium on Precision Clock Synchronization for
  Measurement, Control, and Communication (ISPCS)}, 2015, pp. 37--42.

\bibitem{dptp}
P.~G. Kannan, R.~Joshi, and M.~C. Chan, ``Precise time-synchronization in the
  data-plane using programmable switching asics,'' in \emph{Proceedings of the
  2019 ACM Symposium on SDN Research}, ser. SOSR '19.\hskip 1em plus 0.5em
  minus 0.4em\relax San Jose, CA, USA: Association for Computing Machinery,
  2019, pp. 8--20.

\bibitem{ecnNew}
\BIBentryALTinterwordspacing
D.~L. Black, ``{Relaxing Restrictions on Explicit Congestion Notification (ECN)
  Experimentation},'' RFC 8311, Jan. 2018. [Online]. Available:
  \url{https://www.rfc-editor.org/info/rfc8311}
\BIBentrySTDinterwordspacing

\bibitem{rfc9331}
\BIBentryALTinterwordspacing
K.~D. Schepper and B.~Briscoe, ``{The Explicit Congestion Notification (ECN)
  Protocol for Low Latency, Low Loss, and Scalable Throughput (L4S)},'' RFC
  9331, Jan. 2023. [Online]. Available:
  \url{https://www.rfc-editor.org/info/rfc9331}
\BIBentrySTDinterwordspacing

\bibitem{KaanTCP}
K.~Aykurt, J.~Zerwas, A.~Blenk, and W.~Kellerer, ``When tcp meets
  reconfigurations: A comprehensive measurement study,'' \emph{IEEE
  Transactions on Network and Service Management}, vol.~21, no.~2, pp.
  1372--1386, 2024.

\bibitem{dctcp}
\BIBentryALTinterwordspacing
M.~Alizadeh, A.~Greenberg, D.~A. Maltz, J.~Padhye, P.~Patel, B.~Prabhakar,
  S.~Sengupta, and M.~Sridharan, ``Data center tcp (dctcp),'' \emph{SIGCOMM
  Comput. Commun. Rev.}, vol.~40, no.~4, p. 63–74, Aug. 2010. [Online].
  Available: \url{https://doi.org/10.1145/1851275.1851192}
\BIBentrySTDinterwordspacing

\bibitem{RED}
S.~Floyd and V.~Jacobson, ``Random early detection gateways for congestion
  avoidance,'' \emph{IEEE/ACM Transactions on Networking}, vol.~1, no.~4, pp.
  397--413, 1993.

\bibitem{codel}
\BIBentryALTinterwordspacing
K.~Nichols and V.~Jacobson, ``Controlling queue delay,'' \emph{Commun. ACM},
  vol.~55, no.~7, p. 42–50, Jul. 2012. [Online]. Available:
  \url{https://doi.org/10.1145/2209249.2209264}
\BIBentrySTDinterwordspacing

\bibitem{Pie}
R.~Pan, P.~Natarajan, C.~Piglione, M.~S. Prabhu, V.~Subramanian, F.~Baker, and
  B.~VerSteeg, ``Pie: A lightweight control scheme to address the bufferbloat
  problem,'' in \emph{2013 IEEE 14th International Conference on High
  Performance Switching and Routing (HPSR)}, 2013, pp. 148--155.

\bibitem{verdi2021heavy}
F.~L. Verdi and M.~Chiesa, ``Heavy hitter detection on multi-pipeline
  switches,'' in \emph{Proceedings of the Symposium on Architectures for
  Networking and Communications Systems (ANCS '21)}, 2021.

\bibitem{chiesa2023multipipe}
M.~Chiesa and F.~L. Verdi, ``Network monitoring on multi-pipe switches,'' in
  \emph{Abstract Proceedings of the 2023 ACM SIGMETRICS International
  Conference on Measurement and Modeling of Computer Systems}, 2023, pp.
  49--50.

\bibitem{int}
{The P4.org Applications Working Group}, ``{In-band Network Telemetry (INT)
  Dataplane Specification, Version 2.1},''
  \url{https://p4.org/p4-spec/docs/INT_v2_1.pdf}, 2020.

\bibitem{trident4}
{Broadcom Inc.}, ``{StrataXGS Trident 4 BCM56880 Series},'' Press release, Jun.
  11, 2019; product page:
  \url{https://www.broadcom.com/products/ethernet-connectivity/switching/strataxgs/bcm56880-series},
  2019, accessed: 2026-11-04.

\bibitem{juniperifa}
{Juniper Networks}, ``{Understand Inband Flow Analyzer 2.0},'' Paragon Insights
  Data Ingest Guide,
  \url{https://www.juniper.net/documentation/us/en/software/paragon-insights/data-ingest-guide/topics/concept/pi-inband-telemetry-overview.html},
  accessed: 2026-11-04.

\bibitem{ciscoint}
{Cisco Systems}, ``{Cisco Nexus 9000 Series NX-OS Programmability Guide,
  Release 9.3(x)},'' Chapter ``Inband Network Telemetry'';
  \url{https://www.cisco.com/c/en/us/support/switches/nexus-9000-series-switches/products-programming-reference-guides-list.html},
  accessed: 2026-11-04.

\bibitem{variance1}
T.~Jones and G.~A. Koenig, ``A clock synchronization strategy for minimizing
  clock variance at runtime in high-end computing environments,'' in \emph{2010
  22nd International Symposium on Computer Architecture and High Performance
  Computing}, 2010, pp. 207--214.

\bibitem{Whitt1990HeavyTail}
W.~Whitt, ``The impact of a heavy-tailed service-time distribution upon the
  {M/GI/s} waiting-time distribution,'' \emph{Queueing Systems}, vol.~36, no.
  1--3, pp. 71--87, 2000.

\bibitem{Garetto2003LongTail}
M.~Garetto and D.~F. Towsley, ``Modeling, simulation and measurements of
  queuing delay under long-tail internet traffic,'' in \emph{Proceedings of ACM
  SIGMETRICS}, San Diego, CA, USA, 2003, pp. 47--57.

\bibitem{Liebeherr2009DelayTail}
J.~Liebeherr, A.~Burchard, and F.~Ciucu, ``Delay bounds in communication
  networks with heavy-tailed and self-similar traffic,'' \emph{IEEE
  Transactions on Information Theory}, vol.~58, no.~2, pp. 1010--1024, 2012.

\bibitem{sim2hw}
\BIBentryALTinterwordspacing
J.~Sp\"{a}th, M.~Helm, B.~Jaeger, and G.~Carle, ``Sim2hw: Modeling latency
  offset between network simulations and hardware measurements,'' in
  \emph{Proceedings of the 3rd GNNet Workshop on Graph Neural Networking
  Workshop}, ser. GNNet '24.\hskip 1em plus 0.5em minus 0.4em\relax New York,
  NY, USA: Association for Computing Machinery, 2024, p. 20–26. [Online].
  Available: \url{https://doi.org/10.1145/3694811.3697820}
\BIBentrySTDinterwordspacing

\bibitem{norris1998markov}
J.~R. Norris, \emph{Markov chains}.\hskip 1em plus 0.5em minus 0.4em\relax
  Cambridge university press, 1998, no.~2.

\bibitem{kleinrock1964}
L.~Kleinrock, \emph{Communication Nets: Stochastic Message Flow and
  Delay}.\hskip 1em plus 0.5em minus 0.4em\relax New York: McGraw-Hill, 1964.

\bibitem{sundial}
Y.~Li, G.~Kumar, H.~Hariharan, H.~Wassel, P.~Hochschild, D.~Platt, S.~Sabato,
  M.~Yu, N.~Dukkipati, P.~Chandra, and A.~Vahdat, ``Sundial: Fault-tolerant
  clock synchronization for datacenters,'' in \emph{14th {USENIX} Symposium on
  Operating Systems Design and Implementation ({OSDI} 20)}, 2020.

\bibitem{multipathPTP}
A.~Shpiner, Y.~Revah, and T.~Mizrahi, ``Multi-path time protocols,'' in
  \emph{2013 IEEE International Symposium on Precision Clock Synchronization
  for Measurement, Control and Communication (ISPCS) Proceedings}, 2013, pp.
  1--6.

\bibitem{PTPsec}
A.~Finkenzeller, O.~Butowski, E.~Regnath, M.~Hamad, and S.~Steinhorst,
  ``Ptpsec: Securing the precision time protocol against time delay attacks
  using cyclic path asymmetry analysis,'' in \emph{IEEE INFOCOM 2024 - IEEE
  Conference on Computer Communications}, 2024, pp. 461--470.

\bibitem{ZHUGE2026112151}
J.~Zhuge, Z.~Gao, Y.~Tian, J.~Wang, S.~Yun, and X.~Gong, ``{DSCC}: Dynamic
  synergistic congestion control for lossless {RDMA} datacenter networks,''
  \emph{Computer Networks}, vol. 280, p. 112151, 2026.

\bibitem{ziegler2001quantitative}
T.~Ziegler, C.~Brandauer, and S.~Fdida, ``A quantitative model for the
  parameter setting of red with tcp traffic,'' in \emph{International Workshop
  on Quality of Service}.\hskip 1em plus 0.5em minus 0.4em\relax Springer,
  2001, pp. 202--216.

\bibitem{RED_Selection}
C.~Hollot, V.~Misra, D.~Towsley, and W.-B. Gong, ``On designing improved
  controllers for aqm routers supporting tcp flows,'' in \emph{Proceedings IEEE
  INFOCOM 2001. Conference on Computer Communications. Twentieth Annual Joint
  Conference of the IEEE Computer and Communications Society (Cat.
  No.01CH37213)}, vol.~3, 2001, pp. 1726--1734 vol.3.

\bibitem{virtulaisedTCP}
\BIBentryALTinterwordspacing
B.~Cronkite-Ratcliff, A.~Bergman, S.~Vargaftik, M.~Ravi, N.~McKeown,
  I.~Abraham, and I.~Keslassy, ``Virtualized congestion control,'' in
  \emph{Proceedings of the 2016 ACM SIGCOMM Conference}, ser. SIGCOMM
  '16.\hskip 1em plus 0.5em minus 0.4em\relax New York, NY, USA: Association
  for Computing Machinery, 2016, p. 230–243. [Online]. Available:
  \url{https://doi.org/10.1145/2934872.2934889}
\BIBentrySTDinterwordspacing

\bibitem{cheng2024pet}
K.~Cheng, T.~Wang, X.~Du, S.~Du, and H.~Cai, ``Pet: Multi-agent independent
  ppo-based automatic ecn tuning for high-speed data center networks,''
  \emph{arXiv preprint arXiv:2405.11956}, 2024.

\bibitem{jung2023rlecn}
E.-S. Jung, H.~S. Kim \emph{et~al.}, ``Rlecn—a learning based dynamic
  threshold control of ecn,'' \emph{ICT Express}, vol.~9, no.~6, pp.
  1007--1012, 2023.

\bibitem{acc_paper}
S.~Yan, X.~Wang, X.~Zheng, Y.~Xia, D.~Liu, and W.~Deng, ``Acc: Automatic ecn
  tuning for high-speed datacenter networks,'' in \emph{Proceedings of the 2021
  ACM SIGCOMM 2021 Conference}, 2021, pp. 384--397.

\bibitem{heirachichaltuning}
J.~Hu, Y.~Wang, Z.~Zhou, S.~Rao, R.~Xin, J.~Wang, and S.~He, ``Haecn:
  Hierarchical automatic ecn tuning with ultra-low overhead in datacenter
  networks,'' in \emph{Algorithms and Architectures for Parallel Processing},
  Z.~Tari, K.~Li, and H.~Wu, Eds.\hskip 1em plus 0.5em minus 0.4em\relax
  Singapore: Springer Nature Singapore, 2024, pp. 324--343.

\bibitem{lanz}
{Arista Networks}, ``{Latency Analyzer (LANZ)},'' EOS User Manual,
  \url{https://www.arista.com/en/um-eos/eos-latency-analyzer-lanz}, accessed:
  2026-10-30.

\bibitem{rfc9197}
\BIBentryALTinterwordspacing
F.~Brockners, S.~Bhandari, and T.~Mizrahi, ``{Data Fields for In Situ
  Operations, Administration, and Maintenance (IOAM)},'' RFC 9197, May 2022.
  [Online]. Available: \url{https://www.rfc-editor.org/info/rfc9197}
\BIBentrySTDinterwordspacing

\bibitem{haestimtation2013}
Y.~Ha, E.~Pak, J.~Park, T.~Kim, and J.~W. Yoon, ``Clock offset estimation for
  systems with asymmetric packet delays,'' \emph{IEEE/ACM Transactions on
  Networking}, vol.~31, no.~4, pp. 1838--1853, 2023.

\end{thebibliography}

\clearpage
\BuildModeAppendix
\appendices

\section{Conditions for Improvement}
\label{sec:conditions_for_improvement}

\subsection{Condition for Variance Reduction}
\label{subsec:proof_variance_reduction}

We derive the general condition for variance reduction with $R$ thresholds.
We need to find the condition for which $\operatorname{Var}(\mathcal{Q}_{m,l}) > \operatorname{Var}(\mathcal{D}_{m,l})$. 
For ease of notation, we write $\mathcal{Q}_{m,l}$ as $X$, $\mathcal{D}_{m,l}$ for a given $R$ as $X_R$ and the threshold $\Delay{K}{*}$ as $x$.  
Hence, 
\[
X_R = X - x \sum_{i=1}^{R} \mathbf{1}_{\{X>ix\}} .
\]
Let, 
\[
\mu=\E[X], \quad p_{ix}=\P(X>ix) \quad \text{and} \quad m_{ix}=\E[X\mid X>ix].  
\]
Furthermore, note the identity,  
\begin{equation}
    \operatorname{Var}(A-B)=\operatorname{Var}(A)-2\operatorname{Cov}(A,B)+\operatorname{Var}(B)
    \label{eq:var_cov_identity}
\end{equation}
For a given $x$ and $R$, if $X_R$ exists, then so does $X_{R-1}$ and has the relation
\[
X_R = X_{R-1}-x\mathbf{1}_{\{X>Rx\}}.
\]
Using the identitiy from Equation~\ref{eq:var_cov_identity} we obtain,
\begin{equation}
\begin{aligned}
\operatorname{Var}(X_R) &= \operatorname{Var}(X_{R-1}) 
- 2x\,\operatorname{Cov}\bigl(X_{R-1}, \mathbf{1}_{\{X>Rx\}}\bigr) \\
&\quad + x^2 p_{Rx}(1-p_{Rx}),
\end{aligned}
\label{eq:expanded_variance}
\end{equation}
and therefore the necessary condition for the variance to be reduced, 
\begin{equation}
    2x\,\operatorname{Cov}\bigl(X_{R-1}, \mathbf{1}_{\{X>Rx\}}\bigr) \ge x^2 p_{Rx}(1-p_{Rx}). 
    \label{eq:necessary_condition}
\end{equation}
Now using
\[
X_{R-1} = X - x \sum_{i=1}^{R-1} \mathbf{1}_{\{X>ix\}},
\]
and by the linearity property of covariance,
\begin{equation*}
\operatorname{Cov}(A-B,C)
=
\operatorname{Cov}(A,C)
-
\operatorname{Cov}(B,C),
\end{equation*}
we simplify the LHS from Equation~\ref{eq:necessary_condition},
\begin{multline*}
\operatorname{Cov}\bigl(X_{R-1},\mathbf{1}_{\{X>Rx\}}\bigr)
=
\operatorname{Cov}\bigl(X,\mathbf{1}_{\{X>Rx\}}\bigr) \\
-
\operatorname{Cov}\Bigl(x\sum_{i=1}^{R-1}\mathbf{1}_{\{X>ix\}},\mathbf{1}_{\{X>Rx\}}\Bigr).
\end{multline*}
Factoring out $x$ from the second term yields
\begin{multline}
\operatorname{Cov}\bigl(X_{R-1},\mathbf{1}_{\{X>Rx\}}\bigr)
=
\operatorname{Cov}\bigl(X,\mathbf{1}_{\{X>Rx\}}\bigr) \\
-
x\sum_{i=1}^{R-1}
\operatorname{Cov}\bigl(\mathbf{1}_{\{X>ix\}},\mathbf{1}_{\{X>Rx\}}\bigr).
\label{eq:covariance_term}
\end{multline}
By the definition of covariance,
\begin{multline}
\operatorname{Cov}\bigl(X,\mathbf{1}_{\{X>Rx\}}\bigr)
=
\E\bigl[X\,\mathbf{1}_{\{X>Rx\}}\bigr] 
-
\E[X]\E\bigl[\mathbf{1}_{\{X>Rx\}}\bigr]\\
=
p_{Rx} m_{Rx}
-
\mu p_{Rx}
=
p_{Rx}(m_{Rx}-\mu), 
\label{eq:cov_x_1}
\end{multline}
and for $i<R$, 
\begin{multline}
\operatorname{Cov}\bigl(\mathbf{1}_{\{X>ix\}},\mathbf{1}_{\{X>Rx\}}\bigr)
=
p_{Rx}
-
p_{ix}p_{Rx}
=
p_{Rx}(1-p_{ix}).
\label{eq:cov_1_1}
\end{multline}
By plugging \ref{eq:cov_1_1} and \ref{eq:cov_x_1} in \ref{eq:covariance_term}, we obtain a simplified condition of Equation \ref{eq:necessary_condition}, 
\begin{equation}
    2(m_{Rx}-\mu)-2x\sum_{i=1}^{R-1}(1-p_{ix}) \ge x(1-p_{Rx})\, .\label{eq:simplified_necessary_condition}
\end{equation}

From monotonicity of any \ac{CCDF}, we have $p_{ix} \ge p_{Rx}$ for $i<R$ and hence
\[
\sum_{i=1}^{R-1} (1 - p_{ix}) \le (R-1)(1 - p_{Rx}).
\]

Plugging the RHS instead of the sum in the necessary condition in~\ref{eq:simplified_necessary_condition}, and from the fact that $m_{Rx}\ge Rx$ yields the sufficient condition \textbf{C1} for variance reduction:
\[
x \ge \frac{2\mu}{1 + (2R-1)\,p_{Rx}}.
\]
\subsection{Condition for Difference of Means Reduction}
For ease of notation, we denote the forward and reverse delay random variables by $A$ and $B$. Assume that $\E[A]>\E[B]$ without loss of generality.
Let,
\[
 p_{ix}^{A}=\P(A>ix) \quad \text{and} \quad p_{ix}^{B}=\P(B>ix),
\]
and let, $C=A-B$ and well as $C'=\sum_{i=1}^R x(\mathbf{1}\gk{A>ix}-\mathbf{1}\gk{B>ix})$.
Thus,
\begin{equation*}
    \E[\mathcal{D}_{\text{Req}}-\mathcal{D}_{\text{Resp}}]^{2}=\rk{\E[C]-\E[C']}^2\, .
\end{equation*}
Note that the function $c'\mapsto (c-c')^2$ decreases if and only if $c'\in (0,2c)$. This then translates to
\begin{equation*}
    0<\E\ek{C'}<2\E[C]\, ,
\end{equation*}
and hence we have a reduction in the mean term if and only if
\begin{equation}\label{eq:impro}
    0<\frac{x\sum_{i=1}^R \rk{p_{ix}^{A}-p_{ix}^{B}}}{\E\ek{A}-\E\ek{B}}<2\, 
\end{equation}
by writing $\E\ek{C'}$ in the numerator.

The term is greater than 0 if \textbf{C2} is satisfied, while \textbf{C3} takes care of the upper bound on the RHS. We remark that under tail dominance, we can bound $p_{ix}^{A}-p_{ix}^{B}$ using Markov's inequality in terms of $x$ and hence get a tighter lower bound on $x$ for Eq.~\eqref{eq:impro} to be satisfied. We leave this to the reader. %
\section{Improvement Results}
\label{sec:app_improvement_results}
\subsection{Improvement with R}

Every addition of a new threshold step $R$ will have a variance reduction defined as:
\begin{equation*}
    \Delta_{R} = \operatorname{Var}(X_{R-1})-\operatorname{Var}(X_{R}). 
\end{equation*}
Using Equation~\ref{eq:expanded_variance}, we get,
\begin{equation*}
    \Delta_{R} = 2x\operatorname{Cov}\bigl(X_{R-1}, \mathbf{1}_{\{X>Rx\}}\bigr) - x^2 p_{Rx}(1-p_{Rx}). 
\end{equation*}
If the sufficient condition \textbf{C1} holds, then every $\Delta_{R} \ge 0$. 
Which implies a non-increase with $R$, i.e., $\operatorname{Var}(X_{R})\le\operatorname{Var}(X_{R-1})\le...\le\operatorname{Var}(X)$. 
\ROne{Now examining the proof of \textbf{C1} at the end of Sec.~\ref{subsec:proof_variance_reduction}, equality holds if and only if both $p_{ix}=p_{Rx}$ and $m_{Rx}=Rx$, that is no probability mass exists between the interval $(ix,Rx]$ for $i<R$, and all probability mass above $Rx$ is concentrated at $Rx$.} 
This is impossible for any queuing distribution, as they are all non-degenerate outside the origin. 
\ROne{Thus, with loss of equality, the non-increasing becomes strictly decreasing with $R$,} proving the result of Corollary 1. 
A special case in which equality holds occurs when there is no mass above the threshold $ix$; however, we consider this trivial, since the variance can only be reduced if the threshold is placed somewhere within the distribution. 

\subsection{Multihop Marking}
\label{sec:app_multihop_general}
Let the forward and reverse delays be defined as a sum of $s$ \ROne{individual delays,} 
\[
A = \sum_{i=1}^{s} A_i,
\qquad
B = \sum_{i=1}^{s} B_i\, .
\]
\ROne{For ease of notation, we write $A_i = \mathcal{Q}_{\text{Req},i}$ and
$B_i = \mathcal{Q}_{\text{Resp},i}$ in this appendix, and allow a hop-specific threshold $x_i$; in the main text, $x_i = \Delay{K}{*}$ for all hops.}
Under \textbf{A1}, we have that
\begin{equation}
    \mathrm{MSE}=\frac{1}{4}(\mathrm{Var(A)}+\mathrm{Var}(B)+\E[A-B]^2)
\end{equation}
Write $A_i'$ (resp. $B_i'$) for the modified delays, i.e., after applying congestion marking (if no marking at hop $i$ is applied, then $A_i'=A_i$, analogous for $B_i$). Using the same strategy as in \eqref{eq:impro}, we find that if
\begin{equation}
    0 \le \frac{\sum_{i=1}^{s} \E\big[(A_i - A'_i) - (B_i - B'_i)\big]}{\E[A - B]} \le 2 ,
\end{equation}
then the MSE does not increase. Replacing $\le$ by $<$ implies an improvement of the bias term. 
\ROne{If hop $i$ marks with $R_i$ thresholds of size $x_i$, its correction
of the bias is}
\begin{equation}
c_i = \sum_{k=1}^{R_i} x_i \big[\P(A_i > kx_i) - \P(B_i > kx_i)\big],
\label{eq:biastermImproves}
\end{equation}
\ROne{with $c_i = 0$ if hop $i$ does not mark, so that the condition above
becomes $0 \le \sum_{i=1}^{s} c_i \le 2\,\E[A - B]$, which is \textbf{C4}. Defining the
margin of hop $i$ as $g_i = 2\,\E[A_i - B_i] - c_i$, \textbf{C4} is equivalent to
$\sum_i c_i \ge 0$ and $\sum_i g_i \ge 0$. We call hop $i$ \emph{good} if
$c_i \ge 0$ and $g_i \ge 0$. If all hops are good, \textbf{C4} holds. For a marking hop with
$\E[A_i] > \E[B_i]$, \textbf{C2} yields $c_i > 0$ and \textbf{C3} yields $g_i > 0$, so such a
hop is good. A hop with $\E[A_i] < \E[B_i]$ cannot satisfy both $c_i \ge 0$ and $g_i \ge 0$, since this would require $0 \le c_i \le 2\,\E[A_i - B_i] < 0$, so it cannot be
good; its contribution must then be compensated by good hops, which is
possible since \textbf{C4} only constrains the sums.}

Next we expand the variance part
\begin{equation}
    \mathrm{Var}(A)=\sum_{i=1}^s\mathrm{Var}(A_i)+\sum_{i\neq j}^{s}\mathrm{Cov}(A_i,A_{j})\, .
\end{equation}
To improve the $\mathrm{Var}(A_i)$ term, we refer the reader to the previous section. To examine under which conditions the covariance term improves, we write $\Delta_i=A_i-A_i'$. Then, the covariance term improves if and only if
\begin{equation}
    \sum_{i\neq j}^{s}\mathrm{Cov}(A_i,A_{j})>\sum_{i\neq j}^{s}\mathrm{Cov}(A_i',A_{j}')\, .
\end{equation}
Expanding the covariances leads to the condition that
\begin{equation}
    2\sum_{i\neq j}^s\mathrm{Cov}(A_i,\Delta_j)>\sum_{i\neq j}^s\mathrm{Cov}(\Delta_i,\Delta_j)\, .
\end{equation}
If no marking is applied at hop $i$, we set $R_i=0$, so that the formulas below are always consistent.
Note that
\begin{equation}
    \mathrm{Cov}(A_i,\Delta_j)=x_j\sum_{k=1}^{R_{j}}p^{(j)}_{k}\big(\E[A_i|A_{j}>kx_{j}]-\E[A_i]\big)\, ,
\end{equation}
where $p^{(j)}_{k}=\P(A_{j}>kx_{j})$. Furthermore,
\begin{multline}
   \mathrm{Cov}(\Delta_i,\Delta_j)=x_j\sum_{k=1}^{R_j}p_k^{(j)}\\\sum_{l=1}^{R_i}x_i\big(\P(A_i>x_il|A_j>x_jk)-p_l^{(i)}\big)\, .
\end{multline}
Hence, the variance term improves \textit{if and only if}
\begin{multline}\label{eq:condition_bias}
     2\sum_{i\neq j}^sx_j\sum_{k=1}^{R_{j}}p^{(j)}_{k}\big(\E[A_i|A_{j}>kx_{j}]-\E[A_i]\big)>\\
     \sum_{i\neq j}^s x_j\sum_{k=1}^{R_j}p_k^{(j)}\sum_{l=1}^{R_i}x_i\big(\P(A_i>x_il|A_j>x_jk)-p_l^{(i)}\big)\, .
\end{multline}
Of course, if for each pair $(i,j)$
\begin{multline}\label{eq:pairwise_bias}
    2\sum_{k=1}^{R_{j}}p^{(j)}_{k}\big(\E[A_i|A_{j}>kx_{j}]-\E[A_i]\big)>\\
    \sum_{k=1}^{R_j}\sum_{l=1}^{R_i}p^{(j)}_{k}x_i\big(\P(A_i>x_il|A_j>x_jk)-p_l^{(i)}\big)\, ,
\end{multline}
\ROne{then overall the variance improves.} 

\ROne{Condition~\eqref{eq:pairwise_bias} is only sufficient: since only
the sum over all pairs in~\eqref{eq:condition_bias} matters, pairs of
hops that violate~\eqref{eq:pairwise_bias} can be compensated by others.
Condition~\eqref{eq:condition_bias}, in contrast, is necessary and
sufficient for the covariance term to decrease, since we have made no
assumptions about the network other than that the covariances are
finite. With $\ge$ in place of $>$, the covariance term does not
increase, which is \textbf{C5}.}

Finally, networks where marking decreases the mean but reduces the variance by a larger amount also have a lower MSE. However, stating this as a condition simply reduces to $\mathrm{MSE}'<\mathrm{MSE}$, which does not provide any insight into the conditions needed to be met at the switches and the correlations between them.

To verify \eqref{eq:condition_bias}, it sufficies to use the joint distribution function
\begin{equation}
    F_{i,j}^A(s,t)=\P(A_i> s, A_j> t)\, ,
\end{equation}
which can be uniformly estimated via its empirical counterpart
\begin{equation}
    F_{n;i,j}^A(s,t)=\frac{1}{n}\sum_{p=1}^n \1\{A_i^{p}> s,\, A_j^{p}> t\}\, ,
\end{equation}
\ROne{where $A^p_i$ and $A^p_j$ are the queuing delays of the same packet
$p$ at hops $i$ and $j$.} Similarly, we write $F^A_i$ for the distribution function of $A_i$

\ROne{Indeed, to verify \textbf{C4}, one needs to check that}
\begin{multline}
    0\le\sum_{i=1}^s\sum_{k=1}^{R_i} x_i[F^{A}_i(kx_i)-F^{B}_i(kx_i)]\\ \le 2\sum_{i=1}^s\big(\mathrm{Int}[F^{A}_i]-\mathrm{Int}[F^{B}_i]\big)\, ,
\end{multline}
where $\mathrm{Int}$ is numerical integration. To verify \eqref{eq:condition_bias}, recall that
\begin{equation}
    \begin{split}
        \P(A_i>x_il|A_j>x_jk)&= \frac{ F_{i,j}^A(x_il,x_jk)}{F_j^A(x_jk)}\, ,\\
        \E[A_i|A_{j}>kx_{j}]&=\frac{\mathrm{Int}[F_{i,j}^A(\, \cdot\, ,x_jk)]}{F_j^A(x_jk)} .
    \end{split}
\end{equation}
Verifying the two equations as well as the conditions of Proposition~\ref{prop:conditions} \ROne{lets one conclude that the MSE indeed reduces.}

\section{Efficient Implementation of the Counter Propagation Algorithm}
\label{sec:app_algorithm_implementation}

This Section discusses the practical implementation of the counter propagation Algorithm~\ref{alg1}, its computational complexity, and the structural properties of the resulting stochastic process.

\subsection{Markov Chain Representation}

The evolution of the correction counter along a path can be represented as a \textbf{finite-state time-inhomogeneous Markov chain}. \ROne{Let $T_l$ denote the value of the counter after processing the $l$-th hop} along the path. The counter takes values in the finite state space

\[
\mathcal{S} = \gk{0, 1, \ldots, N},
\]

\ROne{where $N$ denotes the maximum counter value.}

\ROne{The transition from state $T_{l-1}$ to $T_l$ depends on the delay random variable $\mathcal{Q}_{m,l}$ at hop $l$} and the threshold mapping defined by the algorithm. \ROne{Since the update rule depends only on the current counter value and the delay observed at the current hop, the process satisfies the Markov property under \textbf{A3}.} However, because the delay distribution may vary across hops, the transition probabilities \ROne{depend on the hop index $l$.} The resulting process is therefore a \textbf{time-inhomogeneous Markov chain}.

Let

\[
P_l(i,j) = \Pr(T_l = j \mid T_{l-1} = i)
\]

\ROne{denote the transition probability from state $i$ to state $j$ at hop $l$. The distribution of the counter after hop $l$ is} described by the probability vector

\[
\boldsymbol{\pi}_l = [\Pr(T_l=0),\Pr(T_l=1),\dots,\Pr(T_l=N)].
\]

The process is initialized as

\[
\boldsymbol{\pi}_0 = [1,0,\dots,0],
\]

reflecting that the counter initially equals zero. \ROne{The distribution after hop $l$ is then} obtained via the forward recursion

\[
\boldsymbol{\pi}_l = \boldsymbol{\pi}_{l-1} P_l .
\]

\ROne{After traversing the entire path consisting of $L = |\mathcal{L}_m|$ hops,} the final distribution becomes

\[
\boldsymbol{\pi}_L = \boldsymbol{\pi}_0 P_1 P_2 \cdots P_L,
\]

\ROne{which directly yields the probability distribution of the final counter value $T_L$.}
\ROne{The same applies to the distribution $\mathcal{D}_m$,} albeit with a much larger state space, as the per-hop delay distributions are the mixture distributions explained in Sec.~\ref{sec:analysis}. 

\subsection{Efficient Computation}

A direct implementation of Algorithm 1 corresponds to constructing a tree of all possible counter-evolutions along the path. In the worst case, \ROne{each state may generate up to $R+1$ successor states, where $R$ denotes} the number of thresholds. Consequently, the number of possible paths grows exponentially with the number of hops.

However, many of these paths correspond to identical counter states. Since future transitions depend only on the current counter value, all paths that reach the same state can be merged without loss of accuracy. The forward propagation described above, therefore, provides an equivalent dynamic programming implementation that tracks only the probability mass associated with each state.

In practice, the transition probabilities depend only on the probability that the delay falls within the algorithm-defined threshold intervals. For hop ($l$), the required quantities are

\[
p_{l,r} = \Pr\big(r\Delay{K}{*} < \mathcal{Q}_{m,l} \le (r+1)\Delay{K}{*}\big), \quad r = 0, \ldots, R-1,
\]
\ROne{where for $r = 0$ the interval includes zero, and $p_{l,R} = \Pr(\mathcal{Q}_{m,l} > R\Delay{K}{*})$ for the open-ended last interval. These}
\ROne{can be computed once by \textbf{discretizing} the empirical delay samples in $R+1$ bins.} 
Once these probabilities of being in each bin are available, each propagation step reduces to updating the probability vector according to the transition structure.

\ROne{The transition matrices $P_l$ exhibit} a sparse structure. \ROne{For any state $i$, the algorithm maps the delay observation to one of at most $R+1$ possible counter updates. Consequently, each row of $P_l$ contains at most $R+1$ non-zero entries.}

\ROne{Let $\text{nnz}(P_l)$ denote the number of non-zero elements in $P_l$.} The sparsity property implies

\[
\text{nnz}(P_l) \le (N+1)(R+1).
\]

\subsection{Absorbing and Boundary States}

\ROne{The counter is bounded between $0$ and $N$. In particular, once the counter reaches the maximum value $N$,} further increments are not possible. As a result, \ROne{state $N$} behaves as an \textbf{absorbing state} for transitions that would otherwise increase the counter beyond this bound. Similarly, the lower bound at zero prevents negative counter values.

These boundary conditions guarantee that the state space remains finite and that probability mass cannot \ROne{leave the set $\mathcal{S}$.} In practice, the presence of the upper absorbing boundary also reduces the effective branching of the process, since probability mass tends to accumulate near this boundary as the number of hops increases.

\subsection{Computational Complexity}
The forward propagation algorithm maintains a probability vector of size ($N+1$) and updates it sequentially for each hop. Because each state has at most ($R+1$) possible successor states, the computational complexity of the propagation step is

\[
O(LNR).
\]

The memory requirement is ($O(N)$), since only the current probability vector must be stored.

In addition, a preprocessing step is required to compute the discretized interval probabilities from the empirical delay distributions. If each empirical distribution contains ($S$) samples, sorting the samples requires ($O(S\log S)$) time. After sorting, the interval probabilities can be computed efficiently using cumulative counts.

For the parameter ranges considered in this work (e.g., ($N \leq 32$) and ($L \leq 7$)), the resulting computation is lightweight and enables repeated evaluations of the algorithm during threshold parameter optimization.
\section{Table of Symbols and Notation}

\begin{table}[h]
\centering
\caption{Core Notation}
\label{tab:core_notation}
\begin{tabular}{ll}
\toprule
\textbf{Symbol} & \textbf{Description} \\
\midrule
$T_1, T_2, T_3, T_4$ & Timestamps in synchronization message exchange \\
$\theta$ & True clock offset between client and server \\
$\hat{\theta}$ & Estimated clock offset \\
$\epsilon$ & Offset estimation error ($\theta - \hat{\theta}$) \\
$\delta_{\text{Req}}, \delta_{\text{Resp}}$ & End-to-end delays of request and response messages \\
$\delta^{\text{base}}_m$ & Base delay (propagation + transmission) of message $m$ \\
$\delta^q_m$ & Total queuing delay of message $m$ \\
$\hat{\delta}^q_m$ & Estimated queuing delay \\
$\delta^*_K$ & Delay corresponding to threshold $K$ \\
$K$ & Queue length threshold for congestion marking \\
$R$ & Number of thresholds per hop \\
$N$ & Maximum marking counter value \\
$n$ & Observed marking counter value \\
$m \in \{\text{Req}, \text{Resp}\}$ & Message type \\
$L_m$ & Set of hops traversed by message $m$ \\
$\delta^q_{m,l}$ & Queuing delay at hop $l$ \\
$LR$ & Link line rate (bits per second) \\
\bottomrule
\end{tabular}
\end{table}

\begin{table}[h]
\centering
\caption{Analysis Notation}
\label{tab:analysis_notation}
\begin{tabular}{ll}
\toprule
\textbf{Symbol} & \textbf{Description} \\
\midrule
$\delta^E_m$ & Queuing delay estimation error of message $m$ \\
$\delta^E_{m,l}$ & Delay estimation error at hop $l$ \\
$Q_{m,l}$ & Queuing delay at hop $l$ (random variable) \\
$D_{m,l}$ & Delay estimation error at hop $l$ (random variable) \\
$Q_m$ & End-to-end queuing delay (random variable)\\
$D_m$ & End-to-end delay estimation error (random variable) \\
$\rho_l$ & Utilization of link $l$ \\
$\mathbb{E}[\cdot]$ & Expectation operator \\
$\mathrm{Var}(\cdot)$ & Variance operator \\
$\mathcal{E}$ & Expected MSE without correction \\
$\mathcal{E}_{\text{comp}}$ & Expected MSE after congestion marking \\
$T_{m,l}$ & Marking counter state at hop $l$ \\
$\chi(\cdot)$ & Indicator function \\
$r_{m,l}$
$r$ & Number of thresholds used at a hop \\
\bottomrule
\end{tabular}
\end{table} 
\end{document}